\documentclass{amsart}

\usepackage{amssymb,gensymb,comment}
\usepackage[a4paper,left=2.25cm,right=2.25cm,top=3cm,bottom=3cm]{geometry}
\usepackage{tikz}\usetikzlibrary{decorations.pathmorphing,decorations.markings,backgrounds}
\usepackage{hyperref}
\usepackage{enumerate}
\usepackage{mathdots}
\usepackage{graphicx}
\usepackage{bm}
\usepackage{xargs}

\allowdisplaybreaks[1]

\def\slt{-0.5}
\pgfmathsetmacro{\ae}{atan(\slt)}
\pgfmathsetmacro{\aw}{\ae+180}
\pgfmathsetmacro{\an}{90-\ae}
\pgfmathsetmacro{\as}{\an+180}
\pgfmathsetmacro{\sltb}{sqrt(1-\slt*\slt)}
\pgfmathsetmacro{\lcrot}{45-atan(\slt/\sltb)*0.5}
\tikzset{distort/.style={cm={1,0,-\slt,\sltb,(0,0)}}}

\tikzset{bgplaq/.style={fill=lightgray!20!white}}
\tikzset{dbgplaq/.style={fill=lightgray}}
\tikzset{arrow/.style={postaction={decorate,thick,decoration={markings,mark = at position #1 with {\arrow{>}}}}},arrow/.default=0.5}
\tikzset{fline/.style={line width=0.2mm}}
\tikzset{wline/.style={line width=0.2mm,draw=white}}
\tikzset{bline/.style={line width=0.5mm}}
\tikzset{rline/.style={line width=0.5mm,draw=black!20!red}}
\tikzset{gline/.style={line width=0.5mm,draw=black!20!green}}
\tikzset{bbline/.style={line width=0.5mm,draw=white!20!blue}}

\newcommand\wlight[2]{\draw[dotted,bgplaq] (#1,#2) -- (#1+1,#2) -- (#1+1,#2+1) -- (#1,#2+1) -- (#1,#2);}
\newcommand\wwlight[6]{\draw[dotted,bgplaq] (#1,#2) -- node{${\bm #3}$} (#1+1,#2) -- node{${\bm #4}$} (#1+1,#2+1) -- node{${\bm #5}$} (#1,#2+1) -- node{${\bm #6}$} (#1,#2);}
\newcommand\blight[2]{\draw[fline,bgplaq] (#1,#2) -- (#1+1,#2) -- (#1+1,#2+1) -- (#1,#2+1) -- (#1,#2);}
\newcommand\bdark[2]{\draw[fline,dbgplaq] (#1,#2) -- (#1+1,#2) -- (#1+1,#2+1) -- (#1,#2+1) -- (#1,#2);}
\newcommand\bbull[3]{\filldraw[fill=black, draw=black] (#1,#2) circle (#3cm);}
\newcommand\rbull[3]{\filldraw[fill=black!20!red, draw=black!20!red] (#1,#2) circle (#3cm);}
\newcommand\gbull[3]{\filldraw[fill=black!20!green,draw=black!20!green] (#1,#2) circle (#3cm);}
\newcommand\ebull[3]{\draw[fill=white,draw=black] (#1,#2) circle (#3cm);}
\newcommandx\bpart[1][1=0.05]{\begin{tikzpicture}[scale=0.5,baseline=-2.5] \bbull{0}{0}{#1}; \end{tikzpicture}}
\newcommandx\rpart[1][1=0.05]{\begin{tikzpicture}[scale=0.5,baseline=-2.5] \rbull{0}{0}{#1}; \end{tikzpicture}}
\newcommandx\gpart[1][1=0.05]{\begin{tikzpicture}[scale=0.5,baseline=-2.5] \gbull{0}{0}{#1}; \end{tikzpicture}}
\newcommandx\hole[1][1=0.05]{\begin{tikzpicture}[scale=0.5,baseline=-2.5] \ebull{0}{0}{#1}; \end{tikzpicture}}

\def\Te{T_{\begin{tikzpicture}[scale=0.5] \ebull{0}{0}{0.1} \end{tikzpicture}}}
\def\Tb{T_{\begin{tikzpicture}[scale=0.5] \bbull{0}{0}{0.1} \end{tikzpicture}}}
\def\Tr{T_{\begin{tikzpicture}[scale=0.5] \rbull{0}{0}{0.1} \end{tikzpicture}}}
\def\Tg{T_{\begin{tikzpicture}[scale=0.5] \gbull{0}{0}{0.1} \end{tikzpicture}}}
\def\Tes{T^{*}_{\begin{tikzpicture}[scale=0.5] \ebull{0}{0}{0.1} \end{tikzpicture}}}
\def\Tbs{T^{*}_{\begin{tikzpicture}[scale=0.5] \bbull{0}{0}{0.1} \end{tikzpicture}}}
\def\Trs{T^{*}_{\begin{tikzpicture}[scale=0.5] \rbull{0}{0}{0.1} \end{tikzpicture}}}

\newcommand{\greenoc}[3]{\filldraw[gline,fill=black!20!green] (#1,#2) circle (0.12cm);
\node[text centered] at (#1,#2) {\color{white} \tiny $#3$};}
\newcommand{\redoc}[2]{\filldraw[rline,fill=black!20!red] (#1,#2) circle (0.12cm);}
\newcommand{\blackoc}[3]{\filldraw[bline,fill=black] (#1,#2) circle (0.12cm);
\node[text centered] at (#1,#2) {\color{white} \tiny $#3$};}

\newcommand{\gdata}[4]{
\ifnum#1=0
\else 
\greenoc{#2+0.5}{#3+1}{#1};
\fi
\ifnum#4=0
\else 
\greenoc{#2+0.5}{#3}{#4};
\fi
}

\newcommand{\rdata}[4]{
\ifnum#1=0
\else 
\redoc{#2+0.5}{#3+1};
\fi
\ifnum#4=0
\else 
\redoc{#2+0.5}{#3};
\fi
}

\newcommand{\bdata}[4]{
\ifnum#1=0
\else 
\blackoc{#2+0.5}{#3+1}{#1};
\fi
\ifnum#4=0
\else 
\blackoc{#2+0.5}{#3}{#4};
\fi
}

\newcommand{\gbdata}[6]{
\ifnum#1=0
\else
\greenoc{#3+0.25}{#4+1}{#1};
\fi
\ifnum#2=0
\else 
\blackoc{#3+0.75}{#4+1}{#2};
\fi
\ifnum#5=0
\else
\greenoc{#3+0.25}{#4}{#5};
\fi
\ifnum#6=0
\else
\blackoc{#3+0.75}{#4}{#6};
\fi
}

\newcommand{\rbdata}[6]{
\ifnum#1=0
\else
\redoc{#3+0.25}{#4+1};
\fi
\ifnum#2=0
\else 
\blackoc{#3+0.75}{#4+1}{#2};
\fi
\ifnum#5=0
\else
\redoc{#3+0.25}{#4};
\fi
\ifnum#6=0
\else
\blackoc{#3+0.75}{#4}{#6};
\fi
}

\newcommandx{\gblank}[4][1=0,4=0]{
\begin{scope}[on background layer]
\blight{#2}{#3};
\end{scope}
\gdata{#1}{#2}{#3}{#4}
}
\newcommandx{\gplus}[4][1=0,4=0]{
\begin{scope}[on background layer]
\blight{#2}{#3};
\draw[gline] (#2+1,#3+0.5) -- (#2+0.5,#3+0.5) -- (#2+0.5,#3+0);
\end{scope}
\gdata{#1}{#2}{#3}{#4}
}
\newcommandx{\gminus}[4][1=0,4=0]{
\begin{scope}[on background layer]
\blight{#2}{#3};
\draw[gline] (#2+0.5,#3+1) -- (#2+0.5,#3+0.5) -- (#2,#3+0.5);
\end{scope}
\gdata{#1}{#2}{#3}{#4}
}
\newcommandx{\gfull}[4][1=0,4=0]{
\begin{scope}[on background layer]
\blight{#2}{#3};
\draw[gline] (#2,#3+0.5) -- (#2+1,#3+0.5);
\end{scope}
\gdata{#1}{#2}{#3}{#4}
}
\newcommandx{\darkgfull}[4][1=0,4=0]{
\begin{scope}[on background layer]
\bdark{#2}{#3};
\draw[gline] (#2,#3+0.5) -- (#2+1,#3+0.5);
\end{scope}
\bdata{#1}{#2}{#3}{#4}
}

\newcommandx{\rblank}[4][1=0,4=0]{
\begin{scope}[on background layer]
\blight{#2}{#3};
\end{scope}
\rdata{#1}{#2}{#3}{#4}
}
\newcommandx{\rplus}[4][1=0,4=0]{
\begin{scope}[on background layer]
\blight{#2}{#3};
\draw[rline] (#2+1,#3+0.5) -- (#2+0.5,#3+0.5) -- (#2+0.5,#3+0);
\end{scope}
\rdata{#1}{#2}{#3}{#4}
}
\newcommandx{\rminus}[4][1=0,4=0]{
\begin{scope}[on background layer]
\blight{#2}{#3};
\draw[rline] (#2+0.5,#3+1) -- (#2+0.5,#3+0.5) -- (#2,#3+0.5);
\end{scope}
\rdata{#1}{#2}{#3}{#4}
}
\newcommandx{\rfull}[4][1=0,4=0]{
\begin{scope}[on background layer]
\blight{#2}{#3};
\draw[rline] (#2,#3+0.5) -- (#2+1,#3+0.5);
\end{scope}
\rdata{#1}{#2}{#3}{#4}
}
\newcommandx{\darkrfull}[4][1=0,4=0]{
\begin{scope}[on background layer]
\bdark{#2}{#3};
\draw[rline] (#2,#3+0.5) -- (#2+1,#3+0.5);
\end{scope}
\bdata{#1}{#2}{#3}{#4}
}
\newcommandx{\bblank}[4][1=0,4=0]{
\begin{scope}[on background layer]
\bdark{#2}{#3};
\end{scope}
\bdata{#1}{#2}{#3}{#4}
}
\newcommandx{\bplus}[4][1=0,4=0]{
\begin{scope}[on background layer]
\bdark{#2}{#3};
\draw[bline] (#2+1,#3+0.5) -- (#2+0.5,#3+0.5) -- (#2+0.5,#3+0);
\end{scope}
\bdata{#1}{#2}{#3}{#4}
}
\newcommandx{\bminus}[4][1=0,4=0]{
\begin{scope}[on background layer]
\bdark{#2}{#3};
\draw[bline] (#2+0.5,#3+1) -- (#2+0.5,#3+0.5) -- (#2,#3+0.5);
\end{scope}
\bdata{#1}{#2}{#3}{#4}
}
\newcommandx{\bfull}[4][1=0,4=0]{
\begin{scope}[on background layer]
\bdark{#2}{#3};
\draw[bline] (#2,#3+0.5) -- (#2+1,#3+0.5);
\end{scope}
\bdata{#1}{#2}{#3}{#4}
}
\newcommandx{\lightbfull}[4][1=0,4=0]{
\begin{scope}[on background layer]
\blight{#2}{#3};
\draw[bline] (#2,#3+0.5) -- (#2+1,#3+0.5);
\end{scope}
\gdata{#1}{#2}{#3}{#4}
}
\newcommandx{\lightbfullr}[4][1=0,4=0]{
\begin{scope}[on background layer]
\blight{#2}{#3};
\draw[bline] (#2,#3+0.5) -- (#2+1,#3+0.5);
\end{scope}
\rdata{#1}{#2}{#3}{#4}
}

\newcommandx\rbblank[6][1=0,2=0,5=0,6=0]{
\begin{scope}[on background layer]
\blight{#3}{#4};
\end{scope}
\rbdata{#1}{#2}{#3}{#4}{#5}{#6}
}
\newcommandx\rplusb[6][1=0,2=0,5=0,6=0]{
\begin{scope}[on background layer]
\blight{#3}{#4};
\draw[rline] (#3+1,#4+0.5) -- (#3+0.25,#4+0.5) -- (#3+0.25,#4+0);
\end{scope}
\rbdata{#1}{#2}{#3}{#4}{#5}{#6}
}
\newcommandx\bplusr[6][1=0,2=0,5=0,6=0]{
\begin{scope}[on background layer]
\blight{#3}{#4};
\draw[bline] (#3+1,#4+0.5) -- (#3+0.75,#4+0.5) -- (#3+0.75,#4+0);
\end{scope}
\rbdata{#1}{#2}{#3}{#4}{#5}{#6}
}
\newcommandx\rminusb[6][1=0,2=0,5=0,6=0]{
\begin{scope}[on background layer]
\blight{#3}{#4};
\draw[rline] (#3+0.25,#4+1) -- (#3+0.25,#4+0.5) -- (#3,#4+0.5);
\end{scope}
\rbdata{#1}{#2}{#3}{#4}{#5}{#6}
}
\newcommandx\bminusr[6][1=0,2=0,5=0,6=0]{
\begin{scope}[on background layer]
\blight{#3}{#4};
\draw[bline] (#3+0.75,#4+1) -- (#3+0.75,#4+0.5) -- (#3,#4+0.5);
\end{scope}
\rbdata{#1}{#2}{#3}{#4}{#5}{#6}
}
\newcommandx\rbplusmin[6][1=0,2=0,5=0,6=0]{
\begin{scope}[on background layer]
\blight{#3}{#4};
\draw[rline] (#3,#4+0.5) -- (#3+0.25,#4+0.5) -- (#3+0.25,#4+1);
\draw[bline] (#3+0.75,#4) -- (#3+0.75,#4+0.5) -- (#3+1,#4+0.5);
\end{scope}
\rbdata{#1}{#2}{#3}{#4}{#5}{#6}
}
\newcommandx\rfullb[6][1=0,2=0,5=0,6=0]{
\begin{scope}[on background layer]
\blight{#3}{#4};
\draw[rline] (#3,#4+0.5) -- (#3+1,#4+0.5);
\end{scope}
\rbdata{#1}{#2}{#3}{#4}{#5}{#6}
}
\newcommandx\bfullr[6][1=0,2=0,5=0,6=0]{
\begin{scope}[on background layer]
\blight{#3}{#4};
\draw[bline] (#3,#4+0.5) -- (#3+1,#4+0.5);
\end{scope}
\rbdata{#1}{#2}{#3}{#4}{#5}{#6}
}

\newcommandx\gbblank[6][1=0,2=0,5=0,6=0]{
\begin{scope}[on background layer]
\blight{#3}{#4};
\end{scope}
\gbdata{#1}{#2}{#3}{#4}{#5}{#6}
}
\newcommandx\gplusb[6][1=0,2=0,5=0,6=0]{
\begin{scope}[on background layer]
\blight{#3}{#4};
\draw[gline] (#3+1,#4+0.5) -- (#3+0.25,#4+0.5) -- (#3+0.25,#4+0);
\end{scope}
\gbdata{#1}{#2}{#3}{#4}{#5}{#6}
}
\newcommandx\bplusg[6][1=0,2=0,5=0,6=0]{
\begin{scope}[on background layer]
\blight{#3}{#4};
\draw[bline] (#3+1,#4+0.5) -- (#3+0.75,#4+0.5) -- (#3+0.75,#4+0);
\end{scope}
\gbdata{#1}{#2}{#3}{#4}{#5}{#6}
}
\newcommandx\gminusb[6][1=0,2=0,5=0,6=0]{
\begin{scope}[on background layer]
\blight{#3}{#4};
\draw[gline] (#3+0.25,#4+1) -- (#3+0.25,#4+0.5) -- (#3,#4+0.5);
\end{scope}
\gbdata{#1}{#2}{#3}{#4}{#5}{#6}
}
\newcommandx\bminusg[6][1=0,2=0,5=0,6=0]{
\begin{scope}[on background layer]
\blight{#3}{#4};
\draw[bline] (#3+0.75,#4+1) -- (#3+0.75,#4+0.5) -- (#3,#4+0.5);
\end{scope}
\gbdata{#1}{#2}{#3}{#4}{#5}{#6}
}
\newcommandx\gbplusmin[6][1=0,2=0,5=0,6=0]{
\begin{scope}[on background layer]
\blight{#3}{#4};
\draw[gline] (#3,#4+0.5) -- (#3+0.25,#4+0.5) -- (#3+0.25,#4+1);
\draw[bline] (#3+0.75,#4) -- (#3+0.75,#4+0.5) -- (#3+1,#4+0.5);
\end{scope}
\gbdata{#1}{#2}{#3}{#4}{#5}{#6}
}
\newcommandx\gfullb[6][1=0,2=0,5=0,6=0]{
\begin{scope}[on background layer]
\blight{#3}{#4};
\draw[gline] (#3,#4+0.5) -- (#3+1,#4+0.5);
\end{scope}
\gbdata{#1}{#2}{#3}{#4}{#5}{#6}
}
\newcommandx\bfullg[6][1=0,2=0,5=0,6=0]{
\begin{scope}[on background layer]
\blight{#3}{#4};
\draw[bline] (#3,#4+0.5) -- (#3+1,#4+0.5);
\end{scope}
\gbdata{#1}{#2}{#3}{#4}{#5}{#6}
}

\newcommand{\plusg}[2]{
\blight{#1}{#2};
\filldraw[fill=black!20!green,draw=black!20!green] (#1+0.5,#2+0.5) circle (0.25cm);
\draw[gline] (#1+0.5,#2+0.5) -- (#1+1,#2+0.5);
\node[text centered] at (#1+0.5,#2+0.5) {\color{white} $+$};
}

\newcommand{\ming}[2]{
\blight{#1}{#2};
\filldraw[fill=black!20!green,draw=black!20!green] (#1+0.5,#2+0.5) circle (0.25cm);
\draw[gline] (#1,#2+0.5) -- (#1+0.5,#2+0.5);
\node[text centered] at (#1+0.5,#2+0.5) {\color{white} $-$};
}

\newcommand{\flatg}[2]{
\blight{#1}{#2};
\draw[gline] (#1,#2+0.5) -- (#1+1,#2+0.5);
}

\newcommand{\flatG}[2]{
\bdark{#1}{#2};
\draw[gline] (#1,#2+0.5) -- (#1+1,#2+0.5);
}

\newcommand{\plusr}[2]{
\blight{#1}{#2};
\filldraw[fill=black!20!red,draw=black!20!red] (#1+0.5,#2+0.5) circle (0.25cm);
\draw[rline] (#1+0.5,#2+0.5) -- (#1+1,#2+0.5);
\node[text centered] at (#1+0.5,#2+0.5) {\color{white} $+$};
}

\newcommand{\minr}[2]{
\blight{#1}{#2};
\filldraw[fill=black!20!red,draw=black!20!red] (#1+0.5,#2+0.5) circle (0.25cm);
\draw[rline] (#1,#2+0.5) -- (#1+0.5,#2+0.5);
\node[text centered] at (#1+0.5,#2+0.5) {\color{white} $-$};
}

\newcommand{\flatr}[2]{
\blight{#1}{#2};
\draw[rline] (#1,#2+0.5) -- (#1+1,#2+0.5);
}

\newcommand{\flatR}[2]{
\bdark{#1}{#2};
\draw[rline] (#1,#2+0.5) -- (#1+1,#2+0.5);
}

\newcommand{\plusb}[2]{
\bdark{#1}{#2};
\filldraw[fill=black,draw=black] (#1+0.5,#2+0.5) circle (0.25cm);
\draw[bline] (#1+0.5,#2+0.5) -- (#1+1,#2+0.5);
\node[text centered] at (#1+0.5,#2+0.5) {\color{white} $+$};
}

\newcommand{\minb}[2]{
\bdark{#1}{#2};
\filldraw[fill=black,draw=black] (#1+0.5,#2+0.5) circle (0.25cm);
\draw[bline] (#1,#2+0.5) -- (#1+0.5,#2+0.5);
\node[text centered] at (#1+0.5,#2+0.5) {\color{white} $-$};
}

\newcommand{\flatb}[2]{
\bdark{#1}{#2};
\draw[bline] (#1,#2+0.5) -- (#1+1,#2+0.5);
}

\newcommand{\flatB}[2]{
\blight{#1}{#2};
\draw[bline] (#1,#2+0.5) -- (#1+1,#2+0.5);
}

\newcommand{\clear}[2]{
\wlight{#1}{#2};
}

\newcommand{\rhor}[2]{
\wlight{#1}{#2};
\draw[rline] (#1,#2+0.5) -- (#1+1,#2+0.5);
}

\newcommandx{\rp}[2]{
\wlight{#1}{#2};
\draw[rline] (#1+1,#2+0.5) -- (#1+0.5,#2+0.5) -- (#1+0.5,#2+0);
}

\newcommandx{\rmin}[2]{
\wlight{#1}{#2};
\draw[rline] (#1+0.5,#2+1) -- (#1+0.5,#2+0.5) -- (#1,#2+0.5);
}

\newcommand{\rver}[2]{
\wlight{#1}{#2};
\draw[rline] (#1+0.5,#2) -- (#1+0.5,#2+1);
}

\newcommand{\bhor}[2]{
\wlight{#1}{#2};
\draw[bline] (#1,#2+0.5) -- (#1+1,#2+0.5);
}

\newcommand{\bver}[2]{
\wlight{#1}{#2};
\draw[bline] (#1+0.5,#2) -- (#1+0.5,#2+1);
}

\newcommandx{\bp}[2]{
\wlight{#1}{#2};
\draw[bline] (#1+1,#2+0.5) -- (#1+0.5,#2+0.5) -- (#1+0.5,#2+0);
}

\newcommandx{\bmin}[2]{
\wlight{#1}{#2};
\draw[bline] (#1+0.5,#2+1) -- (#1+0.5,#2+0.5) -- (#1,#2+0.5);
}

\newcommand{\rb}[2]{
\wlight{#1}{#2};
\draw[rline] (#1,#2+0.5) -- (#1+1,#2+0.5);
\draw[bline] (#1+0.5,#2) -- (#1+0.5,#2+1);
}

\newcommand{\bra}[1]{\langle #1|}
\newcommand{\ket}[1]{|#1\rangle}
\newcommand{\cev}[1]{\xleftarrow{#1}}
\renewcommand{\vec}[1]{\xrightarrow{#1}}
\renewcommand{\b}[1]{\bar{#1}}
\renewcommand{\leq}{\leqslant}
\renewcommand{\geq}{\geqslant}
\renewcommand\ss{\scriptstyle}
\newcommand\sss{\scriptscriptstyle}

\def\phid{\phi^\dagger}
\def\psid{\psi^\dagger}

\newtheorem{defn}{Definition}
\newtheorem{lem}{Lemma}
\newtheorem{thm}{Theorem}

\newtheorem{prop}{Proposition}

\newtheorem{rmk}{Remark}
\newtheorem{ex}{Example}

\def\gr{\color{black!20!green}}

\def\re{\color{black!20!red}}

\title[]
{Hall polynomials, inverse Kostka polynomials and puzzles}

\author{M.~Wheeler}
\address{Michael Wheeler, School of Mathematics and Statistics, University of Melbourne, Parkville, Victoria 3010, Australia}
\email{mwheeler@ms.unimelb.edu.au}

\author{P.~Zinn-Justin}
\address{Paul Zinn-Justin, Laboratoire de Physique Th\'eorique et Hautes \'Energies, CNRS UMR 7589 and Universit\'e Pierre et Marie Curie (Paris 6), 4 place Jussieu, 75252 Paris cedex 05, France}
\email{pzinn@lpthe.jussieu.fr}

\begin{document}

\begin{abstract}
We study two different one-parameter generalizations of Littlewood--Richardson coefficients, namely Hall polynomials and generalized inverse Kostka polynomials, and derive new combinatorial formulae for them. Our combinatorial expressions are closely related to puzzles, originally introduced by Knutson and Tao in their work on the equivariant cohomology of the Grassmannian.
\end{abstract}

\maketitle

\section{Introduction}

\subsection{Background}

Littlewood--Richardson coefficients, $c^{\lambda}_{\mu\nu}$, are the structure constants for the multiplication of Schur functions:
\begin{align}
\label{schur-prod}
s_{\mu} s_{\nu} = \sum_{\lambda} c^{\lambda}_{\mu\nu} s_{\lambda}.
\end{align}
where $\lambda,\mu,\nu$ are three partitions.
They are also structure constants of the Grothendieck ring of polynomial
representations of $GL(n)$ (ignoring all partitions with more than $n$ rows); if we discard as well partitions whose first part is greater than $m$, we obtain the structure constants of the cohomology ring of the Grassmannian
$Gr(n,m+n)$. The computation of $c^{\lambda}_{\mu\nu}$ has a long history,
starting with the original work of Littlewood and Richardson \cite{lit-ric}
providing a rule for computing them, which was only rigorously proved some forty years later by Thomas \cite{tho} and Sch\"utzenberger \cite{sch}; since then, many alternative formulations have been given, including
the puzzle rule of Knutson and Tao \cite{knu-tao} used in their proof of the saturation conjecture \cite{knu-tao(saturation)}. The puzzle formulation is important for our purposes
because it displays most explicitly the underlying integrability of
the $c^{\lambda}_{\mu\nu}$ \cite{z-j}.

\subsection{Description of the paper}
Schur functions are specializations of a diverse range of symmetric functions\footnote{In a different direction, and one that preserves links with enumerative geometry, they are also special cases of Schubert polynomials (see, for example, the reviews \cite{mac2} and \cite{knu}).}. These include Hall--Littlewood, $q$--Whittaker and Jack polynomials as one-parameter generalizations, and Macdonald polynomials with two extra parameters. It is natural to ask whether the puzzle approach of \cite{knu-tao,knu-tao-woo} can be applied to these other families, leading to new combinatorial expressions for the relevant structure constants. 
In this paper, we focus on the Hall--Littlewood polynomials. There are various motivations for introducing them: {\bf 1.} They interpolate between Schur polynomials (at the value $t=0$ of the parameter), Schur $Q$-polynomials (at $t=-1$) and monomial symmetric polynomials (at $t=1$); {\bf 2.} On the geometric/representation-theoretic side, they occur
when lifting the Borel--Weil construction to the cotangent bundle of the flag variety \cite{bro,bry}, and (up to a plethystic substitution) in relation to the cohomology of the Springer fiber \cite{HS}; {\bf 3.} Their structure constants, the Hall polynomials, count short exact sequences of finite abelian $p$-groups for $t=p$ \cite{hal}.

We will study here two different generalizations of $c^{\lambda}_{\mu\nu}$. The first are 
the Hall polynomials $f^{\lambda}_{\mu\nu}(t)$, mentioned just above, which arise in the multiplication of 
Hall--Littlewood polynomials \cite{mac}:
\begin{align}
\label{HL-prod}
P_{\mu} P_{\nu} = \sum_{\lambda} f^{\lambda}_{\mu\nu}(t) P_{\lambda}.
\end{align}
Our first key result is Theorem \ref{thm-hall}, which expresses $f^{\lambda}_{\mu\nu}(t)$ via the action of divided-difference operators on a partition function $\mathcal{F}^{\lambda}_{\mu\nu}(x;t)$ in an integrable lattice model. More specifically, $\mathcal{F}^{\lambda}_{\mu\nu}(x;t)$ is a non-symmetric homogeneous polynomial of degree $d$ in a set of variables $x$. By acting with a product of $d$ divided-difference operators on $\mathcal{F}^{\lambda}_{\mu\nu}(x;t)$ we recover 
$f^{\lambda}_{\mu\nu}(t)$, up to known multiplicative factors. It is by no means unusual that divided-difference operators should be used in such a context. Indeed, Schubert polynomials can be defined via the action of divided-difference operators on a staircase monomial $\prod_{i=1}^{n} x_i^{n-i}$ \cite{las-sch2}, and the Schubert structure constants can themselves be obtained via the action of skew divided-difference operators on Schubert polynomials \cite{mac2,kir}. By evaluating the action of these operators on the partition function $\mathcal{F}^{\lambda}_{\mu\nu}(x;t)$ we obtain Theorem \ref{thm-hall-puzzle}, a combinatorial formula for 
$f^{\lambda}_{\mu\nu}(t)$ in terms of lattice tilings. 
  
The second generalization pertains to an intermediate situation, namely the product of a Schur and Hall--Littlewood polynomial expanded in the Schur basis:
\begin{align}
\label{schur-HL-prod}
s_{\mu} P_{\nu} = \sum_{\lambda} \b{K}^{\lambda}_{\mu\nu}(t) s_{\lambda}.
\end{align}
The structure constants $\b{K}^{\lambda}_{\mu\nu}(t)$ are clear generalizations of the inverse Kostka polynomials \cite{mac} (the latter are recovered by setting $\mu = 0$), but to our best knowledge they are unstudied in the literature. In Theorem \ref{thm-kost}, we express $\b{K}^{\lambda}_{\mu\nu}(t)$ via the action of divided-difference operators on a partition function $\mathcal{K}^{\lambda}_{\mu\nu}(x;t)$ in another integrable lattice model, different from that of Theorem \ref{thm-hall}. Evaluating the action of these operators on the partition function $\mathcal{K}^{\lambda}_{\mu\nu}(x;t)$, we arrive at Theorem \ref{thm-kost-puzzle}, a combinatorial formula for 
$\b{K}^{\lambda}_{\mu\nu}(t)$ in terms of lattice tilings analogous to those used in Theorem \ref{thm-hall-puzzle}.

It is important to note that, while the tiling formulae in Theorems \ref{thm-hall-puzzle} and \ref{thm-kost-puzzle} are conceptually analogous to the puzzles of \cite{knu-tao,knu-tao-woo}, they do not reduce to the latter at $t=0$. This raises the question of where the Knutson--Tao puzzles sit in the landscape of this paper. It turns out that they can be recovered from the same model (at $t=0$) as that used to study $\b{K}^{\lambda}_{\mu\nu}(t)$, but using a different partition function, which we denote $\mathcal{C}^{\lambda}_{\mu\nu}(x)$. This is the content of Theorem \ref{thm-LR}. From there we are able to express $c^{\lambda}_{\mu\nu}$ in terms of similar lattice tilings to those used to evaluate $f^{\lambda}_{\mu\nu}(t)$ and $\b{K}^{\lambda}_{\mu\nu}(t)$, before stating that these tilings are in bijection with the puzzles of \cite{knu-tao,knu-tao-woo}.

\subsection{General approach}

Since the same basic method will be applied in each of the cases \eqref{schur-prod}--\eqref{schur-HL-prod}, it is worthwhile to give a unified description here. First, we remark that each product rule has a corresponding {\it coproduct} version which arises from duality arguments. The coproduct versions of \eqref{schur-prod}--\eqref{schur-HL-prod} are, respectively,
\begin{align}
\label{coproducts}
s_{\lambda/\mu} = \sum_{\nu} c^{\lambda}_{\mu\nu} s_{\nu},
\quad
Q_{\lambda/\mu} = \sum_{\nu} f^{\lambda}_{\mu\nu}(t) Q_{\nu},
\quad
S_{\lambda/\mu} = \sum_{\nu} \b{K}^{\lambda}_{\mu\nu}(t) Q_{\nu},
\end{align}
where the left hand side of each identity is a certain symmetric polynomial assigned to a skew Young diagram (see \cite{mac} and the rest of this paper for further details). It is these identities which will be our focus.

In the initial step of our approach, we write the left hand side of each coproduct identity \eqref{coproducts} as an expectation value of monodromy matrix operators in a suitable integrable lattice model. In the case of skew Hall--Littlewood $Q_{\lambda/\mu}$ and Schur polynomials 
$s_{\lambda/\mu}$, it is well-known that this can be done using a model of $t$-deformed bosons (and its $t=0$ specialization) \cite{tsi,korff,bor,bor-pet,whe-z-j,bog}. In the case of the skew polynomials 
$S_{\lambda/\mu}$ (which are the skew version of the ``$t$-Schur'' polynomials\footnote{Macdonald gave no name to $S_{\lambda}(x;t)$ in \cite{mac}, but they are sometimes termed ``big'' Schur functions in the literature.} defined in Section 4, Chapter III of \cite{mac}) a similar construction is possible, using a free-fermionic six-vertex model. All such expectation values can be interpreted schematically as lattice partition functions of the form
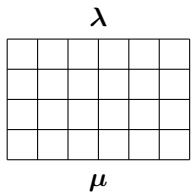
\begin{align*}
\begin{tikzpicture}[scale=0.4]
\node[text centered] at (3,4.75) {$\bm \lambda$};
\node[text centered] at (3,-0.75) {$\bm \mu$};
\foreach\y in {0,...,4}{
\draw (0,\y) -- (6,\y);
}
\foreach\x in {0,...,6}{
\draw (\x,0) -- (\x,4);
}
\end{tikzpicture}
\end{align*}
in which the partitions $\lambda$, $\mu$ are encoded as particle states at the top and bottom of the lattice (and the left and right edges have appropriate, uniform boundary conditions). A rapidity variable $x_i$ is assigned to the $i$-th row of the lattice for all $i$, and the uniformity of the side boundary conditions ensures that the partition function is symmetric in its rapidities.

In the second step, it is necessary to find a higher-rank version of the models discussed above. The first of these is a model containing two commuting copies of the $t$-boson algebra, and accordingly we refer to it as a ``boson-boson'' model. The entries of its $L$-matrix act in a tensor product $\mathcal{B} \otimes \mathcal{B}$ of two bosonic Fock spaces, meaning that states in the Hilbert space $H$ are now indexed by a {\it pair} of partitions. The second model is a ``fermion-boson'' model, and the entries of its $L$-matrix act in a tensor product $\mathcal{F} \otimes \mathcal{B}$ of a fermionic and bosonic Fock space. In both models, we are then interested in partition functions of the type
\begin{align}
\label{general-puzzle}
\begin{tikzpicture}[scale=0.4,baseline=1cm]
\node[text centered] at (3,4.75) {$(\bm \lambda,\bm 0)$};
\node[text centered] at (3,-0.75) {$(\bm \mu,\ )$};
\node[text centered] at (-1.25,2) {$(\ ,\bm \nu)$};
\foreach\y in {0,...,4}{
\draw (0,\y) -- (6,\y);
}
\foreach\x in {0,...,6}{
\draw (\x,0) -- (\x,4);
}
\end{tikzpicture}
\end{align}
where $(\bm \lambda,\bm 0)$ denotes the state $\ket{\lambda} \otimes \ket{0}$ (\textit{i.e.} the state in $H$ indexed by the pair of partitions $\lambda$ and $0$) and $(\bm \mu,\ )$ denotes the state $\bra{\mu} \otimes \bra{\ }$ (\textit{i.e.} the state in $H^{*}$ indexed by $\mu$ and the empty, or trivial, partition\footnote{In this work, we distinguish between zero partitions $(0,\dots,0)$ containing a certain number of repetitions of 0, and the empty partition.}). The boundary conditions at the left edge of the lattice, denoted $(\ ,\bm \nu)$, are no longer uniform; they now encode a third partition $\nu$. As before, each row of the lattice has a corresponding rapidity variable, but the non-uniformity of the left boundary means that the resulting partition function is no longer symmetric in its $x$ variables.

This partition function can then be directly related to the structure constants under consideration. They are either recovered under the action of divided-difference operators (acting on the rapidity variables), as is the case for $f^{\lambda}_{\mu\nu}(t)$ and $\b{K}^{\lambda}_{\mu\nu}(t)$, or else can be read off as the coefficient of a single monomial in $x$, which is the case for $c^{\lambda}_{\mu\nu}$. The layout of the partition function \eqref{general-puzzle}, which is framed by three non-trivial partitions, resembles closely the puzzles introduced in \cite{knu-tao,knu-tao-woo}. In the case of $c^{\lambda}_{\mu\nu}$, we are able to make this correspondence precise.

\subsection{Notation and conventions} 

A partition $\lambda = (\lambda_1,\dots,\lambda_{\ell})$ is a finite, weakly-decreasing sequence of non-negative integers: $\lambda_1 \geq \cdots \geq \lambda_{\ell} \geq 0$. The part-multiplicities $m_i(\lambda)$ of a partition $\lambda$ indicate the number of repetitions of each part:
\begin{align*}
\lambda = 0^{m_0} 1^{m_1} 2^{m_2} \dots,
\end{align*}
and we take the length of $\lambda$ to be the sum of all multiplicities, $\ell(\lambda) = \sum_{i \geq 0} m_i(\lambda)$. Note that {\it we include parts of size zero in our definition of the length,} which is non-standard. The complement of a partition by $L$, denoted $\b{\lambda}$ when the value of $L$ is clear from context, is the partition with parts 
\begin{align*}
\b\lambda_i = L - \lambda_{\ell-i+1},
\quad
1 \leq i \leq \ell,
\quad
\ell \equiv \ell(\lambda).
\end{align*}

We use the notation $\lambda/\mu \in \mathfrak{h}$ to indicate that the skew diagram $\lambda-\mu$ forms a horizontal strip, and $\lambda/\mu \in \mathfrak{h}_k$ to specify that the horizontal strip contains $k$ boxes. Similar conventions apply to $\lambda/\mu \in \mathfrak{v}$ and $\lambda/\mu \in \mathfrak{v}_k$, which indicate vertical strips.

Divided-difference operators are denoted $\Delta_{i,j}$, and defined as
\begin{align*}
\Delta_{i,j}
=
\frac{1}{x_i - x_{j}}
(\sigma_{i,j}-1),
\end{align*}
where $\sigma_{i,j}$ transposes $x_i$ and $x_j$, namely $\sigma_{i,j}(g(\dots,x_i,\dots,x_j,\dots)) = g(\dots,x_j,\dots,x_i,\dots)$ for any function $g$. We will only consider the case where $j=i+1$, when it is convenient to write $\Delta_{i,i+1} \equiv \Delta_i$.

To improve clarity of presentation, we will distinguish bosonic and fermionic operators by colour-coding. Bosonic operators, vector spaces and corresponding basis elements will be coloured green (or black), while those of fermions will be coloured red.

\section{Hall--Littlewood polynomials from a rank-one model of bosons}

\subsection{Bosonic Fock space \texorpdfstring{$\mathcal{B}$}{B}}

Consider a semi-infinite one-dimensional lattice, with sites labelled by non-negative integers. In a finite configuration of this lattice each site $i \geq 0$ is occupied by $m_i \geq 0$ particles, and there exists $M \in \mathbb{N}$ such that $m_k = 0$ for all $k \geq M$. The bosonic Fock space $\mathcal{B}$ is obtained by taking linear combinations of all possible finite configurations:
\begin{align}
\label{boson-space}
\mathcal{B}
=
{\rm Span} 
\left\{
\ket{m_0}_0
\otimes
\ket{m_1}_1
\otimes
\ket{m_2}_2
\otimes
\cdots
\right\},
\qquad
m_i \geq 0,\ \forall\ i \geq 0.
\end{align}
The Fock space $\mathcal{B}$ will be the physical space for the model that we are about to study. It has a natural dual vector space 
\begin{align*}
\mathcal{B}^{*} = {\rm Span} 
\left\{
\bra{m_0}_0
\otimes
\bra{m_1}_1
\otimes
\bra{m_2}_2
\otimes
\cdots
\right\},
\qquad
m_i \geq 0,\ \forall\ i \geq 0,
\end{align*}
whose action on $\mathcal{B}$ is deduced by linearity and the relation
\begin{align*}
\langle m | n \rangle = \prod_{i=0}^{\infty} \delta_{m_i,n_i},
\quad
\forall\
\bra{m} = \bigotimes_{k=0}^{\infty} \bra{m_k}_k,
\quad
\ket{n} = \bigotimes_{k=0}^{\infty} \ket{n_k}_k.
\end{align*}

\subsection{Mapping partitions to states in \texorpdfstring{$\mathcal{B}$}{B}}

There is a simple mapping between partitions and the basis vectors of $\mathcal{B}$. Given a partition 
$\lambda = 0^{m_0} 1^{m_1} 2^{m_2} \dots$, associate the following state $\ket{{\gr\vec \lambda}} \in \mathcal{B}$:
\begin{align*}
\ket{{\gr\vec \lambda}}
=
\bigotimes_{k=0}^{\infty}
\ket{m_k}_k.
\end{align*}
Whenever we wish to display the partition $\lambda = (\lambda_1,\dots,\lambda_{\ell})$ explicitly we will write $\ket{{\gr\vec \lambda}} = \ket{{\gr (\lambda_{\ell},\dots,\lambda_1)}}$, with parts increasing from left to right. We will use $\ket{\ }$ to denote the completely empty state; the state for which all occupation numbers $m_i$ are zero.

\begin{ex}
\label{ex1}
{\rm 
Consider the partition $\lambda = (5,3,3,1,0)$, for which $m_0(\lambda) = 1$, $m_1(\lambda) = 1$, $m_3(\lambda) = 2$, $m_5(\lambda) = 1$, and all other $m_k(\lambda) = 0$. Then
\begin{align*}
\ket{{\gr\vec \lambda}}
=
\ket{{\gr 1}}_0
\otimes
\ket{{\gr 1}}_1
\otimes
\ket{{\gr 0}}_2
\otimes
\ket{{\gr 2}}_3
\otimes
\ket{{\gr 0}}_4
\otimes
\ket{{\gr 1}}_5
\otimes
\ket{{\gr 0}}_6
\otimes
\cdots.
\end{align*}
This mapping admits a simple pictorial interpretation. Starting from the Young diagram of $\lambda$ (in its traditional orientation), we assign particles to each upward edge at the boundary of $\lambda$. These particles are then projected down onto the integer lattice, whereby some sites become multiply-occupied:
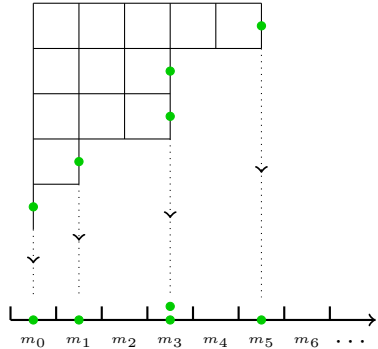
\begin{center}
\begin{tikzpicture}[scale=0.6]
\draw (0,4) -- (5,4);
\draw (0,3) -- (5,3);
\draw (0,2) -- (3,2);
\draw (0,1) -- (3,1);
\draw (0,0) -- (1,0);
\draw (0,-1) -- (0,4);
\draw (1,0) -- (1,4);
\draw (2,1) -- (2,4);
\draw (3,1) -- (3,4);
\draw (4,3) -- (4,4);
\draw (5,3) -- (5,4);
\gbull{0}{-0.5}{0.09};
\gbull{1}{0.5}{0.09};
\gbull{3}{1.5}{0.09};\gbull{3}{2.5}{0.09};
\gbull{5}{3.5}{0.09};
\draw[dotted,arrow=0.5] (0,-1) -- (0,-2.5);
\draw[dotted,arrow=0.5] (1,0) -- (1,-2.5);
\draw[dotted,arrow=0.5] (3,1) -- (3,-2.5);
\draw[dotted,arrow=0.5] (5,3) -- (5,-2.5);
\draw[thick,arrow=1] (-0.5,-3) -- (7.5,-3);
\foreach\x in {1,...,8}{
\draw[thick] (-1.5+\x,-3) -- (-1.5+\x,-2.7);
}
\gbull{0}{-3}{0.09};
\gbull{1}{-3}{0.09};
\gbull{3}{-3}{0.09};\gbull{3}{-2.7}{0.09};
\gbull{5}{-3}{0.09};
\node at (0,-3.5) {$\sss m_0$};
\node at (1,-3.5) {$\sss m_1$};
\node at (2,-3.5) {$\sss m_2$};
\node at (3,-3.5) {$\sss m_3$};
\node at (4,-3.5) {$\sss m_4$};
\node at (5,-3.5) {$\sss m_5$};
\node at (6,-3.5) {$\sss m_6$};
\node at (7,-3.5) {$\cdots$};
\end{tikzpicture}
\end{center}
The corresponding basis vector $\ket{{\gr \vec \lambda}} \in \mathcal{B}$ is then obtained by reading the occupation numbers along the integer lattice.
}
\end{ex}

\subsection{Reverse partition states}
\label{sec:rev}

Let $\lambda = 0^{m_0} 1^{m_1} 2^{m_2} \dots$ be a partition and $L=\lambda_1$ its largest part. The {\it reverse partition state} $\ket{{\gr \cev \lambda}} \in \mathcal{B}$ is defined as
\begin{align*}
\ket{{\gr \cev \lambda}}
=
\bigotimes_{k=0}^{\infty}
\ket{\overline{m}_k}_k,
\quad\quad
\overline{m}_k = m_{L-k},
\ \forall\ 0 \leq k \leq L,
\quad\quad
\overline{m}_k = 0,
\ \forall\ k > L.
\end{align*}
Reverse partition states can also be written as $\ket{{\gr \cev \lambda}} = \ket{{\gr (\lambda_1,\dots,\lambda_{\ell})}}$, with parts decreasing from left to right.

\begin{ex}
{\rm
Let $\mu = (5,4,2,2,0)$, for which $m_0(\mu) = 1$, $m_2(\mu) = 2$, $m_4(\mu) = 1$, $m_5(\mu) =1$, and all other $m_k(\mu)=0$. This yields $\overline{m}_5(\mu) = 1$, $\overline{m}_3(\mu) = 2$, $\overline{m}_1(\mu) = 1$, $\overline{m}_0(\mu) =1$, and all other $\overline{m}_k(\mu)=0$. Then
\begin{align*}
\ket{{\gr\cev \mu}}
=
\ket{{\gr 1}}_0
\otimes
\ket{{\gr 1}}_1
\otimes
\ket{{\gr 0}}_2
\otimes
\ket{{\gr 2}}_3
\otimes
\ket{{\gr 0}}_4
\otimes
\ket{{\gr 1}}_5
\otimes
\ket{{\gr 0}}_6
\otimes
\cdots,
\end{align*}
which is the same state as in Example \ref{ex1}. Graphically, the reverse partition state can be obtained by first rotating the Young diagram of $\mu$ by 180\degree, then performing the same projection procedure as above:
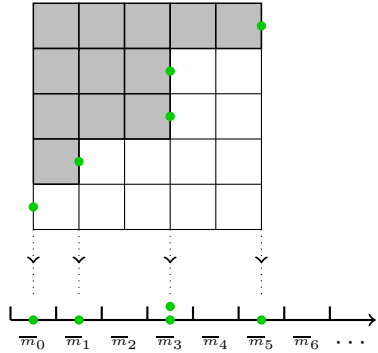
\begin{center}
\begin{tikzpicture}[scale=0.6]
\bdark{0}{3}; \bdark{1}{3}; \bdark{2}{3}; \bdark{3}{3}; \bdark{4}{3};
\bdark{0}{2}; \bdark{1}{2}; \bdark{2}{2};
\bdark{0}{1}; \bdark{1}{1}; \bdark{2}{1};
\bdark{0}{0}; 
\draw (3,3) -- (5,3); 
\draw (3,2) -- (5,2);
\draw (1,1) -- (5,1);
\draw (0,0) -- (5,0);
\draw (0,-1) -- (5,-1);
\draw (0,-1) -- (0,0);
\draw (1,-1) -- (1,1); 
\draw (2,-1) -- (2,1);
\draw (3,-1) -- (3,3);
\draw (4,-1) -- (4,3);
\draw (5,-1) -- (5,4);
\gbull{0}{-0.5}{0.09};
\gbull{1}{0.5}{0.09};
\gbull{3}{1.5}{0.09};\gbull{3}{2.5}{0.09};
\gbull{5}{3.5}{0.09};
\draw[dotted,arrow=0.5] (0,-1) -- (0,-2.5);
\draw[dotted,arrow=0.5] (1,-1) -- (1,-2.5);
\draw[dotted,arrow=0.5] (3,-1) -- (3,-2.5);
\draw[dotted,arrow=0.5] (5,-1) -- (5,-2.5);
\draw[thick,arrow=1] (-0.5,-3) -- (7.5,-3);
\foreach\x in {1,...,8}{
\draw[thick] (-1.5+\x,-3) -- (-1.5+\x,-2.7);
}
\gbull{0}{-3}{0.09};
\gbull{1}{-3}{0.09};
\gbull{3}{-3}{0.09};\gbull{3}{-2.7}{0.09};
\gbull{5}{-3}{0.09};
\node at (0,-3.5) {$\sss \overline{m}_0$};
\node at (1,-3.5) {$\sss \overline{m}_1$};
\node at (2,-3.5) {$\sss \overline{m}_2$};
\node at (3,-3.5) {$\sss \overline{m}_3$};
\node at (4,-3.5) {$\sss \overline{m}_4$};
\node at (5,-3.5) {$\sss \overline{m}_5$};
\node at (6,-3.5) {$\sss \overline{m}_6$};
\node at (7,-3.5) {$\cdots$};
\end{tikzpicture}
\end{center}
From this picture the relationship with Example \ref{ex1} is explained: the chosen partitions $\lambda$ and $\mu$ are related under complementation ($\lambda$ is indicated using darkly shaded boxes) This leads us to the following result, directly relating reverse partition states with the complementation of partitions.

}
\end{ex}

\begin{prop}
\label{rev-prop}
{\rm
Let $\lambda$ be a partition with largest part $L$, and $\b\lambda$ the complement of $\lambda$ by $L$. Then $\ket{{\gr \cev \lambda}} = \ket{{\gr \vec {\b\lambda}}}$.
}
\end{prop}

\subsection{A rank-one model of bosons}
\label{model-bb}

Introduce bosonic creation and annihilation operators ${\gr \phid}$ and ${\gr \phi}$, and the particle-number operator ${\gr N}$, which satisfy the bilinear relations
\begin{align}
\label{t-bos-rel}
{\gr \phi} {\gr \phid}
-
t
{\gr \phid} {\gr \phi}
=
1-t,
\quad
[{\gr N}, {\gr \phi}] = -1,
\quad
[{\gr N}, {\gr \phid}] = 1.
\end{align}
Let ${\gr \mathfrak{b}}$ denote the algebra generated by $\{ {\gr \phi}, {\gr \phid}, {\gr N} \}$ modulo the relations \eqref{t-bos-rel}. We shall use the standard Fock representation of ${\gr \mathfrak{b}}$ on a single site of $\mathcal{B}$:
\begin{align*}
{\gr \phid} \ket{{\gr m}}
=
(1-t^{m+1})
\ket{{\gr m+1}},
\quad
{\gr \phi} \ket{{\gr m}}
=
\ket{{\gr m-1}},
\quad
{\gr N} \ket{{\gr m}}
=
m
\ket{{\gr m}},
\end{align*}
where it is assumed that ${\gr \phi} \ket{{\gr 0}} = 0$. 
%
The $L$-matrix\footnote{We denote the $L$-matrix \eqref{Lmat-b} by $L^{*}$ for consistency with our notation in \cite{whe-z-j}.}
\begin{align}
\label{Lmat-b}
L^{*}_a(x|{\gr\mathfrak b})
=
\begin{pmatrix}
1 & {\gr \phid}
\\
x {\gr \phi} & x
\end{pmatrix}_a
=
\left(
\begin{array}{cc}
\begin{tikzpicture}[scale=0.6]
\gblank{0}{0}
\end{tikzpicture} 
& 
\begin{tikzpicture}[scale=0.6]
\gplus{0}{0}
\end{tikzpicture} 
\\
\begin{tikzpicture}[scale=0.6]
\gminus{0}{0}
\end{tikzpicture} 
& 
\begin{tikzpicture}[scale=0.6]
\gfull{0}{0}
\end{tikzpicture} 
\end{array}
\right)_a
\end{align}
satisfies the intertwining equation
\begin{align*}
R_{ab}(y/x)
L^{*}_a(x|{\gr\mathfrak b})
L^{*}_b(y|{\gr\mathfrak b})
=
L^{*}_b(y|{\gr\mathfrak b})
L^{*}_a(x|{\gr\mathfrak b})
R_{ab}(y/x),
\end{align*}
with $R$-matrix given by
\begin{align}
\label{Rmat-b}
R_{ab}(z)
=
\left(
\begin{array}{cc|cc}
1-tz & 0 & 0 & 0
\\
0 & t(1-z) & (1-t)z & 0
\\
\hline
0 & 1-t & 1-z & 0
\\
0 & 0 & 0 & 1-tz
\end{array}
\right)_{ab}.
\end{align}
It is very useful to adopt a graphical representation for the entries of the $L$-matrix, as shown on the right of \eqref{Lmat-b}. We refer to these as {\it tiles.} Expectation values of an entry of the $L$-matrix are then indicated by placing occupation numbers at the top and bottom edges of its corresponding tile. For example,
\begin{align*}
\bra{{\gr 3}}
{\gr \phid}
\ket{{\gr 2}}
\equiv
\begin{tikzpicture}[scale=0.8,baseline=1cm]
\gplus[2]{1}{1}[3]  
\end{tikzpicture}
=
(1-t^3),
\qquad
\bra{{\gr 0}}
x
{\gr \phi}
\ket{{\gr 1}}
\equiv
\begin{tikzpicture}[scale=0.8,baseline=1cm]
\gminus[1]{1}{1}  
\end{tikzpicture}
=
x.
\end{align*}
An alternative notation is to indicate an occupation number, $n$, by $n$ lines that propagate vertically through a tile (see, for example, \cite{bor-pet}), but we will mostly avoid this convention in the present work. 

The preceding construction gives the $t$-boson model \cite{bog-bul,bog-ize-kit,whe-z-j}, and the $R$-matrix \eqref{Rmat-b} is that of the six-vertex model. We construct a monodromy matrix by taking a product of $L$-matrices over all non-negative sites $i$ of the integer lattice:
\begin{align*}
T^{*}_a(x)
=
L^{*}_a(x|{\gr \mathfrak{b}_0})
L^{*}_a(x|{\gr \mathfrak{b}_1})
\cdots
=
\prod_{i=0}^{\infty}
L^{*}_a(x|{\gr \mathfrak{b}_i})
:=
\begin{pmatrix}
A(x) & B(x)
\\
C(x) & D(x)
\end{pmatrix}_a,
\end{align*}
where the $L$-matrix at each site depends on its own copy of the algebra ${\gr \mathfrak{b}}$, and different copies are assumed to be commuting. In our subsequent constructions we mainly consider $A(x) \in {\rm End}(\mathcal{B})$, which can be regarded as the sum of all possible (semi-infinite) rows of the tiles \eqref{Lmat-b}, whose left-most tile is \begin{tikzpicture}[scale=0.5,baseline=0.1cm] \gblank{0}{0} \end{tikzpicture} or \begin{tikzpicture}[scale=0.5,baseline=0.1cm] \gplus{0}{0} \end{tikzpicture} and such that sufficiently far to the right only \begin{tikzpicture}[scale=0.5,baseline=0.1cm] \gblank{0}{0} \end{tikzpicture} occurs. One has $[A(x),A(y)] = 0$ for all $x,y$, which follows immediately from the intertwining equation once it is applied to two monodromy matrices $T^{*}_a(x), T^{*}_b(y)$.

\subsection{Skew Hall--Littlewood polynomials \texorpdfstring{$P_{\lambda/\mu}(x;t)$}{P lambda/mu(x;t)}}

We survey some basic facts about Hall--Littlewood polynomials. For more details on the theory, we refer the reader to \cite{mac}. Hall--Littlewood polynomials are a $t$-deformation of the Schur polynomials, which are given explicitly by
\begin{align*}
P_{\lambda}(x_1,\dots,x_n;t)
=
\sum_{\sigma \in S_n}
\sigma\left(
\prod_{i=1}^{n}
x_i^{\lambda_i}
\prod_{1 \leq i<j \leq n}
\frac{x_i-tx_j}{x_i-x_j}
\right).
\end{align*}
They are orthogonal with respect to the Hall inner product:
\begin{align}
\label{ortho-rel}
\langle Q_{\lambda}, P_{\mu} \rangle
=
\delta_{\lambda,\mu},
\qquad
Q_{\lambda}(x_1,\dots,x_n;t)
=
b_{\lambda}(t)
P_{\lambda}(x_1,\dots,x_n;t),
\qquad
b_{\lambda}(t)
=
\prod_{i \geq 1}
\prod_{j=1}^{m_i(\lambda)}
(1-t^j).
\end{align}
This orthogonality relation can be used to define skew Hall--Littlewood polynomials $P_{\lambda/\mu}(x;t)$:
\begin{align}
\label{skew-def}
\langle Q_{\lambda}, P_{\mu} P_{\nu} \rangle
:=
\langle Q_{\lambda/\mu}, P_{\nu} \rangle,
\qquad
Q_{\lambda/\mu}(x_1,\dots,x_n;t)
=
\frac{b_{\lambda}(t)}{b_{\mu}(t)}
P_{\lambda/\mu}(x_1,\dots,x_n;t).
\end{align}
From the point of view of integrable lattice models, the most important property of the Hall--Littlewood polynomials is their branching rule:
\begin{align}
\label{hl-branch}
P_{\lambda}(x_1,\dots,x_n,y_1,\dots,y_m;t)
=
\sum_{\mu}
P_{\lambda/\mu}(x_1,\dots,x_n;t)
P_{\mu}(y_1,\dots,y_m;t),
\end{align}
valid for any partition $\lambda$ and any two sets of variables $(x_1,\dots,x_n)$, $(y_1,\dots,y_m)$, and where the sum on the right hand side is taken over all partitions $\mu$. This allows 
$P_{\lambda/\mu}(x_1,\dots,x_n;t)$ to be constructed recursively, provided that the skew Hall--Littlewood polynomial in one variable is known. In fact, the latter has a very simple form:
\begin{align}
\label{1var-P}
P_{\lambda/\mu}(z;t)
=
\left\{
\begin{array}{ll}
\psi_{\lambda/\mu}(t) 
z^{|\lambda-\mu|},
\quad
&
\lambda / \mu \in \mathfrak{h},
\\ \\
0,
\quad
&
\text{otherwise},
\end{array}
\right.
\qquad
\psi_{\lambda/\mu}(t)
=
\prod_{i \geq 1: m_i(\mu) = m_i(\lambda)+1}
(1-t^{m_i(\mu)}).
\end{align}
In view of \eqref{1var-P} and the branching rule \eqref{hl-branch}, skew Hall--Littlewood polynomials can then be expressed as a sum over sequences of interlacing partitions (equivalently, in terms of semi-standard Young tableaux):
\begin{align}
\label{skew-ssyt}
P_{\lambda/\mu}(x_1,\dots,x_n;t)
&=
\sum_{\mu \equiv \nu^{(0)} \prec \cdots \prec \nu^{(n)} \equiv \lambda}\
\prod_{i=1}^{n}
\psi_{\nu^{(i)} / \nu^{(i-1)}}(t)\
x_i^{|\nu^{(i)} - \nu^{(i-1)}|},
\end{align}
where we write $\nu^{(i-1)} \prec \nu^{(i)}$ to indicate that $\nu^{(i)}/\nu^{(i-1)} \in \mathfrak{h}$.

\subsection{Lattice expression for \texorpdfstring{$P_{\lambda/\mu}(x;t)$}{P lambda/mu(x;t)}}

Next, we review the well-known construction of Hall--Littlewood polynomials as expectation values in the $t$-boson model \cite{tsi,korff}.

\begin{lem}{\rm 
Skew Hall--Littlewood polynomials are given by the following expectation value:
\begin{align}
\label{skew-HL}
P_{\lambda/\mu}(x_1,\dots,x_n;t)
=
\frac{\prod_{j=1}^{m_0(\lambda)}(1-t^j)}{\prod_{j=1}^{m_0(\mu)}(1-t^j)}
\bra{{\gr\vec\mu}}
A(x_1)
\dots
A(x_n)
\ket{{\gr\vec\lambda}}.
\end{align}
%
}
\end{lem}

\begin{proof}
By inserting a complete set of states in the expectation value on the right hand side of \eqref{skew-HL}, we have
\begin{align*}
\bra{{\gr\vec \mu}}
A(x_1)
\dots
A(x_n)
\ket{{\gr\vec \lambda}}
=
\sum_{\nu}
\bra{{\gr \vec \mu}}
A(x_1)
\dots
A(x_{n-1})
\ket{{\gr \vec \nu}}
\bra{{\gr \vec \nu}}
A(x_n)
\ket{{\gr \vec \lambda}}.
\end{align*}
Owing to the branching rule \eqref{hl-branch}, to complete the proof it suffices to show that
\begin{align}
\label{1var-exp}
P_{\lambda/\nu}(z;t)
=
\frac{\prod_{j=1}^{m_0(\lambda)}(1-t^j)}{\prod_{j=1}^{m_0(\nu)}(1-t^j)}
\bra{{\gr\vec \nu}}
A(z)
\ket{{\gr\vec \lambda}}
\end{align}
for any two partitions $\lambda$ and $\nu$, which is simply the one-variable case of \eqref{skew-HL}. We must show that the right hand side of \eqref{1var-exp} matches that of \eqref{1var-P}. By studying the lattice representation of $\bra{{\gr\vec \nu}} A(z) \ket{{\gr\vec \lambda}}$, for example
\begin{align*}
\bra{{\gr (0,0,1,3,4)}}
A(z)
\ket{{\gr (0,1,3,3,5)}}
=
\begin{tikzpicture}[scale=0.9,baseline=1cm]
\gplus[1]{1}{1}[2]  
\gfull[1]{2}{1}[1]
\gfull[0]{3}{1}[0]
\gminus[2]{4}{1}[1]
\gplus[0]{5}{1}[1]
\gminus[1]{6}{1}[0]
\gblank[0]{7}{1}[0]
\node[text centered] at (1.5,0.5) {\gr \tiny $m_0(\nu)$};
\node[text centered] at (2.5,0.5) {\gr \tiny $m_1(\nu)$};
\node[text centered] at (3.5,0.5) {\gr \tiny $m_2(\nu)$};
\node[text centered] at (4.5,0.5) {\gr \tiny $m_3(\nu)$};
\node[text centered] at (5.5,0.5) {\gr \tiny $m_4(\nu)$};
\node[text centered] at (6.5,0.5) {\gr \tiny $m_5(\nu)$};
\node[text centered] at (7.5,0.5) {\gr \tiny $\cdots$};
\node[text centered] at (1.5,2.5) {\gr \tiny $m_0(\lambda)$};
\node[text centered] at (2.5,2.5) {\gr \tiny $m_1(\lambda)$};
\node[text centered] at (3.5,2.5) {\gr \tiny $m_2(\lambda)$};
\node[text centered] at (4.5,2.5) {\gr \tiny $m_3(\lambda)$};
\node[text centered] at (5.5,2.5) {\gr \tiny $m_4(\lambda)$};
\node[text centered] at (6.5,2.5) {\gr \tiny $m_5(\lambda)$};
\node[text centered] at (7.5,2.5) {\gr \tiny $\cdots$};
\end{tikzpicture}
=
z^4 (1-t^2)(1-t),
\end{align*}
it is immediately apparent that $\bra{{\gr\vec \nu}} A(z) \ket{{\gr\vec \lambda}} = 0$ if $\lambda/\nu \not\in \mathfrak{h}$ (since there is no way to connect the two states using only the tiles in the $L$-matrix \eqref{Lmat-b}). 
If $\lambda/\nu\in \mathfrak{h}$, we acquire a weight of $z$ for every horizontal unit step taken by the green line, which gives the correct factor $z^{|\lambda - \nu|}$. Also, in the transition from $\lambda$ to $\nu$, every time the number of particles at a lattice site increases from $m-1$ to $m$ we acquire a factor of $1-t^m$. This gives rise to the total factor
$
\psi_{\lambda/\nu}(t)
\prod_{i=1}^{m_0(\nu)}(1-t^i) / \prod_{i=1}^{m_0(\lambda)}(1-t^i)
$.

\end{proof}

\begin{ex}{\rm 
Let $\lambda = (3,1,1)$ and $\mu = (0,0,0)$. Then $P_{\lambda/\mu}(x_1,x_2,x_3;t) = P_{\lambda}(x_1,x_2,x_3;t)$ can be expressed as
\begin{align*}
\prod_{i=1}^{m_0(\mu)}
(1-t^i)
P_{\lambda}(x_1,x_2,x_3;t)
&=
\begin{tikzpicture}[scale=0.8,baseline=2cm]
\gplus[0]{1}{3}
\gminus[2]{2}{3}
\gblank[0]{3}{3}
\gblank[1]{4}{3}
\gplus[1]{1}{2}
\gminus[1]{2}{2}
\gblank[0]{3}{2}
\gblank[1]{4}{2}
%
\gplus[2]{1}{1}[3]
\gfull[0]{2}{1}[0]
\gfull[0]{3}{1}[0]
\gminus[1]{4}{1}[0]
\end{tikzpicture}
+
\begin{tikzpicture}[scale=0.8,baseline=2cm]
\gplus[0]{1}{3}
\gminus[2]{2}{3}
\gblank[0]{3}{3}
\gblank[1]{4}{3}
\gplus[1]{1}{2}
\gminus[1]{2}{2}
\gplus[0]{3}{2}
\gminus[1]{4}{2}
\gplus[2]{1}{1}[3]
\gfull[0]{2}{1}[0]
\gminus[1]{3}{1}[0]
\gblank[0]{4}{1}[0]
\end{tikzpicture}
+
\begin{tikzpicture}[scale=0.8,baseline=2cm]
\foreach\x in {1,...,3}{\node at (5.5,\x+0.5) {$x_\x$};}
\gplus[0]{1}{3}
\gminus[2]{2}{3}
\gblank[0]{3}{3}
\gblank[1]{4}{3}
\gplus[1]{1}{2}
\gfull[1]{2}{2}
\gfull[0]{3}{2}
\gminus[1]{4}{2}
\gplus[2]{1}{1}[3]
\gminus[1]{2}{1}[0]
\gblank[0]{3}{1}[0]
\gblank[0]{4}{1}[0]
\end{tikzpicture}
\\
&+
\begin{tikzpicture}[scale=0.8,baseline=2cm]
\gplus[0]{1}{3}
\gminus[2]{2}{3}
\gplus[0]{3}{3}
\gminus[1]{4}{3}
\gplus[1]{1}{2}
\gfull[1]{2}{2}
\gminus[1]{3}{2}
\gblank[0]{4}{2}
\gplus[2]{1}{1}[3]
\gminus[1]{2}{1}[0]
\gblank[0]{3}{1}[0]
\gblank[0]{4}{1}[0]
\end{tikzpicture}
+
\begin{tikzpicture}[scale=0.8,baseline=2cm]
\gplus[0]{1}{3}
\gminus[2]{2}{3}
\gplus[0]{3}{3}
\gminus[1]{4}{3}
\gplus[1]{1}{2}
\gminus[1]{2}{2}
\gblank[1]{3}{2}
\gblank[0]{4}{2}
\gplus[2]{1}{1}[3]
\gfull[0]{2}{1}[0]
\gminus[1]{3}{1}[0]
\gblank[0]{4}{1}[0]
\end{tikzpicture}
+
\begin{tikzpicture}[scale=0.8,baseline=2cm]
\foreach\x in {1,...,3}{\node at (5.5,\x+0.5) {$x_\x$};}
\gplus[0]{1}{3}
\gfull[2]{2}{3}
\gfull[0]{3}{3}
\gminus[1]{4}{3}
\gplus[1]{1}{2}
\gminus[2]{2}{2}
\gblank[0]{3}{2}
\gblank[0]{4}{2}
\gplus[2]{1}{1}[3]
\gminus[1]{2}{1}[0]
\gblank[0]{3}{1}[0]
\gblank[0]{4}{1}
\end{tikzpicture}
\end{align*}
in which all possible lattice configurations are summed over. Calculating the Boltzmann weight of each configuration, we obtain the explicit sum 
\begin{align*}
&
\prod_{i=1}^{m_0(\mu)} (1-t^i)
P_{\lambda}(x_1,x_2,x_3;t)
=
\\
&
(1-t)(1-t^2)(1-t^3)
\Big( x_1^3 x_2 x_3 + (1-t) x_1^2 x_2^2 x_3 + x_1 x_2^3 x_3 
+ (1-t) x_1 x_2^2 x_3^2 + x_1 x_2 x_3^3 + (1-t) x_1^2 x_2 x_3^2 \Big).
\end{align*}
The common factor $\prod_{i=1}^{m_0(\mu)} (1-t^i) = (1-t)(1-t^2)(1-t^3)$ originates from the zeroth column, which is the same across all configurations. 
}
\end{ex}

The expression \eqref{skew-HL} depends on the partitions $\lambda$ and $\mu$ only via the boundary conditions. Naturally, there is the option to evaluate the partition function under 180\degree\ rotation, which leads to 
\begin{lem}{\rm 
Skew Hall--Littlewood polynomials are given alternatively by the following expectation value: 
\begin{align*}
Q_{\lambda/\mu}(x_1,\dots,x_n;t)
=
\bra{{\gr\cev\lambda}}
A(x_1)
\dots
A(x_n)
\ket{{\gr\cev\mu}}
=
\bra{{\gr\vec{\b\lambda}}}
A(x_1)
\dots
A(x_n)
\ket{{\gr\vec{\b\mu}}},
\end{align*}
where $\b\lambda$ and $\b\mu$ denote the complements of $\lambda$ and $\mu$ by $L$, the largest part of 
$\lambda$.
}
\end{lem}

\begin{proof}
By rearranging factors, equation \eqref{skew-HL} can be written as
\begin{align*}
Q_{\lambda/\mu}(x_1,\dots,x_n;t)
&=
\frac{B_{\lambda}(t)}{B_{\mu}(t)}
\bra{{\gr\vec\mu}}
A(x_1)
\dots
A(x_n)
\ket{{\gr\vec\lambda}},
\qquad
B_{\lambda}(t) := \prod_{i=1}^{m_0(\lambda)} (1-t^i) b_{\lambda}(t).
\end{align*}
Now take the lattice representation of $\bra{{\gr\vec\mu}}
A(x_1)
\dots
A(x_n)
\ket{{\gr\vec\lambda}}$, and consider rotating the associated partition function by $180\degree$. From the form of the $L$-matrix \eqref{Lmat-b}, it is easy to deduce that it remains invariant under this rotation, up to the factor $B_{\mu}(t)/B_{\lambda}(t)$, which cancels the factor already present. Hence,
\begin{align*}
Q_{\lambda/\mu}(x_1,\dots,x_n;t)
=
\bra{{\gr\cev\lambda}}
A(x_n)
\dots
A(x_1)
\ket{{\gr\cev\mu}},
\end{align*}
and we are done in view of the commutativity of the $A(x_i)$ operators.

\end{proof}

\section{Hall polynomials from a rank-two model of bosons}

\subsection{Model with two bosons}
\label{model-2-bosons}

We now consider a higher-rank version of the model in Section \ref{model-bb}. This model has also appeared recently in \cite{cd-gw}, where it is the second in a series of higher-rank bosonic models, used in an integrable construction of Macdonald polynomials. The entries of its $L$-matrix act in ${\gr \mathcal{B}} \otimes \mathcal{B}$, the tensor product of two bosonic spaces, and are given by
\begin{align}
\label{Lmat-bb}
L_a(x| {\gr \mathfrak{b}} \otimes \mathfrak{b} )
=
\begin{pmatrix}
x[1 \otimes 1] & x[1 \otimes {\phid}] & x[{\gr \phid} \otimes 1]
\\ \\
t^{\gr N} \otimes {\phi} &  t^{\gr N} \otimes 1 & 0
\\ \\
{\gr \phi} \otimes 1 & {\gr \phi} \otimes {\phid}  &  1 \otimes 1 
\end{pmatrix}_a
=
\left(
\begin{array}{ccc}
\begin{tikzpicture}[scale=0.6,baseline=0.25cm]
\gblank{0}{0}
\bblank{1}{0}
\end{tikzpicture} 
& 
\begin{tikzpicture}[scale=0.6,baseline=0.25cm]
\gblank{0}{0}
\bplus{1}{0}
\end{tikzpicture} 
&
\begin{tikzpicture}[scale=0.6,baseline=0.25cm]
\gplus{0}{0}
\darkgfull{1}{0}
\end{tikzpicture} 
\\ \\
\begin{tikzpicture}[scale=0.6,baseline=0.25cm]
\lightbfull{0}{0}
\bminus{1}{0}
\end{tikzpicture} 
& 
\begin{tikzpicture}[scale=0.6,baseline=0.25cm]
\lightbfull{0}{0}
\bfull{1}{0}
\end{tikzpicture} 
&
0
\\ \\
\begin{tikzpicture}[scale=0.6,baseline=0.25cm]
\gminus{0}{0}
\bblank{1}{0}
\end{tikzpicture} 
& 
\begin{tikzpicture}[scale=0.6,baseline=0.25cm]
\gminus{0}{0}
\bplus{1}{0}
\end{tikzpicture} 
&
\begin{tikzpicture}[scale=0.6,baseline=0.25cm]
\gfull{0}{0}
\darkgfull{1}{0}
\end{tikzpicture} 
\end{array}
\right)_a.
\end{align}
By considering only the entries (1,1), (1,3), (3,1) and (3,3) of the $L$-matrix \eqref{Lmat-bb} it reduces to the model \eqref{Lmat-b}, up to inversion of and multiplication by $x$. The $L$-matrix satisfies the intertwining equation
\begin{align}
\label{inter-2-bosons}
R_{ab}(x/y)
L_a(x| {\gr \mathfrak{b}} \otimes \mathfrak{b})
L_b(y| {\gr \mathfrak{b}} \otimes \mathfrak{b})
=
L_b(y| {\gr \mathfrak{b}} \otimes \mathfrak{b})
L_a(x| {\gr \mathfrak{b}} \otimes \mathfrak{b})
R_{ab}(x/y)
\end{align}
with respect to the $R$-matrix
\begin{align*}
R_{ab}(z)
=
\left(
\begin{array}{ccc|ccc|ccc}
1-tz & 0 & 0 & 0 & 0 & 0 & 0 & 0 & 0
\\
0 & t(1-z) & 0 & (1-t)z & 0 & 0 & 0 & 0 & 0
\\
0 & 0 & t(1-z) & 0 & 0 & 0 & (1-t)z & 0 & 0
\\
\hline
0 & 1-t & 0 & 1-z & 0 & 0 & 0 & 0 & 0
\\
0 & 0 & 0 & 0 & 1-tz & 0 & 0 & 0 & 0
\\
0 & 0 & 0 & 0 & 0 & t(1-z) & 0 & (1-t)z & 0
\\
\hline
0 & 0 & 1-t & 0 & 0 & 0 & 1-z & 0 & 0
\\
0 & 0 & 0 & 0 & 0 & 1-t & 0 & 1-z & 0
\\
0 & 0 & 0 & 0 & 0 & 0 & 0 & 0 & 1-tz
\end{array}
\right)_{ab}.
\end{align*}
We construct a monodromy matrix using the $L$ matrix, namely
\begin{align}
\label{mon-prod}
T_a(x)
=
L_a(x|{\gr \mathfrak b_0} \otimes {\mathfrak{b}_0})
L_a(x|{\gr \mathfrak b_1} \otimes {\mathfrak{b}_1})
\cdots
=
\prod_{i=0}^{\infty}
L_a(x|{\gr \mathfrak b_i} \otimes {\mathfrak{b}_i}).
\end{align}
Since $L_a(x|{\gr \mathfrak b} \otimes \mathfrak{b})$ is a $3 \times 3$ matrix, the result of the matrix multiplication \eqref{mon-prod} is also a $3 \times 3$ matrix:
\begin{align*}
T_a(x)
=
\begin{pmatrix}
\Te(x) & \star & \star
\\
\Tb(x) & \star & \star
\\
\Tg(x) & \star & \star
\end{pmatrix}_a,
\end{align*}
where $\Te(x), \Tb(x), \Tg(x)$ are operators acting in $\mathcal{\gr B} \otimes \mathcal{B}$. We suppress the matrix entries in the remaining two columns, since they are never used in our calculations. Two commutation relations will be important in our later calculations:
\begin{align*}
(y-x)
\Tb(x) \Te(y)
+
(1-t) y
\Te(x) \Tb(y)
=
(y-tx)
\Te(y) \Tb(x),
\quad\quad
\Tb(x) \Tb(y)
=
\Tb(y) \Tb(x).
\end{align*}

\subsection{Definition of the Hall polynomials}

Hall polynomials $f^{\lambda}_{\mu\nu}(t)$ are the structure constants of the Hall algebra, originally studied in \cite{hal}. Equivalently, they are the expansion coefficients in the product of two Hall--Littlewood polynomials:
\begin{align*}
P_{\mu}(x;t)
P_{\nu}(x;t)
=
\sum_{\lambda}
f^{\lambda}_{\mu\nu}(t)
P_{\lambda}(x;t).
\end{align*}
It is a non-trivial fact that $f^{\lambda}_{\mu\nu}(t) \in \mathbb{Z}[t]$, justifying the nomenclature ``polynomial''. However, unlike the Kostka polynomials (to be discussed later on), $f^{\lambda}_{\mu\nu}(t)$ are polynomials in $t$ with integer coefficients (rather than natural number coefficients). {\it A priori,} this rules out hope of a positive combinatorial formula for $f^{\lambda}_{\mu\nu}(t)$. In addition to their polynomiality they satisfy a number of properties \cite{mac}:
\begin{enumerate}[{\bf 1.}]
\item If the Littlewood--Richardson coefficient $c^{\lambda}_{\mu\nu}=0$, then $f^{\lambda}_{\mu\nu}(t) = 0$ identically. This implies, in particular, that $f^{\lambda}_{\mu\nu}(t)$ is non-zero only if $|\lambda| = |\mu| + |\nu|$ and $\mu,\nu \subseteq \lambda$.

\item For all $\lambda,\mu,\nu$, one has the symmetry $f^{\lambda}_{\mu\nu}(t) = f^{\lambda}_{\nu\mu}(t)$. 

\item The quasi-symmetry $f^{\lambda}_{\mu\nu}(t) = f^{\b\mu}_{\b\lambda \nu}(t) B_{\b\lambda}(t) / B_{\b\mu}(t)$ holds, where $\b\lambda$ and $\b\mu$ are the $L$-complements of $\lambda$ and $\mu$, $L$ being the largest part of $\lambda$.

\end{enumerate}

Due to the definition \eqref{skew-def}, the Hall polynomials can also be considered as the expansion coefficients of a skew Hall--Littlewood polynomial:
\begin{align*}
Q_{\lambda/\mu}(x;t)
=
\sum_{\nu}
f^{\lambda}_{\mu\nu}(t)
Q_{\nu}(x;t),
\end{align*}
and it is this latter point of view which is the most convenient for our purposes. Our aim in this work is to obtain a new combinatorial formula for $f^{\lambda}_{\mu\nu}(t)$, using the integrable model introduced above. 

\subsection{Hall polynomials as constant terms}
\label{CT}

Here we derive a useful identity expressing the Hall polynomial $f^{\lambda}_{\mu\nu}(t)$ as the coefficient of a particular monomial in a certain formal power series (which we refer to as a constant term identity). 

\begin{lem}
\label{CT-lem}
{\rm
Let $\lambda,\mu,\nu$ be three partitions, and suppose $\ell(\nu) = n$. The (renormalized) Hall polynomial $f^{\lambda}_{\mu\nu}(t) b_{\nu}(t)$ is given by the coefficient of $z_1^{\nu_1} \dots z_n^{\nu_n}$ in the expression
\begin{align*}
\prod_{1 \leq i<j \leq n}
\left(
\frac{1-z_j/z_i}{1- t z_j/z_i}
\right)
Q_{\lambda/\mu}(z_1,\dots,z_n;t),
\end{align*}
interpreted as a formal power series in $z_1,\dots,z_n$ (or by assuming $|z_j/z_i|$ are small for all $1 \leq i<j \leq n$.)
}
\end{lem}

\begin{proof}
Following Section 2, Chapter III of \cite{mac}, the Hall--Littlewood polynomial $Q_{\lambda}(x_1,\dots,x_m;t)$ is given by the coefficient of $z_1^{\lambda_1} \dots z_n^{\lambda_n}$ in the expression
\begin{align*}
\prod_{1 \leq i<j \leq n}
\left(
\frac{1-z_j/z_i}{1- t z_j/z_i}
\right)
\prod_{i=1}^{m}
\prod_{j=1}^{n}
\left(
\frac{1-t x_i z_j}{1-x_i z_j}
\right).
\end{align*}
Combining this with the Cauchy identity
\begin{align*}
\prod_{i=1}^{m}
\prod_{j=1}^{n}
\left(
\frac{1-t x_i z_j}{1-x_i z_j}
\right)
=
\sum_{\mu}
Q_{\mu}(x_1,\dots,x_m;t)
P_{\mu}(z_1,\dots,z_n;t),
\end{align*}
we find that
\begin{align*}
\sum_{\mu}
Q_{\mu}(x_1,\dots,x_m;t)
\times
{\rm Coeff}\left[
\prod_{1 \leq i<j \leq n}
\left(
\frac{1-z_j/z_i}{1- t z_j/z_i}
\right)
P_{\mu}(z_1,\dots,z_n;t),
z_1^{\lambda_1} \dots z_n^{\lambda_n}
\right]
=
Q_{\lambda}(x_1,\dots,x_n;t).
\end{align*}
By the linear independence of the $\{Q_{\mu}\}$, this means that
\begin{align*}
{\rm Coeff}\left[
\prod_{1 \leq i<j \leq n}
\left(
\frac{1-z_j/z_i}{1- t z_j/z_i}
\right)
P_{\mu}(z_1,\dots,z_n;t),
z_1^{\lambda_1} \dots z_n^{\lambda_n}
\right]
=
\delta_{\lambda,\mu}.
\end{align*}
Hence if we expand any symmetric polynomial $G(x_1,\dots,x_n)$ in terms of Hall--Littlewood polynomials, $G(x_1,\dots,x_n) = \sum_{\lambda} g_{\lambda}(t) P_{\lambda}(x_1,\dots,x_n;t)$, the expansion coefficients are given by
\begin{align*}
{\rm Coeff}\left[
\prod_{1 \leq i<j \leq n}
\left(
\frac{1-z_j/z_i}{1- t z_j/z_i}
\right)
G(z_1,\dots,z_n),
z_1^{\lambda_1} \dots z_n^{\lambda_n}
\right]
=
g_{\lambda}(t).
\end{align*}
In particular, choosing $G(x_1,\dots,x_n) = Q_{\lambda/\mu}(x_1,\dots,x_n;t)$, we recover the Hall polynomials:
\begin{align}
\label{CT-Hall}
{\rm Coeff}\left[
\prod_{1 \leq i<j \leq n}
\left(
\frac{1-z_j/z_i}{1- t z_j/z_i}
\right)
Q_{\lambda/\mu}(z_1,\dots,z_n;t),
z_1^{\nu_1} \dots z_n^{\nu_n}
\right]
=
f^{\lambda}_{\mu\nu}(t)
b_{\nu}(t).
\end{align}
\end{proof}

\subsection{Hall polynomials from action of divided-difference operators}

\begin{defn}{\rm
Let $\nu = (\nu_1,\dots,\nu_n)$ be a partition. Associate a set 
$k(\nu) = \{k_1,\dots,k_n\}$ to $\nu$, where 
\begin{align*}
k_i = i+\sum_{j=n-i+1}^{n} \nu_j.
\end{align*}
We call $k(\nu)$ the {\it shifted partial sums} of $\nu$. One simple interpretation of $k(\nu)$ comes from arranging the boxes of a Young diagram $\nu$ along the integer lattice. For example,
\begin{align*}
\nu = (3,1,1,0)
\leftrightarrow
\begin{tikzpicture}[baseline=-0.75cm,scale=0.5]
\blackoc{1}{-1.5}{}
\blackoc{3}{-1.5}{}
\blackoc{5}{-1.5}{}
\blackoc{6}{-1.5}{}
\blackoc{7}{-1.5}{}
\draw (0.5,-1) -- (1.5,-1); 
\draw (0.5,-2) -- (1.5,-2);
\draw (0.5,-2) -- (0.5,-1); \draw (1.5,-2) -- (1.5,-1);
\draw (2.5,-1) -- (3.5,-1); \draw (2.5,-2) -- (3.5,-2);
\draw (2.5,-2) -- (2.5,-1); \draw (3.5,-2) -- (3.5,-1);
\draw (4.5,-1) -- (7.5,-1); \draw (4.5,-2) -- (7.5,-2);
\draw (4.5,-2) -- (4.5,-1); \draw (5.5,-2) -- (5.5,-1); 
\draw (6.5,-2) -- (6.5,-1); \draw (7.5,-2) -- (7.5,-1);
\draw[thick,arrow=1] (-0.5,-3) -- (9.5,-3);
\foreach\x in {1,...,10}{
\draw[thick] (-1.5+\x,-3) -- (-1.5+\x,-2.7);
}
\node at (0,-3.5) {$\sss 1$};
\node at (2,-3.5) {$\sss 3$};
\node at (4,-3.5) {$\sss 5$};
\node at (8,-3.5) {$\sss 9$};
\end{tikzpicture}
\leftrightarrow
k(\nu) = \{1,3,5,9\}.
\end{align*}
}
\end{defn}

\begin{thm}
\label{thm-hall}
{\rm
Let $\lambda, \mu, \nu$ be three partitions with respective lengths $\ell,m,n$ such that $\ell=m$, and largest parts $L,M,N$ such that $M \geq L$ and $M \geq N$. Let $\b\lambda,\b\mu,\b\nu$ denote the partitions obtained by complementing each part of $\lambda,\mu,\nu$ by $M$, and assume that $|\b\lambda| = |\b\mu|+|\b\nu|$. Then
\begin{align}
\label{Hall-PF}
\mathcal{F}^{\lambda}_{\mu\nu}(x;t)
:=
\bra{{\gr\cev \mu}}
\otimes
\bra{\ }
\prod_{i=1}^{n+|\nu|}
\left\{
\begin{array}{ll}
\Te(x_i),
&
i \in k(\nu)
\\
\Tb(x_i),
&
i \not\in k(\nu)
\end{array}
\right\}
\ket{{\gr\cev \lambda}}
\otimes
\ket{{\cev 0}}
\end{align}
is a homogeneous polynomial in $\{x_1,\dots,x_{n+|\nu|}\}$ of total degree $n+|\nu|$, and the Hall polynomial $f^{\b\lambda}_{\b\mu \b\nu}(t)$ is given by the repeated action of divided-difference operators on this partition function:
\begin{align}
\label{main-formula}
f^{\b\lambda}_{\b\mu \b\nu}(t)
=
t^{-(m+1)|\nu|}
\times
\frac{B_{\b\lambda}(t)}
{B_{\b\mu}(t) b_{\b\nu}(t)}
\times
\left(
\prod_{j \not\in k(\nu)}
\Delta_{j}
\right)
\left(
\prod_{i \in k(\nu)}
\b{x}_i
\right)
\mathcal{F}^{\lambda}_{\mu\nu}(x;t),
\end{align}
where the product of divided-difference operators is ordered from left to right, starting at $j=1$ and finishing at $j = n+|\nu|$, and omitting $j \in k(\nu)$.

}
\end{thm}

The rest of the section will be devoted to the proof of Theorem \ref{thm-hall}. Before proceeding to the proof, we give an example which better illustrates the meaning of \eqref{Hall-PF}.

\begin{ex}{\rm
Fix three partitions $\lambda = (2,1,0,0)$, $\mu = (3,3,1,1)$, $\nu = (2,1,1)$. Then $M=3$ and the complement partitions are given by $\b\lambda = (3,3,2,1)$, $\b\mu = (2,2,0,0)$, $\b\nu = (2,2,1)$, so that $|\b\lambda| = |\b\mu|+|\b\nu|$. We have
\begin{align*}
\mathcal{F}^{\lambda}_{\mu\nu}(x;t)
=
\bra{{\gr\cev \mu}}
\otimes
\bra{\ }
\Tb(x_1)\Te(x_2)\Tb(x_3)\Te(x_4)\Tb(x_5)\Tb(x_6)\Te(x_7)
\ket{{\gr\cev \lambda}}
\otimes
\ket{{\cev 0}}
\end{align*}
which has the graphical representation
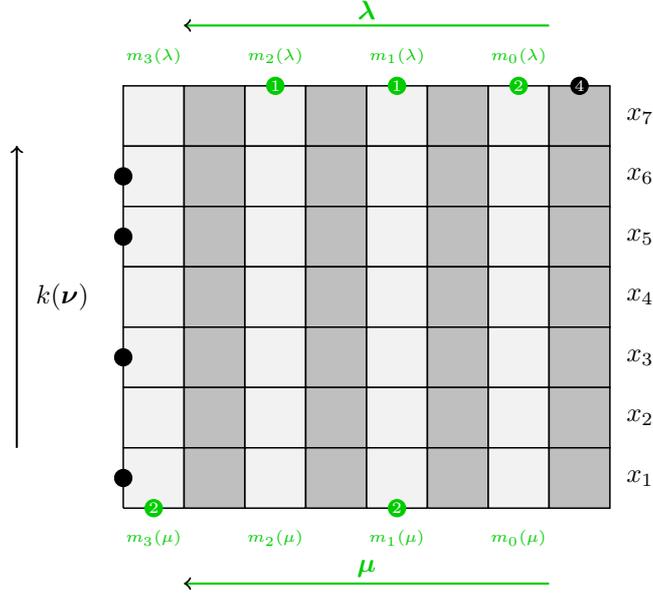
\begin{align*}
\begin{tikzpicture}[scale=0.8]
\foreach\x in {1,...,7}{\node at (9.5,\x-0.5) {$x_\x$};}
\draw[arrow=1,draw=black!20!green,thick] (8,-1.25) -- (2,-1.25);
\node at (5,-1) {${\gr \bm \mu}$};
\node[text centered] at (1.5,-0.5) {\gr \tiny $m_3(\mu)$};
\node[text centered] at (3.5,-0.5) {\gr \tiny $m_2(\mu)$};
\node[text centered] at (5.5,-0.5) {\gr \tiny $m_1(\mu)$};
\node[text centered] at (7.5,-0.5) {\gr \tiny $m_0(\mu)$};
\draw[arrow=1,draw=black!20!green,thick] (8,8) -- (2,8);
\node at (5,8.25) {${\gr \bm \lambda}$};
\node[text centered] at (1.5,7.5) {\gr \tiny $m_3(\lambda)$};
\node[text centered] at (3.5,7.5) {\gr \tiny $m_2(\lambda)$};
\node[text centered] at (5.5,7.5) {\gr \tiny $m_1(\lambda)$};
\node[text centered] at (7.5,7.5) {\gr \tiny $m_0(\lambda)$};
\draw[arrow=1,draw=black,thick] (-0.75,1) -- (-0.75,6);
\node at (0,3.5) {$k({\bm \nu})$};
\blackoc{1}{0.5}{}
\blackoc{1}{2.5}{}
\blackoc{1}{4.5}{}
\blackoc{1}{5.5}{}
\gblank{1}{6}
\bblank{2}{6}
\gblank[1]{3}{6}
\bblank{4}{6}
\gblank[1]{5}{6}
\bblank{6}{6}
\gblank[2]{7}{6}
\bblank[4]{8}{6}
\gblank{1}{5}
\bblank{2}{5}
\gblank{3}{5}
\bblank{4}{5}
\gblank{5}{5}
\bblank{6}{5}
\gblank{7}{5}
\bblank{8}{5}
\gblank{1}{4}
\bblank{2}{4}
\gblank{3}{4}
\bblank{4}{4}
\gblank{5}{4}
\bblank{6}{4}
\gblank{7}{4}
\bblank{8}{4}
\gblank{1}{3}
\bblank{2}{3}
\gblank{3}{3}
\bblank{4}{3}
\gblank{5}{3}
\bblank{6}{3}
\gblank{7}{3}
\bblank{8}{3}
\gblank{1}{2}
\bblank{2}{2}
\gblank{3}{2}
\bblank{4}{2}
\gblank{5}{2}
\bblank{6}{2}
\gblank{7}{2}
\bblank{8}{2}
\gblank{1}{1}
\bblank{2}{1}
\gblank{3}{1}
\bblank{4}{1}
\gblank{5}{1}
\bblank{6}{1}
\gblank{7}{1}
\bblank{8}{1}
\gblank{1}{0}[2]
\bblank{2}{0}
\gblank{3}{0}
\bblank{4}{0}
\gblank{5}{0}[2]
\bblank{6}{0}
\gblank{7}{0}
\bblank{8}{0}
\end{tikzpicture}
\end{align*}
The corresponding Hall polynomial is given by 
\begin{align*}
f^{\b\lambda}_{\b\mu \b\nu}(t)
=
\frac{
\Delta_{1} \Delta_{3} \Delta_{5} \Delta_{6} \b{x}_2  \b{x}_4 \b{x}_7
\mathcal{F}^{\lambda}_{\mu\nu}(x;t)}
{t^{16} (1-t)(1-t^2)^2}
.
\end{align*}
}
\end{ex}

\subsubsection{Proof of homogeneity and degree}

Since the $L$ matrices \eqref{Lmat-bb} are used in the construction of the partition function $\mathcal{F}^{\lambda}_{\mu\nu}(x;t)$, we obtain an $x$ weight every time the left edge of a lightly-shaded tile is vacant. Conversely, there is no $x$ weight when such an edge is occupied. Let us abbreviate these by VLEs and OLEs (vacant and occupied left edges).

In any legal lattice configuration, the green particles give rise to $|\mu| - |\lambda| = |\b\nu| = Mn-|\nu|$ OLEs, while the black particles give rise to $(M+1)|\nu|$. This is a combined 
$M(n+|\nu|)$ OLEs, meaning that in any legal configuration there are exactly $(M+1)(n+|\nu|)-M(n+|\nu|) = n+|\nu|$ VLEs, each with an associated $x$ weight. The homogeneity and degree of the polynomial follow immediately. Furthermore since there is a trivial VLE in the $k_i$-th row, for all $1 \leq i \leq n$, 
$\mathcal{F}^{\lambda}_{\mu\nu}(x;t)$ has an obvious common factor of $x_{k_1} \dots x_{k_n}$. This common factor is divided away by $\prod_{i \in k(\nu)} \b{x}_i$.

\subsubsection{Multiple sum expression for $\mathcal{F}^{\lambda}_{\mu\nu}(x;t)$}

The first step is to calculate $\mathcal{F}^{\lambda}_{\mu\nu}(x;t)$ explicitly as a repeated sum. To do that, we recall the commutation relations 
\begin{align}
\label{g-b}
\Tb(x) \Te(y) 
&=
\frac{y-tx}{y-x} \Te(y) \Tb(x)
+
\frac{(1-t)y}{x-y} \Te(x) \Tb(y),
\\
\label{g-g}
\Tb(x) \Tb(y)
&=
\Tb(y) \Tb(x).
\end{align}
\begin{lem}
\label{lem:bethe}
{\rm 
By virtue of the commutation relations \eqref{g-b} and \eqref{g-g}, one can derive the following relation:
\begin{multline}
\label{ggg-b}
\bra{{\gr\cev \mu}}
\otimes
\bra{\ }
\Tb(x_1)
\dots
\Tb(x_{k-1})
\Te(x_k)
=
\prod_{j=1}^{k-1}
\left(
\frac{x_k-tx_j}{x_k-x_j}
\right)
\bra{{\gr\cev \mu}}
\otimes
\bra{\ }
\Te(x_k)
\Tb(x_1)
\dots
\Tb(x_{k-1})
\\
+
x_k
\sum_{i=1}^{k-1}
\left(
\frac{(1-t)}{x_i-x_k}
\prod_{j\not=i}^{k-1}
\left(
\frac{x_i-tx_j}{x_i-x_j}
\right)
\bra{{\gr\cev \mu}}
\otimes
\bra{\ }
\Te(x_i)
\Tb(x_1)
\dots
\widehat{\Tb}(x_i)
\dots
\Tb(x_k)
\right),
\end{multline}
valid for any $k \geq 1$.}
\end{lem}

\begin{proof}
This is a well-known equation in models solvable by the algebraic Bethe Ansatz \cite{fad,kbi}. The leading term on the right hand side of \eqref{ggg-b} comes from applying $k-1$ times the commutation relation \eqref{g-b}, and retaining only the term $(x_k-tx_j)/(x_k-x_j) \Te(x_{k}) \Tb(x_j)$ at each step (since this is the only way to produce terms proportional to $\bra{{\gr\cev \mu}} \otimes \bra{\ } \Te(x_k) \Tb(x_1) \dots \Tb(x_{k-1})$, as we desire). 

Now isolate the $i=k-1$ term from the sum on the right hand side of \eqref{ggg-b}. This comes from applying the commutation relation \eqref{g-b} once and retaining the term 
$(1-t)x_k / (x_{k-1} - x_k)  \Te(x_{k-1}) \Tb(x_k)$, then applying \eqref{g-b} a further $k-2$ times and retaining only the term $(x_{k-1}-tx_j)/(x_{k-1}-x_j) \Te(x_{k-1}) \Tb(x_j)$ at each step (again, this is the only way to produce terms proportional to $\bra{{\gr\cev \mu}} \otimes \bra{\ }\Te(x_{k-1}) \Tb(x_1) \dots \Tb(x_{k-2}) \Tb(x_{k})$, which is our current focus). 

The other terms in the sum can be deduced by symmetrizing the $i=k-1$ term, since the left hand side of equation \eqref{ggg-b} is symmetric in $\{x_1,\dots,x_{k-1}\}$.
\end{proof}

Next, we notice that \eqref{ggg-b} can be written more succinctly as a single sum:
\begin{multline}
\label{1-iterat}
\b{x}_k
\times
\bra{{\gr\cev \mu}}
\otimes
\bra{\ }
\Tb(x_1)
\dots
\Tb(x_{k-1})
\Te(x_k)
=
\\
\sum_{i=1}^{k}
\b{x}_i
\frac{\prod_{j=1}^{k-1} (x_i-tx_j)}{\prod_{j\not=i}^{k} (x_i-x_j)}
\bra{{\gr\cev \mu}}
\otimes
\bra{\ }
\Te(x_i)
\Tb(x_1)
\dots
\widehat{\Tb}(x_i)
\dots
\Tb(x_k).
\end{multline}
Repeatedly iterating \eqref{1-iterat}, it is then straightforward to show that
\begin{multline}
\label{n-iterat}
\prod_{j=1}^{n}
\b{x}_{k_j}
\times
\mathcal{F}^{\lambda}_{\mu\nu}(x;t)
=
\sum_{1 \leq i_1 \leq k_1}
\cdots
\sum_{\substack{1 \leq i_{n} \leq k_{n} \\ i_n \not=i_1,\dots,i_{n-1}}}
(\b{x}_{i_1} \cdots \b{x}_{i_n})
\times
\\
\frac{\prod_{j_1}^{k_1-1} (x_{i_1}-tx_{j_1})}
{\prod_{j_1\not=i_1}^{k_1} (x_{i_1}-x_{j_1})}
\cdots
\frac{\prod_{j_n\not=i_1,\dots,i_{n-1}}^{k_n-1} (x_{i_n}-tx_{j_n})}
{\prod_{j_n\not=i_1,\dots,i_n}^{k_n} (x_{i_n}-x_{j_n})}
\bra{{\gr\cev \mu}}
\otimes
\bra{\ }
\Te(x_{i_1})
\dots
\Te(x_{i_n})
\Tb(x_{\hat\imath_1})
\dots
\Tb(x_{\hat\imath_{|\nu|}})
\ket{{\gr\cev \lambda}}
\otimes
\ket{{\cev 0}}
\end{multline}
where the summation is over distinct integers $\{i_1,\dots,i_n\}$ such that $i_j \leq k_j$ for all $1 \leq j \leq n$, and where $\{\hat\imath_1,\dots,\hat\imath_{|\nu|}\}$ denotes the complement of $\{i_1,\dots,i_n\}$ in $\{1,\dots,n+|\nu|\}$.

\subsubsection{Trivial action}

Here we present a lemma which enables us to proceed further in the calculation of 
$\mathcal{F}^{\lambda}_{\mu\nu}(x;t)$. It is a simple result but an important one, since it allows us to eliminate the $\Tb(x)$ operators from \eqref{n-iterat}, effectively converting it to an expectation value in the rank-one model \eqref{Lmat-b}. 

\begin{lem}
\label{lem:triv}
{\rm
Let $\ket{{\gr\cev \lambda}}$ be an arbitrary length-$\ell$ partition state in ${\gr \mathcal{B}}$, and $\ket{{\cev 0}}$ the length-$N$ zero partition $(0,\dots,0)$ in $\mathcal{B}$ ($N$ repetitions of 0). Then
\begin{align}
\label{trivial-act}
\Tb(y_1)
\dots
\Tb(y_N) 
\ket{{\gr\cev \lambda}}
\otimes
\ket{{\cev 0}}
=
t^{\ell N}
\ket{{\gr\cev \lambda}}
\otimes
\ket{\ }.
\end{align}

}
\end{lem}

\begin{proof}
Consider the action of a single operator $\Tb(y)$ on $\ket{{\gr\cev \lambda}} \otimes \ket{{\cev 0}}$, for example when $\ket{{\gr\cev \lambda}} = \ket{{\gr(2,1,0,0)}}$ and $\ket{{\cev 0}} = \ket{(0,0,0,0)}$, for which $\ell = N = 4$. Writing $\Tb(y) \ket{{\gr\cev \lambda}} \otimes \ket{{\cev 0}}$ in graphical form, it is clear that the black particle exiting on the left edge of the lattice must have come from the reservoir of particles at the top edge:
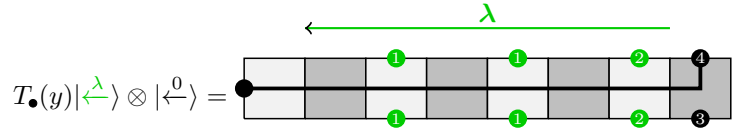
\begin{align*}
\Tb(y)
\ket{{\gr\cev \lambda}} \otimes \ket{{\cev 0}}
=
\begin{tikzpicture}[scale=0.8,baseline=5cm]
\draw[arrow=1,draw=black!20!green,thick] (8,7.5) -- (2,7.5);
\node at (5,7.75) {${\gr \bm \lambda}$};
\blackoc{1}{6.5}{}
\lightbfull{1}{6}
\bfull{2}{6}
\lightbfull[1]{3}{6}[1]
\bfull{4}{6}
\lightbfull[1]{5}{6}[1]
\bfull{6}{6}
\lightbfull[2]{7}{6}[2]
\bminus[4]{8}{6}[3]
\end{tikzpicture}
\end{align*}
Hence all internal vertical edges in this row are occupied by black particles, meaning that the green particles are forced to propagate in straight lines. The Boltzmann weight of the resulting configuration is $t^4=t^{\ell}$. Since this action reduces the number of black particles by one, by iterating it $N$ times all black particles are eliminated from the lattice, and we recover precisely \eqref{trivial-act}.
\end{proof}

\subsubsection{Multiple integral expression for $\mathcal{F}^{\lambda}_{\mu\nu}(x;t)$}

Returning to \eqref{n-iterat}, we immediately apply the relation \eqref{trivial-act}, giving
\begin{multline}
\label{something}
\prod_{j=1}^{n}
\b{x}_{k_j}
\times
\mathcal{F}^{\lambda}_{\mu\nu}(x;t)
=
t^{m |\nu|}
\sum_{1 \leq i_1 \leq k_1}
\cdots
\sum_{\substack{1 \leq i_{n} \leq k_{n} \\ i_n \not=i_1,\dots,i_{n-1}}}
(\b{x}_{i_1} \cdots \b{x}_{i_n})
\times
\\
\frac{\prod_{j_1}^{k_1-1} (x_{i_1}-t x_{j_1})}
{\prod_{j_1\not=i_1}^{k_1} (x_{i_1}-x_{j_1})}
\cdots
\frac{\prod_{j_n\not=i_1,\dots,i_{n-1}}^{k_n-1} (x_{i_n}-t x_{j_n})}
{\prod_{j_n\not=i_1,\dots,i_n}^{k_n} (x_{i_n}-x_{j_n})}
\bra{{\gr\cev \mu}}
\otimes
\bra{\ }
\Te(x_{i_1})
\dots
\Te(x_{i_n})
\ket{{\gr\cev \lambda}}
\otimes
\ket{\ },
\end{multline}
since $\ell=m$. It is clear that the final expectation value has no dependence on the black bosons (all operators contained within act only in ${\gr \mathcal{B}}$). Hence we can calculate it using the results of Section \ref{model-bb}. Recalling that the spectral parameter in the model \eqref{Lmat-bb} is the reciprocal of that used in \eqref{Lmat-b}, we find that
\begin{align*}
\bra{{\gr\cev \mu}}
\otimes
\bra{\ }
\Te(x_{i_1})
\dots
\Te(x_{i_n})
\ket{{\gr\cev \lambda}}
\otimes
\ket{\ }
=
\prod_{j=1}^{n} x_{i_j}^{M+1}
Q_{\mu/\lambda}(\b{x}_{i_1},\dots,\b{x}_{i_n};t)
=
\frac{B_{\b\mu}(t)}{B_{\b\lambda}(t)}
\prod_{j=1}^{n} x_{i_j}^{M+1}
Q_{\b\lambda/\b\mu}(\b{x}_{i_1},\dots,\b{x}_{i_n};t).
\end{align*}
We can then convert the sum \eqref{something} into a multiple integral whose integrand contains a skew Hall--Littlewood polynomial:
\begin{multline}
\label{first-int}
\prod_{j=1}^{n}
\b{x}_{k_j}
\times
\mathcal{F}^{\lambda}_{\mu\nu}(x;t)
=
t^{m |\nu|}
\times
\frac{B_{\b\mu}(t)}{B_{\b\lambda}(t)}
\times
\\
\oint_{w_1}
\cdots
\oint_{w_n}
\prod_{1 \leq i<j \leq n}
\frac{
(w_j - w_i)
}
{
(w_j-t w_i)
}
\prod_{i=1}^{n}
\frac{
\prod_{j=1}^{k_i-1}
(w_i - t x_j)
}
{
\prod_{j=1}^{k_i}
(w_i - x_j)
}
\prod_{i=1}^{n}
w_i^{M}
Q_{\b\lambda/\b\mu}(\b{w}_1,\dots,\b{w}_n;t).
\end{multline}
where each contour of integration is identical, and can be taken to be a circle centred on the origin and surrounding the points $(x_1,\dots,x_{n+|\nu|})$, which are the only poles of the integrand.

At this stage, it would be ideal to take some limit of the parameters $x_1,\dots,x_{n+|\nu|}$, with the aim of recovering the constant term expression \eqref{CT-Hall} for the Hall polynomial. We were unable to find any specialization of $x_1,\dots,x_{n+|\nu|}$ that achieves this. The most obvious try, $x_1,\dots,x_{n+|\nu|} \rightarrow 0$, fails trivially (since the right hand side of \eqref{first-int} is necessarily still a homogeneous polynomial of degree $|\nu|$ in these variables). This suggests that the correct approach is to lower the degree of \eqref{first-int} by acting with appropriate degree-lowering operators.   

\subsubsection{Acting with divided-difference operators}

Consider the following simple relation involving the divided-difference operator $\Delta_{k-1} \equiv \Delta_{k-1,k}$, where $g(x)$ is any symmetric function in $(x_1,\dots,x_k)$:
\begin{align}
\label{dd-act}
\Delta_{k-1}
\left(
\frac{w-t x_1}{w-x_1}
\cdots
\frac{w-t x_{k-1}}{w- x_{k-1}}
\frac{g(x)}{w-x_k}
\right)
&=
\frac{1}{x_{k-1}-x_k}
(\sigma_{k-1,k}-1)
\left(
\frac{w-t x_1}{w-x_1}
\cdots
\frac{w-tx_{k-1}}{w-x_{k-1}}
\frac{g(x)}{w-x_k}
\right)
\\
&=
t\times
\frac{w-tx_1}{w-x_1}
\cdots
\frac{w-tx_{k-2}}{w-x_{k-2}}
\frac{g(x)}{(w-x_{k-1})(w-x_k)}.
\nonumber
\end{align}
Using this equation repeatedly, we easily find that
\begin{multline*}
\left(
\prod_{j \not\in k(\nu)}
\Delta_j
\right)
\left(
\prod_{i \in k(\nu)}
\b{x}_i
\right)
\mathcal{F}^{\lambda}_{\mu\nu}(x;t)
=
t^{(m+1)|\nu|}
\times
\frac{B_{\b\mu}(t)}{B_{\b\lambda}(t)}
\times
\\
\oint_{w_1}
\cdots
\oint_{w_n}
\prod_{1 \leq i<j \leq n}
\frac{
(w_j - w_i)
}
{
(w_j-t w_i)
}
\prod_{i=1}^{n}
\frac{
\prod_{1\leq j \leq k_{i-1}}
(w_i - t x_j)
}
{
\prod_{1 \leq j \leq k_i}
(w_i - x_j)
}
\prod_{i=1}^{n}
w_i^{M}
Q_{\b\lambda/\b\mu}(\b{w}_1,\dots,\b{w}_n;t).
\end{multline*}
where we have defined $k_0 = 0$. Since we have acted $|\nu|$ times with divided-difference operators on a polynomial of total degree $|\nu|$, it follows that the above expression is in fact constant with respect to $x_1,\dots,x_{n+|\nu|}$ (which is by no means apparent from the integral, taken at face value). We are therefore at liberty to take any limit of the variables that we choose, since our expression is independent of them.

\subsubsection{Homogeneous limit}

Taking the limit $x_i \rightarrow 0$ for all $1 \leq i \leq n+|\nu|$, we obtain
\begin{align*}
t^{-(m+1)|\nu|}
\times
\frac{B_{\b\lambda}(t)}{B_{\b\mu}(t)}
&\times
\lim_{x \rightarrow 0}
\left(
\prod_{j \not\in k(\nu)}
\Delta_j
\right)
\left(
\prod_{i \in k(\nu)}
\b{x}_i
\right)
\mathcal{F}^{\lambda}_{\mu\nu}(x;t)
\\
&=
\oint_{w_1}
\cdots
\oint_{w_n}
\prod_{1 \leq i<j \leq n}
\frac{
(w_j - w_i)
}
{
(w_j-t w_i)
}
\prod_{i=1}^{n}
w_i^{M+k_{i-1}-k_i}
Q_{\b\lambda/\b\mu}(\b{w}_1,\dots,\b{w}_n;t)
\\
&=
{\rm Coeff}
\left[
\prod_{1 \leq i<j \leq n}
\frac{(w_j-w_i)}{(w_j-tw_i)}
Q_{\b\lambda/\b\mu}(\b{w}_1,\dots,\b{w}_n;t),
\prod_{i=1}^{n}
w_i^{k_i-k_{i-1}-M-1}
\right]
\\
&=
{\rm Coeff}
\left[
\prod_{1 \leq i<j \leq n}
\frac{(z_i-z_j)}{(z_i-tz_j)}
Q_{\b\lambda/\b\mu}(z_1,\dots,z_n;t),
\prod_{i=1}^{n}
z_i^{M+1+k_{i-1}-k_i}
\right].
\end{align*}
It remains only to observe that $M+1+k_{i-1}-k_i = M-\nu_{n-i+1} \equiv \b\nu_i$ for all $1 \leq i \leq n$, which recovers the constant term expression \eqref{CT-Hall} for 
$b_{\b\nu}(t) f^{\b\lambda}_{\b\mu\b\nu}(t)$. This completes the proof of equation \eqref{main-formula}, and Theorem \ref{thm-hall}.

\subsection{Combinatorial interpretation of equation \eqref{main-formula}}

Equation \eqref{main-formula} is an explicit formula for the Hall polynomials, although it is complicated by the presence of divided-difference operators. Our aim in this section is to find a true combinatorial formula, that does not depend on auxiliary variables and operators which act on them. We begin by remarking that the partition function 
$\mathcal{F}^{\lambda}_{\mu\nu}(x;t)$ is easily decomposed into its monomials. Indeed, to calculate 
$\text{Coeff}[\mathcal{F}^{\lambda}_{\mu\nu}(x;t), \prod_{i=1}^{n+|\nu|} x_i^{p_i}]$ for any set of exponents $\{p_1,\dots,p_{n+|\nu|}\}$, one sums over all configurations in $\mathcal{F}^{\lambda}_{\mu\nu}(x;t)$ which have exactly $p_i$ vacant left edges in the $i$-th row for all $1 \leq i \leq n+|\nu|$. Hence the combinatorial implication of \eqref{main-formula} will become apparent by studying which monomials in $(x_1,\dots,x_{n+|\nu|})$ survive under the action of the divided-difference operators.

Before giving our combinatorial expression for the Hall polynomials, we require a number of simple definitions.

\begin{defn}{\rm
\label{def-dipole}
Consider a rectangular grid $G$ of light and dark squares, \begin{tikzpicture}[scale=0.5,baseline=0.1cm] \blight{0}{0} \end{tikzpicture} and \begin{tikzpicture}[scale=0.5,baseline=0.1cm] \bdark{0}{0} \end{tikzpicture}\ . A {\it dipole tiling} of $G$ is a placement of horizontal dipoles  
\begin{tikzpicture}[scale=0.5,baseline=0.1cm] 
\filldraw[fill=black!20!green,draw=black!20!green] (0.5,0.5) circle (0.25cm); 
\draw[gline] (0.5,0.5) -- (2,0.5); 
\filldraw[fill=black!20!green,draw=black!20!green] (2,0.5) circle (0.25cm);
\node[text centered] at (0.5,0.5) {\color{white} $+$};  
\node[text centered] at (2,0.5) {\color{white} $-$};
\end{tikzpicture} 
and
\begin{tikzpicture}[scale=0.5,baseline=0.1cm] 
\filldraw[fill=black,draw=black] (0.5,0.5) circle (0.25cm); 
\draw[bline] (0.5,0.5) -- (2,0.5); 
\filldraw[fill=black,draw=black] (2,0.5) circle (0.25cm);
\node[text centered] at (0.5,0.5) {\color{white} $+$};  
\node[text centered] at (2,0.5) {\color{white} $-$};
\end{tikzpicture}
of any length on $G$, such that green (black) dipoles begin and end on light (dark) squares, and dipoles cannot overlap.

}
\end{defn}

\begin{defn}
\label{def-Npuzz}
{\rm
Fix a non-negative integer $N \geq 0$. An $N$--{\it puzzle} is a dipole tiling of the grid
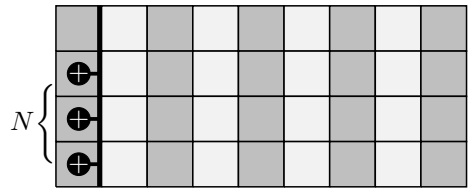
\begin{align*}
\begin{tikzpicture}[scale=0.6]
\node at (-0.5,5.4) {$N \Bigg\{$ };
\plusb{0}{4}
\plusb{0}{5}
\plusb{0}{6}
\bdark{0}{7}
\blight{1}{7}
\bdark{2}{7}
\blight{3}{7}
\bdark{4}{7}
\blight{5}{7}
\bdark{6}{7}
\blight{7}{7}
\bdark{8}{7}
\blight{1}{6}
\bdark{2}{6}
\blight{3}{6}
\bdark{4}{6}
\blight{5}{6}
\bdark{6}{6}
\blight{7}{6}
\bdark{8}{6}
\blight{1}{5}
\bdark{2}{5}
\blight{3}{5}
\bdark{4}{5}
\blight{5}{5}
\bdark{6}{5}
\blight{7}{5}
\bdark{8}{5}
\blight{1}{4}
\bdark{2}{4}
\blight{3}{4}
\bdark{4}{4}
\blight{5}{4}
\bdark{6}{4}
\blight{7}{4}
\bdark{8}{4}
\draw[ultra thick] (0.95,4) -- (0.95,8);
\end{tikzpicture}
\end{align*}
in which the left column is frozen as indicated: its top square must be unoccupied, \begin{tikzpicture}[scale=0.5,baseline=0.1cm] \bdark{0}{0} \end{tikzpicture}\ , and the remaining $N$ squares are required to be of the form \begin{tikzpicture}[scale=0.5,baseline=0.1cm] \plusb{0}{0} \end{tikzpicture}\ . Numbering the rows from top to bottom by $\{0,1,\dots,N\}$, let $r_i$ denote the total number of tiles of the form \begin{tikzpicture}[scale=0.5,baseline=0.1cm] \blight{0}{0} \end{tikzpicture} or \begin{tikzpicture}[scale=0.5,baseline=0.1cm] \plusg{0}{0} \end{tikzpicture} in the $i$-th row. Then we require that $(r_1,\dots,r_N)$ is a partition and 
$\sum_{i=0}^{N} r_i = N+1$. The {\it length} of the $N$--puzzle is the total number of rows $i$ for which $r_i > 0$.

}
\end{defn}

\begin{defn}{\rm
Let $\lambda,\mu,\nu$ be three partitions. A {\it $\nu$--puzzle} with frame $(\mu,\lambda)$ is a dipole tiling of the form
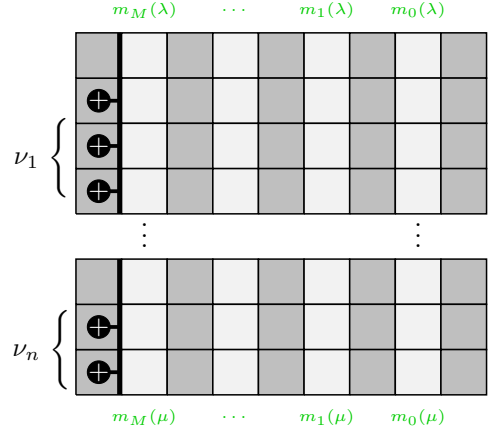
\begin{align*}
\begin{tikzpicture}[scale=0.6]
\node at (-0.75,9.25) {$\nu_1 \  \Bigg\{$ };
\node at (-0.75,5) {$\nu_n \  \Bigg\{$ };
\plusb{0}{4}
\plusb{0}{5}
\bdark{0}{6}
\plusb{0}{8}
\plusb{0}{9}
\plusb{0}{10}
\bdark{0}{11}
\blight{1}{11}
\bdark{2}{11}
\blight{3}{11}
\bdark{4}{11}
\blight{5}{11}
\bdark{6}{11}
\blight{7}{11}
\bdark{8}{11}
\blight{1}{10}
\bdark{2}{10}
\blight{3}{10}
\bdark{4}{10}
\blight{5}{10}
\bdark{6}{10}
\blight{7}{10}
\bdark{8}{10}
\blight{1}{9}
\bdark{2}{9}
\blight{3}{9}
\bdark{4}{9}
\blight{5}{9}
\bdark{6}{9}
\blight{7}{9}
\bdark{8}{9}
\blight{1}{8}
\bdark{2}{8}
\blight{3}{8}
\bdark{4}{8}
\blight{5}{8}
\bdark{6}{8}
\blight{7}{8}
\bdark{8}{8}
\node at (1.5,7.7) {$\vdots$};
\node at (7.5,7.7) {$\vdots$};
\blight{1}{6}
\bdark{2}{6}
\blight{3}{6}
\bdark{4}{6}
\blight{5}{6}
\bdark{6}{6}
\blight{7}{6}
\bdark{8}{6}
\blight{1}{5}
\bdark{2}{5}
\blight{3}{5}
\bdark{4}{5}
\blight{5}{5}
\bdark{6}{5}
\blight{7}{5}
\bdark{8}{5}
\blight{1}{4}
\bdark{2}{4}
\blight{3}{4}
\bdark{4}{4}
\blight{5}{4}
\bdark{6}{4}
\blight{7}{4}
\bdark{8}{4}
\node[text centered] at (1.5,12.5) {\gr \tiny $m_M(\lambda)$};
\node[text centered] at (3.5,12.5) {\gr \tiny $\cdots$};
\node[text centered] at (5.5,12.5) {\gr \tiny $m_1(\lambda)$};
\node[text centered] at (7.5,12.5) {\gr \tiny $m_0(\lambda)$};
\node[text centered] at (1.5,3.5) {\gr \tiny $m_M(\mu)$};
\node[text centered] at (3.5,3.5) {\gr \tiny $\cdots$};
\node[text centered] at (5.5,3.5) {\gr \tiny $m_1(\mu)$};
\node[text centered] at (7.5,3.5) {\gr \tiny $m_0(\mu)$};
\draw[ultra thick] (0.95,4) -- (0.95,7); \draw[ultra thick] (0.95,8) -- (0.95,12);
\end{tikzpicture}
\end{align*}
obtained by vertically concatenating $\nu_i$--puzzles for $1 \leq i \leq n$, such that the total charge of the $j$-th light column from the right is $m_j(\mu) - m_j(\lambda)$, and all internal dark columns have total charge 0. Necessarily, the left and rightmost dark columns will have total charge $|\nu|$ and $-|\nu|$, respectively. The length of the $\nu$--puzzle is the sum of the lengths of its constituent $\nu_i$--puzzles.

}
\end{defn}

An example of a (2,1,1)--puzzle with frame $((3,3,1,1),(2,1,0,0))$ is given below:
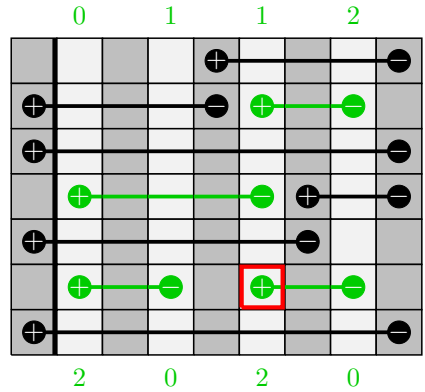
\begin{align*}
\begin{tikzpicture}[scale=0.6]
\node[text centered] at (1.5,8.5) {\gr 0};
\node[text centered] at (3.5,8.5) {\gr 1};
\node[text centered] at (5.5,8.5) {\gr 1};
\node[text centered] at (7.5,8.5) {\gr 2};
%
\node[text centered] at (1.5,0.5) {\gr 2};
\node[text centered] at (3.5,0.5) {\gr 0};
\node[text centered] at (5.5,0.5) {\gr 2};
\node[text centered] at (7.5,0.5) {\gr 0};
\bdark{0}{7};
\plusb{0}{6};
\plusb{0}{5};
\bdark{0}{4};
\plusb{0}{3};
\bdark{0}{2};
\plusb{0}{1};
\blight{1}{7}
\bdark{2}{7}
\blight{3}{7}
\plusb{4}{7}
\flatB{5}{7}
\flatb{6}{7}
\flatB{7}{7}
\minb{8}{7}
\flatB{1}{6}
\flatb{2}{6}
\flatB{3}{6}
\minb{4}{6}
\plusg{5}{6}
\flatG{6}{6}
\ming{7}{6}
\bdark{8}{6}
\flatB{1}{5}
\flatb{2}{5}
\flatB{3}{5}
\flatb{4}{5}
\flatB{5}{5}
\flatb{6}{5}
\flatB{7}{5}
\minb{8}{5}
\plusg{1}{4}
\flatG{2}{4}
\flatg{3}{4}
\flatG{4}{4}
\ming{5}{4}
\plusb{6}{4}
\flatB{7}{4}
\minb{8}{4}
\flatB{1}{3}
\flatb{2}{3}
\flatB{3}{3}
\flatb{4}{3}
\flatB{5}{3}
\minb{6}{3}
\blight{7}{3}
\bdark{8}{3}
\plusg{1}{2}
\flatG{2}{2}
\ming{3}{2}
\bdark{4}{2}
\plusg{5}{2}
\flatG{6}{2}
\ming{7}{2}
\bdark{8}{2}
\flatB{1}{1}
\flatb{2}{1}
\flatB{3}{1}
\flatb{4}{1}
\flatB{5}{1}
\flatb{6}{1}
\flatB{7}{1}
\minb{8}{1}
\draw[ultra thick] (0.95,1) -- (0.95,8);
\draw[ultra thick,red] (5.05,2.05) -- (5.95,2.05) -- (5.95,2.95) -- (5.05,2.95) -- (5.05,2.05);
\end{tikzpicture}
\end{align*}
We let $\mathbb{P}_{\nu}(\mu,\lambda)$ denote the set of all $\nu$--puzzles with frame $(\mu,\lambda)$.

\begin{defn}{\rm 
The {\it cumulative charge} of a tile is the sum (reading from top downward) of all charges in its column, up to and including the tile itself, plus the occupation number at the top of the column. For example, in the previous puzzle, the tile bordered in red has cumulative charge $1+1-1+1=2$.
}
\end{defn}   

\begin{defn}\label{weight}
{\rm
The Boltzmann weight of a puzzle $P$, denoted $W(P)$, is the product of local weights assigned to each tile in the lattice, excluding those in the trivial left column (which receive no weight). We indicate the weight of each tile below:
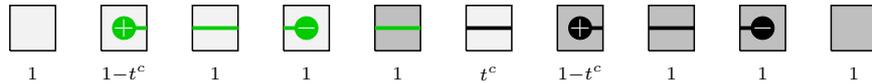
\begin{align*}
\begin{tikzpicture}[scale=0.6]
\blight{1}{1};
\plusg{3}{1};
\flatg{5}{1};
\ming{7}{1};
\flatG{9}{1};
\flatB{11}{1};
\plusb{13}{1};
\flatb{15}{1};
\minb{17}{1};
\bdark{19}{1};
\node[text centered] at (1.5,0.5) {$\ss 1$};
\node[text centered] at (3.5,0.5) {$\ss 1-t^{c}$};
\node[text centered] at (5.5,0.5) {$\ss 1$};
\node[text centered] at (7.5,0.5) {$\ss 1$};
\node[text centered] at (9.5,0.5) {$\ss 1$};
\node[text centered] at (11.5,0.5) {$\ss t^c$};
\node[text centered] at (13.5,0.5) {$\ss 1-t^{c}$};
\node[text centered] at (15.5,0.5) {$\ss 1$};
\node[text centered] at (17.5,0.5) {$\ss 1$};
\node[text centered] at (19.5,0.5) {$\ss 1$};
\end{tikzpicture} 
\end{align*}
where $c$ denotes the cumulative charge of the tile. An immediate consequence of these weights is that if the cumulative charge of any tile is negative, the Boltzmann weight of the whole puzzle is necessarily zero. We can therefore assume that the cumulative charge of any tile is non-negative.
}
\end{defn}

\begin{thm}
\label{thm-hall-puzzle}
{\rm
Let $\lambda,\mu,\nu$ be three partitions with the same definitions as in Theorem \ref{thm-hall}, and $\b\lambda,\b\mu,\b\nu$ their complements. Then the Hall polynomial $f^{\b\lambda}_{\b\mu\b\nu}(t)$ is obtained by summing over all $\nu$--puzzles with frame $(\mu,\lambda)$:
\begin{align}
\label{hall-comb}
f^{\b\lambda}_{\b\mu\b\nu}(t)
=
t^{-(m+1)|\nu|}
\times
\frac{B_{\b\lambda}(t)}
{B_{\b\mu}(t) b_{\b\nu}(t)}
\times
\sum_{P \in \mathbb{P}_{\nu}(\mu,\lambda)}
(-1)^{L(P)}
W(P),
\end{align}
where $L(P)$ denotes the length of $P$.

\begin{proof}
The proof is not difficult, so we only sketch its key points. We start from equation \eqref{main-formula}, and analyse the action of the divided-difference operators on $\mathcal{F}^{\lambda}_{\mu\nu}(t)$. For any four partitions $\alpha,\beta,\gamma,\delta$ and integer 
$N \geq 0$, define
\begin{align}
\label{build}
\rho^{(N)}_{\alpha / \gamma, \beta / \delta}(t)
=
\lim_{x \rightarrow 0}
\left(
\Delta_1 \dots \Delta_N
\b{x}_{N+1}
\bra{{\gr\cev\alpha}} \otimes \bra{{\cev\beta}}
\Tb(x_1) \dots \Tb(x_N) \Te(x_{N+1})
\ket{{\gr \cev\gamma}} \otimes \ket{{\cev \delta}}
\right),
\end{align}
where all $x$ variables are sent to zero, after having acted with the divided-difference operators $\Delta_1 \dots \Delta_N$. It follows that
\begin{align}
\label{double-tableau}
f^{\b\lambda}_{\b\mu\b\nu}(t)
=
t^{-(m+1)|\nu|}
\times
\frac{B_{\b\lambda}(t)}
{B_{\b\mu}(t) b_{\b\nu}(t)}
\times
\sum_{{\bm \alpha}_{\lambda,\mu}}
\sum_{{\bm \beta}_{\nu}}
\prod_{i=1}^{n}
\rho^{(\nu_i)}_{\alpha_i / \alpha_{i-1}, \beta_i / \beta_{i-1}}(t),
\end{align}
where the sum is taken over all possible sequences of partitions 
\begin{align*}
{\bm \alpha}_{\lambda,\mu}
:=
\{
\lambda \equiv \alpha_0 \subseteq \alpha_1 \subseteq \cdots \subseteq \alpha_n \equiv \mu
\},
\qquad
{\bm \beta}_{\nu}
:= 
\{
\underbrace{(0,\dots,0)}_{|\nu|} \equiv \beta_0,\beta_1,\dots,\beta_n \equiv \emptyset
\}
\end{align*}
such that $\ell\left(\alpha_{i-1}\right) = \ell\left(\alpha_{i}\right)$ and $\ell\left(\beta_{i-1}\right) = \ell\left(\beta_{i}\right) + \nu_{i}$ for all $1 \leq i \leq n$. 

Equation \eqref{double-tableau} is, in some sense, a ``double-tableau'' formula for the Hall polynomials, just as \eqref{skew-ssyt} expresses skew Hall--Littlewood polynomials as a sum over semi-standard Young tableaux. To translate it into the formula \eqref{hall-comb}, we only need to find the combinatorial interpretation of \eqref{build}, which are the building blocks. Recalling the lattice representation of $\mathcal{F}^{\lambda}_{\mu\nu}(t)$, we make the trivial bijection 
\begin{tikzpicture}[scale=0.5,baseline=0.1cm] 
\draw[gline] (0.5,0) -- (0.5,0.5) -- (2,0.5) -- (2,1); 
\end{tikzpicture} 
$\leftrightarrow$
\begin{tikzpicture}[scale=0.5,baseline=0.1cm] 
\filldraw[fill=black!20!green,draw=black!20!green] (0.5,0.5) circle (0.25cm); 
\draw[gline] (0.5,0.5) -- (2,0.5); 
\filldraw[fill=black!20!green,draw=black!20!green] (2,0.5) circle (0.25cm);
\node[text centered] at (0.5,0.5) {\color{white} $+$};  
\node[text centered] at (2,0.5) {\color{white} $-$};
\end{tikzpicture} and 
\begin{tikzpicture}[scale=0.5,baseline=0.1cm] 
\draw[bline] (0.5,0) -- (0.5,0.5) -- (2,0.5) -- (2,1); 
\end{tikzpicture} 
$\leftrightarrow$
\begin{tikzpicture}[scale=0.5,baseline=0.1cm] 
\filldraw[fill=black!,draw=black] (0.5,0.5) circle (0.25cm); 
\draw[bline] (0.5,0.5) -- (2,0.5); 
\filldraw[fill=black,draw=black] (2,0.5) circle (0.25cm);
\node[text centered] at (0.5,0.5) {\color{white} $+$};  
\node[text centered] at (2,0.5) {\color{white} $-$};
\end{tikzpicture} 
between lattice paths and dipoles. One readily finds that
\begin{align}
\label{build-puzz}
\rho^{(N)}_{\alpha / \gamma, \beta / \delta}(t)
=
\sum_{P}
(-1)^{L(P)} W(P),
\qquad
P
=
\begin{tikzpicture}[scale=0.6,baseline=3.5cm]
\node at (-0.7,5.4) {$N \Bigg\{$ };
\plusb{0}{4}
\plusb{0}{5}
\plusb{0}{6}
\bdark{0}{7}
\blight{1}{7}
\bdark{2}{7}
\blight{3}{7}
\bdark{4}{7}
\blight{5}{7}
\bdark{6}{7}
\blight{7}{7}
\bdark{8}{7}
\blight{1}{6}
\bdark{2}{6}
\blight{3}{6}
\bdark{4}{6}
\blight{5}{6}
\bdark{6}{6}
\blight{7}{6}
\bdark{8}{6}
\blight{1}{5}
\bdark{2}{5}
\blight{3}{5}
\bdark{4}{5}
\blight{5}{5}
\bdark{6}{5}
\blight{7}{5}
\bdark{8}{5}
\blight{1}{4}
\bdark{2}{4}
\blight{3}{4}
\bdark{4}{4}
\blight{5}{4}
\bdark{6}{4}
\blight{7}{4}
\bdark{8}{4}
\draw[ultra thick] (0.95,4) -- (0.95,8);
\node[text centered] at (3.5,8.5) {\gr \tiny $\cdots$};
\node[text centered] at (5.5,8.5) {\gr \tiny $m_1(\gamma)$};
\node[text centered] at (7.5,8.5) {\gr \tiny $m_0(\gamma)$};
\node[text centered] at (4.5,9) {\tiny $\cdots$};
\node[text centered] at (6.5,9) {\tiny $m_1(\delta)$};
\node[text centered] at (8.5,9) {\tiny $m_0(\delta)$};
\node[text centered] at (3.5,3.5) {\gr \tiny $\cdots$};
\node[text centered] at (5.5,3.5) {\gr \tiny $m_1(\alpha)$};
\node[text centered] at (7.5,3.5) {\gr \tiny $m_0(\alpha)$};
\node[text centered] at (4.5,3) {\tiny $\cdots$};
\node[text centered] at (6.5,3) {\tiny $m_1(\beta)$};
\node[text centered] at (8.5,3) {\tiny $m_0(\beta)$};
\end{tikzpicture}
\end{align}
where the sum is over all $N$--puzzles as shown on the right, with bottom and top boundaries corresponding to the partition states $\bra{{\gr\cev{\alpha}}} \otimes \bra{{\cev{\beta}}}$ and 
$\ket{{\gr \cev{\gamma}}} \otimes \ket{{\cev{\delta}}}$, and where the Boltzmann weights are given by Definition \ref{weight}. The main subtlety in going from \eqref{build} to \eqref{build-puzz} is to calculate the action of the divided-difference operators $\Delta_1 \dots \Delta_{N}$ on $\b{x}_{N+1}
\bra{{\gr\cev{\alpha}}} \otimes \bra{{\cev{\beta}}}
\Tb(x_1) \dots \Tb(x_N) \Te(x_{N+1})
\ket{{\gr \cev{\gamma}}} \otimes \ket{{\cev{\delta}}}$. It is not hard to show that it gives rise to the partition constraint described in Definition \ref{def-Npuzz}, while any minus signs acquired from the action are maintained through the factor $(-1)^{L(P)}$.

Finally, \eqref{double-tableau} calls for the concatenation of $\nu_i$--puzzles for all $1 \leq i \leq n$, and \eqref{hall-comb} follows immediately.

\end{proof}
}
\end{thm}

\begin{rmk}{\rm
In the case where $\nu$ consists of just one part, $N>0$, the sum \eqref{hall-comb} factorizes\footnote{In saying that a polynomial ``factorizes'' we mean that all its non-zero roots are roots of unity.}. This is necessarily the case, because of the Pieri rule
\begin{align*}
(1-t)
P_{\mu} P_{N}
=
\sum_{\lambda: \lambda/\mu \in \mathfrak{h}_N}
\varphi_{\lambda/\mu}(t)
P_{\lambda},
\qquad
\varphi_{\lambda/\mu}(t)
=
\frac{b_{\lambda}(t)}
{b_{\mu}(t)}
\psi_{\lambda/\mu}(t),
\end{align*}
from which we conclude that $f^{\lambda}_{\mu N}(t) = \varphi_{\lambda/\mu}(t)/(1-t)$ if $\lambda/\mu \in \mathfrak{h}_N$, $f^{\lambda}_{\mu N}(t) = 0$ otherwise. However as one progresses to examples where $\nu$ consists of two parts, and beyond, such factorization does not occur in any systematic way.

In Appendix \ref{A}, we give two explicit examples of the combinatorial rule \eqref{hall-comb}. The first is an example where $\nu$ consists of one part, and gives rise to a factorized answer; the second is a more complicated example in which $\nu$ has two parts.
}
\end{rmk}

\section{\texorpdfstring{$t$}{t}-Schur polynomials from a rank-one model of fermions}

\subsection{Fermionic Fock space $\mathcal{F}$}

Consider a semi-infinite one-dimensional lattice, with sites labelled by non-negative integers. In a finite configuration of this lattice, each site $i \geq 0$ is either occupied by a particle (occupation number $a_i = 1$) or unoccupied ($a_i = 0$), and there exists $M \in \mathbb{N}$ such that $a_k = 0$ for all $k \geq M$. Such configurations are equivalent to Maya diagrams\footnote{Conventionally, Maya diagrams are infinite in both directions, with a particle or hole at every site of the integer lattice. They satisfy the finiteness condition that sufficiently far to the left/right only particles/holes occur. In this work we label our integer lattice in such a way that all negative sites are occupied by a particle, so we shall suppress the negative half of our Maya diagrams, since it contains no information.}. The Fermionic Fock space $\mathcal{F}$ is obtained by taking linear combinations of all possible Maya diagrams:
\begin{align*}
\mathcal{F}
=
{\rm Span} 
\left\{
\ket{a_0}_0
\otimes
\ket{a_1}_1
\otimes
\cdots
\right\},
\qquad
0 \leq a_i \leq 1,\ \forall\ i \geq 0.
\end{align*}
The dual vector space is defined as
\begin{align*}
\mathcal{F}^{*} = 
{\rm Span} 
\left\{
\bra{a_0}_0
\otimes
\bra{a_1}_1
\otimes
\cdots
\right\},
\qquad
0 \leq a_i \leq 1,\ \forall\ i \geq 0,
\end{align*}
with its action on $\mathcal{F}$ given by
\begin{align*}
\langle a | b \rangle = \prod_{i = 0}^{\infty} \delta_{a_i,b_i},
\quad
\forall\
\bra{a} = \bigotimes_{k = 0}^{\infty} \bra{a_k}_k,
\quad
\ket{b} = \bigotimes_{k = 0}^{\infty} \ket{b_k}_k.
\end{align*}

\subsection{Mapping partitions to states in \texorpdfstring{$\mathcal{F}$}{F}}

Now we consider the well-known mapping between partitions and Maya diagrams. Given a partition 
$\lambda = (\lambda_1,\dots,\lambda_{\ell})$, we define the shifted (strict) partition $\tilde\lambda$ with parts $\tilde\lambda_i = \lambda_i+\ell-i$. We associate a corresponding state $\ket{{\re\vec \lambda}}$ in $\mathcal{F}$:
\begin{align}
\label{maya}
\ket{{\re\vec \lambda}}
=
\bigotimes_{k =0}^{\infty}
\ket{a_k(\lambda)}_k,
\quad\quad
a_k(\lambda)
=
\left\{
\begin{array}{ll}
1, & k \in \tilde{\lambda},
\\ \\
0, & \text{otherwise}.
\end{array}
\right.
\end{align}
\begin{ex}{\rm 
Let $\lambda = (5,3,3,1,0)$. Then $\tilde\lambda = (9,6,5,2,0)$, and we find that
\begin{align*}
\ket{{\re\vec \lambda}}
=
\ket{{\re 1}}_{0}
\otimes
\ket{{\re 0}}_{1}
\otimes
\ket{{\re 1}}_{2}
\otimes
\ket{{\re 0}}_{3}
\otimes
\ket{{\re 0}}_{4}
\otimes
\ket{{\re 1}}_{5}
\otimes
\ket{{\re 1}}_{6}
\otimes
\ket{{\re 0}}_{7}
\otimes
\ket{{\re 0}}_{8}
\otimes
\ket{{\re 1}}_{9}
\otimes
\ket{{\re 0}}_{10}
\otimes
\cdots.
\end{align*}
This mapping has a simple pictorial explanation. Starting from the Young diagram of $\lambda$ rotated clockwise by 45\degree, we assign particles to every edge with positive slope, and holes to every edge with negative slope. The resulting array of particles/holes can then be projected onto the the integer lattice:
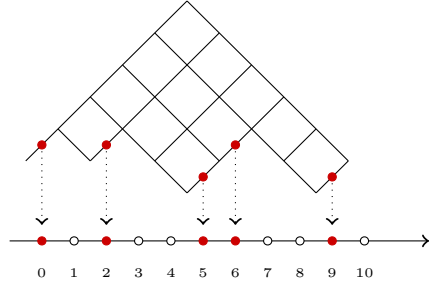
\begin{center}
\begin{tikzpicture}[scale=0.6,rotate=-45]
\foreach\x in {0,1,...,10}{
\node[text centered] at (2+0.5*\x,-2.5+0.5*\x) {$\sss \x$};
}
\draw[arrow=1] (1,-2.5) -- (7.5,4);
\draw[dotted,arrow=1] (0,-0.5) -- (1.25,-1.75);
\draw[dotted,arrow=1] (1,0.5) -- (2.25,-0.75);
\draw[dotted,arrow=1] (3,1.5) -- (3.75,0.75);
\draw[dotted,arrow=1] (3,2.5) -- (4.25,1.25);
\draw[dotted,arrow=1] (5,3.5) -- (5.75,2.75);
\draw (0,4) -- (5,4);
\draw (0,3) -- (5,3);
\draw (0,2) -- (3,2);
\draw (0,1) -- (3,1);
\draw (0,0) -- (1,0);
\draw (0,-1) -- (0,4);
\draw (1,0) -- (1,4);
\draw (2,1) -- (2,4);
\draw (3,1) -- (3,4);
\draw (4,3) -- (4,4);
\draw (5,3) -- (5,4);
\rbull{0}{-0.5}{0.09};
\rbull{1}{0.5}{0.09};
\rbull{3}{1.5}{0.09};\rbull{3}{2.5}{0.09};
\rbull{5}{3.5}{0.09};
\rbull{1.5}{-2}{0.09};
\ebull{2}{-1.5}{0.09};
\rbull{2.5}{-1}{0.09};
\ebull{3}{-0.5}{0.09};
\ebull{3.5}{0}{0.09};
\rbull{4}{0.5}{0.09};
\rbull{4.5}{1}{0.09};
\ebull{5}{1.5}{0.09};
\ebull{5.5}{2}{0.09};
\rbull{6}{2.5}{0.09};
\ebull{6.5}{3}{0.09};
\end{tikzpicture}
\end{center}
From here the Maya diagram $\ket{{\re\vec \lambda}}$ is obtained by reading off the resulting binary string, where we make the identifications $\rpart[0.15] \equiv {\re 1}$ and $\hole[0.15] \equiv {\re 0}$.
}
\end{ex}

\subsection{Reverse Maya diagrams}

Let $\ket{{\re\vec \lambda}}$ be the Maya diagram as defined by \eqref{maya}, corresponding with the length-$\ell$ partition $\lambda$, whose largest part is $L$. The reverse Maya diagram $\ket{{\re\cev \lambda}} \in \mathcal{F}$ is defined as
\begin{align*}
\ket{{\re\cev \lambda}}
=
\bigotimes_{k = 0}^{\infty}
\ket{\b{a}_k(\lambda)}_k,
\quad\quad
\b{a}_k(\lambda)
=
a_{\ell+L-1-k}(\lambda),\ 
\forall\ 
0 \leq k \leq \ell + L-1,
\quad\quad
\b{a}_k(\lambda) = 0,\
\text{otherwise}.
\end{align*}
The pictorial interpretation of a reverse Maya diagram is shown below. In analogy with Section \ref{sec:rev}, one first rotates the Young diagram anti-clockwise by 135\degree, before projecting onto the integer lattice:

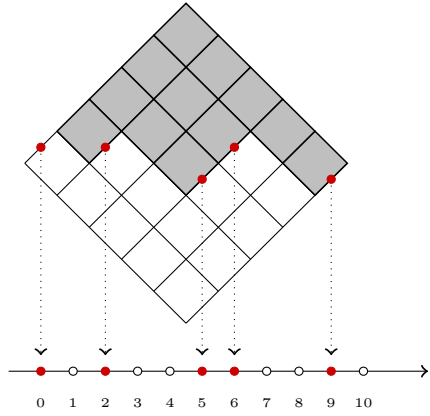
\begin{center}
\begin{tikzpicture}[scale=0.6,rotate=-45]
\foreach\x in {0,1,...,10}{
\node[text centered] at (4+0.5*\x,-4.5+0.5*\x) {$\sss \x$};
}
\draw[arrow=1] (3,-4.5) -- (9.5,2);
\draw[dotted,arrow=1] (0,-0.5) -- (3.25,-3.75);
\draw[dotted,arrow=1] (1,0.5) -- (4.25,-2.75);
\draw[dotted,arrow=1] (3,1.5) -- (5.75,-1.25);
\draw[dotted,arrow=1] (3,2.5) -- (6.25,-0.75);
\draw[dotted,arrow=1] (5,3.5) -- (7.75,0.75);
\bdark{0}{3}; \bdark{1}{3}; \bdark{2}{3}; \bdark{3}{3}; \bdark{4}{3};
\bdark{0}{2}; \bdark{1}{2}; \bdark{2}{2};
\bdark{0}{1}; \bdark{1}{1}; \bdark{2}{1};
\bdark{0}{0}; 
\draw (3,3) -- (5,3); 
\draw (3,2) -- (5,2);
\draw (1,1) -- (5,1);
\draw (0,0) -- (5,0);
\draw (0,-1) -- (5,-1);
\draw (0,-1) -- (0,0);
\draw (1,-1) -- (1,1); 
\draw (2,-1) -- (2,1);
\draw (3,-1) -- (3,3);
\draw (4,-1) -- (4,3);
\draw (5,-1) -- (5,4);
\rbull{0}{-0.5}{0.09};
\rbull{1}{0.5}{0.09};
\rbull{3}{1.5}{0.09};\rbull{3}{2.5}{0.09};
\rbull{5}{3.5}{0.09};
\rbull{3.5}{-4}{0.09};
\ebull{4}{-3.5}{0.09};
\rbull{4.5}{-3}{0.09};
\ebull{5}{-2.5}{0.09};
\ebull{5.5}{-2}{0.09};
\rbull{6}{-1.5}{0.09};
\rbull{6.5}{-1}{0.09};
\ebull{7}{-0.5}{0.09};
\ebull{7.5}{0}{0.09};
\rbull{8}{0.5}{0.09};
\ebull{8.5}{1}{0.09};
\end{tikzpicture}
\end{center}
It is also clear from this picture that in general $\ket{{\re\cev \lambda}} = \ket{{\re\vec {\b\lambda}}}$, precisely as we saw in Proposition \ref{rev-prop}.

\subsection{A rank-one model of fermions}
\label{model-f}

Introduce fermionic creation and annihilation operators ${\re \psid}$ and ${\re \psi}$, and the 
particle-number operator ${\re N}$, which are given explicitly by $2 \times 2$ matrices:
\begin{align*}
{\re \psid}
=
\begin{pmatrix}
0 & 0
\\
1-t & 0
\end{pmatrix},
\quad
{\re \psi}
=
\begin{pmatrix}
0 & 1
\\
0 & 0
\end{pmatrix},
\quad
{\re N}
=
\begin{pmatrix}
0 & 0
\\
0 & 1
\end{pmatrix},
\end{align*}
and let ${\re \mathfrak{f}}$ denote the algebra generated by these and the identity matrix. They have a natural action on a single site of a Maya diagram, obtained by making the identifications $\ket{{\re 0}} = \begin{pmatrix} 1 \\ 0 \end{pmatrix}$ and $\ket{{\re 1}} = \begin{pmatrix} 0 \\ 1\end{pmatrix}$. Namely, we have
\begin{align*}
{\re \psid} \ket{{\re 0}}
=
(1-t)
\ket{{\re 1}},
\quad
{\re \psid} \ket{{\re 1}}
=
0,
\quad
{\re \psi} \ket{{\re 0}}
=
0,
\quad
{\re \psi} \ket{{\re 1}}
=
\ket{{\re 0}},
\quad
{\re N} \ket{{\re 0}}
=
0,
\quad
{\re N} \ket{{\re 1}}
=
\ket{{\re 1}}.
\end{align*}
The $L$-matrix
\begin{align}
\label{Lmat-f}
L^{*}_a(x|{\re\mathfrak f})
=
\begin{pmatrix}
1 & x{\re \psid}
\\
{\re \psi} & x(-t)^{\re N}
\end{pmatrix}_a
=
\left(
\begin{array}{cc}
\begin{tikzpicture}[scale=0.6]
\rblank{0}{0}
\end{tikzpicture} 
& 
\begin{tikzpicture}[scale=0.6]
\rplus{0}{0}
\end{tikzpicture} 
\\
\begin{tikzpicture}[scale=0.6]
\rminus{0}{0}
\end{tikzpicture} 
& 
\begin{tikzpicture}[scale=0.6]
\rfull{0}{0}
\end{tikzpicture} 
\end{array}
\right)_a
\end{align}
satisfies the intertwining relation 
\begin{align*}
R_{ab}(y/x)
L^{*}_a(x|{\re\mathfrak f})
L^{*}_b(y|{\re\mathfrak f})
=
L^{*}_b(y|{\re\mathfrak f})
L^{*}_a(x|{\re\mathfrak f})
R_{ab}(y/x),
\end{align*}
with $R$-matrix given by
\begin{align}
\label{Rmat-f}
R_{ab}(z)
=
\left(
\begin{array}{cc|cc}
1-tz & 0 & 0 & 0
\\
0 & t(1-z) & 1-t & 0
\\
\hline
0 & (1-t)z & 1-z & 0
\\
0 & 0 & 0 & z-t
\end{array}
\right)_{ab}.
\end{align}
We construct a monodromy matrix by taking a product of $L$-matrices over all sites $i \geq 0$ of the lattice:
\begin{align*}
T^{*}_a(x) 
=
L^{*}_a(x|{\re\mathfrak f_0})
L^{*}_a(x|{\re\mathfrak f_1})
\cdots
=
\prod_{i=0}^{\infty}
L^{*}_a(x|{\re \mathfrak f_i})
:=
\begin{pmatrix}
A(x) & B(x)
\\
C(x) & D(x)
\end{pmatrix}_a.
\end{align*}
Once again, we will be subsequently interested in the operator $A(x) \in {\rm End}(\mathcal{F})$, which can be interpreted as the sum of all possible (semi-infinite) rows of the tiles \eqref{Lmat-f}, whose left-most tile is \begin{tikzpicture}[scale=0.5,baseline=0.1cm] \rblank{0}{0} \end{tikzpicture} or \begin{tikzpicture}[scale=0.5,baseline=0.1cm] \rplus{0}{0} \end{tikzpicture} and such that sufficiently far to the right only \begin{tikzpicture}[scale=0.5,baseline=0.1cm] \rblank{0}{0} \end{tikzpicture} occurs. These operators commute: $[A(x),A(y)] = 0$ for all $x,y$.

We point out that the $R$-matrix \eqref{Rmat-f} can be obtained from the Felderhof model (introduced in \cite{fel}, but see also equation (3.1) of \cite{deg-aku}) by a particular choice of the external fields. The model defined by the $L$-matrix \eqref{Lmat-f} is therefore inherently free-fermionic. Indeed, by writing the entries of \eqref{Lmat-f} explicitly, the $L$-matrix is manifestly a free-fermionic six-vertex model:
\begin{align*}
L^{*}_a(x|{\re\mathfrak f_i})
=
\left(
\begin{array}{cc|cc}
1 & 0 & 0 & 0
\\
0 & 1 & (1-t) x & 0
\\
\hline
0 & 1 & x & 0
\\
0 & 0 & 0 & -tx
\end{array}
\right)_{ai}.
\end{align*}

\subsection{\texorpdfstring{$t$}{t}-Schur polynomials \texorpdfstring{$S_{\lambda}(x;t)$}{S lambda(x;t)}}

Macdonald defined a $t$-generalization of Schur polynomials, which we call $t$-Schur polynomials, denoted $S_{\lambda}(x;t)$ \cite{mac}. They can be expressed in Jacobi--Trudi form:
\begin{align}
\label{JT}
S_{\lambda}(x;t)
=
\det_{1 \leq i,j \leq \ell}
\Big( q_{\lambda_i-i+j}(x;t) \Big),
\end{align}
where the matrix entries are polynomials given by the generating series
\begin{align*}
\sum_{k} q_k(x;t) y^k
=
\prod_{i}
\left(
\frac{1-t x_i y}{1-x_i y}
\right).
\end{align*}
The $t$-Schur polynomials $S_{\lambda}(x;t)$ degenerate to Schur polynomials $s_{\lambda}(x)$ at $t=0$, which is easily from the Jacobi--Trudi formula for the latter and the fact that $q_k(x;0) = h_k(x)$, the $k$-th complete symmetric function. Furthermore they satisfy the Cauchy identity
\begin{align}
\label{cauchy-Ss}
\sum_{\lambda}
S_{\lambda}(x;t)
s_{\lambda}(y)
=
\prod_{i,j}
\left(
\frac{1-t x_i y_j}{1-x_i y_j}
\right)
\end{align}
which in turn leads to the duality relation $\langle S_{\lambda}, s_{\mu} \rangle = 
\delta_{\lambda,\mu}$ with respect to the Hall inner product \cite{mac}. Since we wish to construct them via a transfer matrix approach, it is convenient to derive a branching rule for the $t$-Schur polynomials. 
\begin{lem}{\rm 
Let $x$ denote a set of variables $(x_1,\dots,x_n)$, and $z$ a single variable. The $t$-Schur polynomials satisfy the relation
\begin{align}
\label{branch-S}
S_{\lambda}(x,z;t)
=
\sum_{\mu: \lambda/\mu \in \mathfrak{v}}
\
\sum_{\nu: \mu/\nu \in \mathfrak{h}}
(-t)^{|\lambda-\mu|}
z^{|\lambda-\nu|}
S_{\nu}(x;t).
\end{align}
In terms of Young diagrams, we note that the conditions $\mu/\nu \in \mathfrak{h}$ and 
$\lambda/\mu \in \mathfrak{v}$ imply that the skew diagram $\lambda/\nu$ forms a {\it border strip}.
}
\end{lem}

\begin{proof}
Using the Cauchy identity \eqref{cauchy-Ss}, and the generating series expressions of the complete symmetric and elementary symmetric polynomials, we can write
\begin{align*}
\prod_{i}
\frac{1-ty_i z}{1-y_i z}
\prod_{i,j}
\frac{1-t x_i y_j}
{1-x_i y_j}
=
\sum_{k,l}
(-t)^l
z^{k+l}
h_k(y)
e_l(y)
\sum_{\nu}
S_{\nu}(x;t)
s_{\nu}(y).
\end{align*}
Thanks to the horizontal and vertical strip Pieri identities for the Schur polynomials, the previous equation becomes
\begin{align*}
\prod_{i}
\frac{1-ty_i z}{1-y_i z}
\prod_{i,j}
\frac{1-t x_i y_j}
{1-x_i y_j}
&=
\sum_{k,l}
(-t)^l
z^{k+l}
\sum_{\mu: \mu/\nu \in \mathfrak{h}_k}
\sum_{\lambda: \lambda/\mu \in \mathfrak{v}_l}
S_{\nu}(x;t)
s_{\lambda}(y)
\\
&=
\sum_{\lambda}
s_{\lambda}(y)
\sum_{\mu:\lambda/\mu \in \mathfrak{v}}
\sum_{\nu: \mu/\nu \in \mathfrak{h}}
z^{|\lambda-\nu|}
(-t)^{|\lambda-\mu|}
S_{\nu}(x;t).
\end{align*}
Finally, we note that the starting product itself is equal to 
$\sum_{\lambda} S_{\lambda}(x,z;t) s_{\lambda}(y)$, again by the Cauchy identity \eqref{cauchy-Ss}. Comparing coefficients $s_{\lambda}(y)$ on both sides of the resulting equation, we obtain \eqref{branch-S}.
\end{proof}

\subsection{Skew \texorpdfstring{$t$}{t}-Schur polynomials \texorpdfstring{$S_{\lambda/\mu}(x;t)$}{S lambda/mu(x;t)}}

Next, we define a skew version of the $t$-Schur polynomial. One natural way to do this\footnote{Generally speaking, a skew polynomial is normally defined via an orthogonality relation of the form \eqref{ortho-rel}. Later on, we will see that the definition \eqref{S-branch} is equivalent to the orthogonality relation $\langle S_{\lambda}, s_{\mu} P_{\nu} \rangle = \langle S_{\lambda/\mu}, P_{\nu} \rangle$.} is via the expansion of $S_{\lambda}(x,y;t)$ over the basis $S_{\mu}(y;t)$:
\begin{align}
\label{S-branch}
S_{\lambda}(x_1,\dots,x_n,y_1,\dots,y_m;t)
:=
\sum_{\mu}
S_{\lambda/\mu}(x_1,\dots,x_n;t)
S_{\mu}(y_1,\dots,y_m;t)
\end{align}
for any partition $\lambda$ and any two sets of variables $(x_1,\dots,x_n)$, $(y_1,\dots,y_m)$, and where the sum on the right hand side is taken over all partitions $\mu$. This allows $S_{\lambda/\mu}(x_1,\dots,x_n;t)$ to be built up recursively, once given $S_{\lambda/\mu}(z;t)$ in a single variable $z$. Comparing with \eqref{branch-S}, we see that
\begin{align}
\label{one-var-S}
S_{\lambda/\nu}(z;t)
=
z^{|\lambda-\nu|}
\sum_{\mu: \mu/\nu \in \mathfrak{h}, \lambda/\mu \in \mathfrak{v}}
(-t)^{|\lambda-\mu|}.
\end{align}

\subsection{Lattice expression for \texorpdfstring{$S_{\lambda/\mu}(x;t)$}{S lambda/mu(x;t)}}

\begin{lem}
\label{lem:S-exp}
{\rm 
The skew $t$-Schur polynomial $S_{\lambda/\mu}$ is given by the following expectation value:
\begin{align}
\label{lattice-S}
S_{\lambda/\mu}(x_1,\dots,x_n;t)
=
\bra{{\re\vec \mu}}
A(x_1)
\dots
A(x_n)
\ket{{\re\vec\lambda}}.
\end{align}
%
}
\end{lem}

\begin{proof}
Inserting a complete set of states in the right hand side of \eqref{lattice-S}, we have
\begin{align*}
\bra{{\re\vec \mu}}
A(x_1)
\dots
A(x_n)
\ket{{\re\vec \lambda}}
&=
\sum_{\nu}
\bra{{\re\vec \mu}}
A(x_1)
\dots
A(x_{n-1})
\ket{{\re\vec \nu}}
\bra{{\re\vec \nu}}
A(x_n)
\ket{{\re\vec \lambda}}.
\end{align*}
To complete the proof, it is sufficient to show that
\begin{align}
\label{onerow-skew-S}
\bra{{\re\vec \nu}}
A(z)
\ket{{\re\vec \lambda}}
=
S_{\lambda/\nu}(z;t).
\end{align}
We will illustrate this on the example $\lambda = (4,2,2,0)$, $\nu = (3,1,0,0)$, for which we have
\begin{align*}
\bra{{\re(0,0,1,3)}} A(z) \ket{{\re(0,2,2,4)}}
=
\begin{tikzpicture}[scale=0.8,baseline=0.3cm]
\rblank[1]{1}{0}[1]
\rplus[0]{2}{0}[1]
\rfull[0]{3}{0}[0]
\rfull[1]{4}{0}[1]
\rminus[1]{5}{0}[0]
\rblank[0]{6}{0}[0]
\rplus[0]{7}{0}[1]
\rminus[1]{8}{0}[0]
\rblank[0]{9}{0}[0]
\node[text centered] at (1.5,-0.5) {\re \tiny $a_0(\nu)$};
\node[text centered] at (2.5,-0.5) {\re \tiny $a_1(\nu)$};
\node[text centered] at (3.5,-0.5) {\re \tiny $a_2(\nu)$};
\node[text centered] at (4.5,-0.5) {\re \tiny $a_3(\nu)$};
\node[text centered] at (5.5,-0.5) {\re \tiny $a_4(\nu)$};
\node[text centered] at (6.5,-0.5) {\re \tiny $a_5(\nu)$};
\node[text centered] at (7.5,-0.5) {\re \tiny $a_6(\nu)$};
\node[text centered] at (8.5,-0.5) {\re \tiny $a_7(\nu)$};
\node[text centered] at (9.5,-0.5) {\re \tiny $\cdots$};
\node[text centered] at (1.5,1.5) {\re \tiny $a_0(\lambda)$};
\node[text centered] at (2.5,1.5) {\re \tiny $a_1(\lambda)$};
\node[text centered] at (3.5,1.5) {\re \tiny $a_2(\lambda)$};
\node[text centered] at (4.5,1.5) {\re \tiny $a_3(\lambda)$};
\node[text centered] at (5.5,1.5) {\re \tiny $a_4(\lambda)$};
\node[text centered] at (6.5,1.5) {\re \tiny $a_5(\lambda)$};
\node[text centered] at (7.5,1.5) {\re \tiny $a_6(\lambda)$};
\node[text centered] at (8.5,1.5) {\re \tiny $a_7(\lambda)$};
\node[text centered] at (9.5,1.5) {\re \tiny $\cdots$};
\end{tikzpicture}
=
-z^4 t (1-t)^2.
\end{align*}
We find that we can replace the unique configuration of our single row with multiple configurations of a double row lattice, as shown below: 
\begin{align*}
\bra{{\re(0,0,1,3)}} A(z) \ket{{\re(0,2,2,4)}}
&=
\begin{tikzpicture}[scale=0.6,baseline=0.3cm]
\rblank[1]{1}{1}
\rblank{2}{1}
\rblank{3}{1}
\rblank[1]{4}{1}
\rblank[1]{5}{1}
\rblank{6}{1}
\rblank{7}{1}
\rblank[1]{8}{1}
\rblank{9}{1}
\rblank{1}{0}[1]
\rblank{2}{0}[1]
\rblank{3}{0}[0]
\rblank{4}{0}[1]
\rblank{5}{0}[0]
\rblank{6}{0}[0]
\rblank{7}{0}[1]
\rblank{8}{0}[0]
\rblank{9}{0}[0]
\draw[draw=black!20!red,ultra thick,dotted] (1.5,0) -- (1.5,2);
\draw[draw=black!20!red,ultra thick,dotted] (2.5,0) -- (2.5,0.5) -- (3.5,0.5) -- (3.5,1) -- (4.5,2);
\draw[draw=black!20!red,ultra thick,dotted] (4.5,0) -- (4.5,0.5) -- (5.5,0.5) -- (5.5,2);
\draw[draw=black!20!red,ultra thick,dotted] (7.5,0) -- (7.5,0.5) -- (8.5,0.5) -- (8.5,2);
\node[text centered] at (5.5,-0.5) {$-z^4 t$};
\end{tikzpicture}
+
\begin{tikzpicture}[scale=0.6,baseline=0.3cm]
\rblank[1]{1}{1}
\rblank{2}{1}
\rblank{3}{1}
\rblank[1]{4}{1}
\rblank[1]{5}{1}
\rblank{6}{1}
\rblank{7}{1}
\rblank[1]{8}{1}
\rblank{9}{1}
\rblank{1}{0}[1]
\rblank{2}{0}[1]
\rblank{3}{0}[0]
\rblank{4}{0}[1]
\rblank{5}{0}[0]
\rblank{6}{0}[0]
\rblank{7}{0}[1]
\rblank{8}{0}[0]
\rblank{9}{0}[0]
\draw[draw=black!20!red,ultra thick,dotted] (1.5,0) -- (1.5,2);
\draw[draw=black!20!red,ultra thick,dotted] (2.5,0) -- (2.5,0.5) -- (3.5,0.5) -- (3.5,1) -- (4.5,2);
\draw[draw=black!20!red,ultra thick,dotted] (4.5,0) -- (4.5,0.5) -- (5.5,0.5) -- (5.5,2);
\draw[draw=black!20!red,ultra thick,dotted] (7.5,0) -- (7.5,1) -- (8.5,2);
\node[text centered] at (5.5,-0.5) {$z^4 t^2$};
\end{tikzpicture}
\\
&+
\begin{tikzpicture}[scale=0.6,baseline=0.3cm]
\rblank[1]{1}{1}
\rblank{2}{1}
\rblank{3}{1}
\rblank[1]{4}{1}
\rblank[1]{5}{1}
\rblank{6}{1}
\rblank{7}{1}
\rblank[1]{8}{1}
\rblank{9}{1}
\rblank{1}{0}[1]
\rblank{2}{0}[1]
\rblank{3}{0}[0]
\rblank{4}{0}[1]
\rblank{5}{0}[0]
\rblank{6}{0}[0]
\rblank{7}{0}[1]
\rblank{8}{0}[0]
\rblank{9}{0}[0]
\draw[draw=black!20!red,ultra thick,dotted] (1.5,0) -- (1.5,2);
\draw[draw=black!20!red,ultra thick,dotted] (2.5,0) -- (2.5,0.5) -- (3.5,0.5) -- (3.5,1) -- (4.5,2);
\draw[draw=black!20!red,ultra thick,dotted] (4.5,0) -- (4.5,1) -- (5.5,2);
\draw[draw=black!20!red,ultra thick,dotted] (7.5,0) -- (7.5,0.5) -- (8.5,0.5) -- (8.5,2);
\node[text centered] at (5.5,-0.5) {$z^4 t^2$};
\end{tikzpicture}
+
\begin{tikzpicture}[scale=0.6,baseline=0.3cm]
\rblank[1]{1}{1}
\rblank{2}{1}
\rblank{3}{1}
\rblank[1]{4}{1}
\rblank[1]{5}{1}
\rblank{6}{1}
\rblank{7}{1}
\rblank[1]{8}{1}
\rblank{9}{1}
\rblank{1}{0}[1]
\rblank{2}{0}[1]
\rblank{3}{0}[0]
\rblank{4}{0}[1]
\rblank{5}{0}[0]
\rblank{6}{0}[0]
\rblank{7}{0}[1]
\rblank{8}{0}[0]
\rblank{9}{0}[0]
\draw[draw=black!20!red,ultra thick,dotted] (1.5,0) -- (1.5,2);
\draw[draw=black!20!red,ultra thick,dotted] (2.5,0) -- (2.5,0.5) -- (3.5,0.5) -- (3.5,1) -- (4.5,2);
\draw[draw=black!20!red,ultra thick,dotted] (4.5,0) -- (4.5,1) -- (5.5,2);
\draw[draw=black!20!red,ultra thick,dotted] (7.5,0) -- (7.5,1) -- (8.5,2);
\node[text centered] at (5.5,-0.5) {$-z^4 t^3$};
\end{tikzpicture}
\end{align*}
The rules of this transformation are as follows: {\bf 1.} The lattice paths in the bottom row must form horizontal strips, while those in the top row form vertical strips; {\bf 2.} Each right step in the bottom row comes with a weight of $z$, those in the top row come with a weight of $-tz$. The sum \eqref{one-var-S} is recovered immediately. It is straightforward to extend this argument and show that \eqref{onerow-skew-S} holds for arbitrary partitions. 

\end{proof}

\begin{ex}{\rm 
We give an independent proof of \eqref{onerow-skew-S} when $\nu=0$. A quick manual calculation shows that
\begin{align*}
\bra{{\re\vec 0}}
A(z)
\ket{{\re\vec \lambda}}
=
z^{|\lambda|}
\times
\left\{
\begin{array}{ll}
1,
&
\quad
\lambda = 0,
\\ \\
(1-t)(-t)^{m},
&
\quad
\lambda = (p,\underbrace{1,\dots,1}_{m}),\ p \geq 1,
\\
0,
&
\quad
\text{otherwise.}
\end{array}
\right.
\end{align*}
This can be compared directly with the determinant formula \eqref{JT}. The case $\lambda=0$ is obvious, and when $\lambda = (p,\underbrace{1,\dots,1}_{m})$ with $p \geq 1$ we obtain the determinant
\begin{align*}
S_{\lambda}(z;t)
=
\left|
\begin{array}{ccccc}
q_p & q_{p+1} & \cdots & \cdots & q_{p+m}
\\
q_0 & q_1 & \cdots & \cdots & q_{m}
\\
0 & q_0 & q_1 & \cdots & q_{m-1}
\\
\vdots & \ddots & \ddots & \ddots & \vdots
\\
0 & \cdots & 0 & q_0 & q_1
\end{array}
\right|
=
\left|
\begin{array}{lllll}
(1-t) z^p  & (1-t) z^{p+1} & \cdots & \cdots & (1-t) z^{p+m} 
\\
1 & (1-t)z & \cdots & \cdots & (1-t)z^m
\\
0 & 1 & (1-t)z & \cdots & (1-t)z^{m-1}
\\
\vdots & \ddots & \ddots & \ddots & \vdots
\\
0 & \cdots & 0 & 1 & (1-t) z
\end{array}
\right|,
\end{align*}
which factorizes as $z^{p+m} (1-t) (-t)^m$, after basic row operations which put it in upper-triangular form. Finally, in the case where 
$\lambda$ does not form a border strip, the first two parts of $\lambda$ will be greater than 1. This leads to the first two rows of the determinant \eqref{JT} being proportional, and it vanishes.
}
\end{ex}

\begin{ex}{\rm
Consider $\lambda=(2,1)$, and compute $S_{\lambda}(x_1,x_2;t)$ using equation \eqref{lattice-S}. We find that
\begin{align*}
S_{\lambda}(x_1,x_2;t)
&=
\begin{tikzpicture}[scale=0.6,baseline=0.75cm]
\rblank[0]{1}{1}[0]
\rblank[1]{2}{1}[1]
\rblank[0]{3}{1}[0]
\rblank[1]{4}{1}[1]
\rplus{1}{0}[1]
\rfull{2}{0}[1]
\rfull{3}{0}[0]
\rminus{4}{0}[0]
\end{tikzpicture}
+
\begin{tikzpicture}[scale=0.6,baseline=0.75cm]
\rplus[0]{1}{1}[1]
\rminus[1]{2}{1}[0]
\rblank[0]{3}{1}[0]
\rblank[1]{4}{1}[1]
\rblank{1}{0}[1]
\rplus{2}{0}[1]
\rfull{3}{0}[0]
\rminus{4}{0}[0]
\end{tikzpicture}
+
\begin{tikzpicture}[scale=0.6,baseline=0.75cm]
\rblank[0]{1}{1}[0]
\rblank[1]{2}{1}[1]
\rplus[0]{3}{1}[1]
\rminus[1]{4}{1}[0]
\rplus{1}{0}[1]
\rfull{2}{0}[1]
\rminus{3}{0}[0]
\rblank{4}{0}[0]
\end{tikzpicture}
+
\begin{tikzpicture}[scale=0.6,baseline=0.75cm]
\rplus[0]{1}{1}[1]
\rminus[1]{2}{1}[0]
\rplus[0]{3}{1}[1]
\rminus[1]{4}{1}[0]
\rblank{1}{0}[1]
\rplus{2}{0}[1]
\rminus{3}{0}[0]
\rblank{4}{0}[0]
\end{tikzpicture}
+
\begin{tikzpicture}[scale=0.6,baseline=0.75cm]
\foreach\x in {1,...,2}{\node at (5.5,\x-0.5) {$x_\x$};}
\rplus[0]{1}{1}[1]
\rfull[1]{2}{1}[1]
\rfull[0]{3}{1}[0]
\rminus[1]{4}{1}[0]
\rblank{1}{0}[1]
\rblank{2}{0}[1]
\rblank{3}{0}[0]
\rblank{4}{0}[0]
\end{tikzpicture}
\end{align*}
where we have truncated the infinite lattice to show just the non-trivial part of the partition function. 
The five terms sum to give
\begin{align*}
S_{\lambda}(x_1,x_2;t)
&=
-t(1-t) x_1^3
+
(1-t)^2 x_1^2 x_2 
- 
t(1-t)^2 x_1^2 x_2
+
(1-t)^3 x_1 x_2^2
-
t(1-t) x_2^3
\\
&=
(1-t)\Big(-t x_1^3 + (1-t)^2 x_1^2 x_2 + (1-t)^2 x_1 x_2^2 - t x_2^3\Big).
\end{align*}
This agrees with the determinant calculation
\begin{align*}
S_{\lambda}(x_1,x_2;t)
=
\left|
\begin{array}{cc}
q_2 & q_3
\\
q_0 & q_1
\end{array}
\right|
=
q_2(x_1,x_2;t)
q_1(x_1,x_2;t)
-
q_3(x_1,x_2;t).
\end{align*}
}
\end{ex}

\section{Inverse Kostka polynomials from a rank-two fermion-boson model}

\subsection{Model with one fermion and one boson}
\label{fer-bos-model}

We now consider a model consisting of one fermion and one boson. The entries of the $L$-matrix act in ${\re \mathcal{F}} \otimes \mathcal{B}$, the tensor product of a fermionic and bosonic space, and are given by
\begin{align}
\label{Lmat-bf}
L_a(x|{\re \mathfrak{f}} \otimes \mathfrak{b})
=
\begin{pmatrix}
x[1 \otimes 1] & 1\otimes \phid & {\re \psid} \otimes 1
\\ \\ 
x[t^{\re N} \otimes \phi] & t^{\re N} \otimes 1 & 0
\\ \\
x[{\re \psi} \otimes 1] & {\re \psi} \otimes \phid & (-t)^{\re N} \otimes 1
\end{pmatrix}_a
=
\left(
\begin{array}{ccc}
\begin{tikzpicture}[scale=0.6,baseline=0.25cm]
\rblank{0}{0}
\bblank{1}{0}
\end{tikzpicture} 
& 
\begin{tikzpicture}[scale=0.6,baseline=0.25cm]
\rblank{0}{0}
\bplus{1}{0}
\end{tikzpicture} 
&
\begin{tikzpicture}[scale=0.6,baseline=0.25cm]
\rplus{0}{0}
\darkrfull{1}{0}
\end{tikzpicture} 
\\ \\
\begin{tikzpicture}[scale=0.6,baseline=0.25cm]
\lightbfull{0}{0}
\bminus{1}{0}
\end{tikzpicture} 
& 
\begin{tikzpicture}[scale=0.6,baseline=0.25cm]
\lightbfull{0}{0}
\bfull{1}{0}
\end{tikzpicture} 
&
0
\\ \\
\begin{tikzpicture}[scale=0.6,baseline=0.25cm]
\rminus{0}{0}
\bblank{1}{0}
\end{tikzpicture} 
& 
\begin{tikzpicture}[scale=0.6,baseline=0.25cm]
\rminus{0}{0}
\bplus{1}{0}
\end{tikzpicture} 
&
\begin{tikzpicture}[scale=0.6,baseline=0.25cm]
\rfull{0}{0}
\darkrfull{1}{0}
\end{tikzpicture} 
\end{array}
\right)_a
\end{align}
The $L$-matrix satisfies the intertwining equation
\begin{align}
\label{int-bf}
R_{ab}(x/y)
L_a(x| {\re \mathfrak{f}} \otimes \mathfrak{b})
L_b(y| {\re \mathfrak{f}} \otimes \mathfrak{b})
=
L_b(y| {\re \mathfrak{f}} \otimes \mathfrak{b})
L_a(x| {\re \mathfrak{f}} \otimes \mathfrak{b})
R_{ab}(x/y)
\end{align}
with respect to the $R$-matrix
\begin{align*}
R_{ab}(z)
=
\left(
\begin{array}{ccc|ccc|ccc}
1-tz & 0 & 0 & 0 & 0 & 0 & 0 & 0 & 0
\\
0 & t(1-z) & 0 & 1-t & 0 & 0 & 0 & 0 & 0
\\
0 & 0 & t(1-z) & 0 & 0 & 0 & 1-t & 0 & 0
\\
\hline
0 & (1-t)z & 0 & 1-z & 0 & 0 & 0 & 0 & 0
\\
0 & 0 & 0 & 0 & 1-tz & 0 & 0 & 0 & 0
\\
0 & 0 & 0 & 0 & 0 & t(1-z) & 0 & (1-t)z & 0
\\
\hline
0 & 0 & (1-t)z & 0 & 0 & 0 & 1-z & 0 & 0
\\
0 & 0 & 0 & 0 & 0 & 1-t & 0 & 1-z & 0
\\
0 & 0 & 0 & 0 & 0 & 0 & 0 & 0 & z-t
\end{array}
\right)_{ab}.
\end{align*}
By considering only the entries (1,1), (1,3), (3,1) and (3,3) of the $L$-matrix \eqref{Lmat-bf}, it reduces to the model \eqref{Lmat-f}. We construct a monodromy matrix in the same way as in Section \ref{model-f}, namely
\begin{align}
\label{mon-bf}
T_a(x) 
=
L_a(x|{\re\mathfrak f_0} \otimes {\mathfrak b_0})
L_a(x|{\re\mathfrak f_1} \otimes {\mathfrak b_1})
\cdots
=
\prod_{i=0}^{\infty}
L_a(x|{\re \mathfrak f_i} \otimes {\mathfrak b_i})
:=
\begin{pmatrix}
\Te(x) & \star & \star
\\
\Tb(x) & \star & \star
\\
\Tr(x) & \star & \star
\end{pmatrix}_a,
\end{align}
where we only indicate the operators in the first column of the monodromy matrix, since these are the only one we need subsequently. From \eqref{int-bf}, we see that they generate a sub-algebra of the Yang--Baxter algebra, whose relations include
\begin{align}
\label{T-be}
(y-x)
\Tb(x) \Te(y)
+
(1-t) x
\Te(x) \Tb(y)
=
(y-tx)
\Te(y) \Tb(x),
\\
\label{T-re}
(y-x)
\Tr(x) \Te(y)
+
(1-t) x
\Te(x) \Tr(y)
=
(y-tx)
\Te(y) \Tr(x),
\\
\label{T-ee}
\Tb(x) \Tb(y)
=
\Tb(y) \Tb(x),
\quad
(x - ty)
\Tr(x) \Tr(y)
=
(y - tx)
\Tr(y) \Tr(x).
\end{align}

\subsection{Definition of Kostka polynomials and their inverses}

The Kostka polynomials express Schur polynomials in the Hall--Littlewood basis of the ring of symmetric functions:
\begin{align*}
s_{\lambda}(x) = \sum_{\nu} K^{\lambda}_{\nu}(t) P_{\nu}(x;t).
\end{align*}
While it is not obvious from their definition, $K^{\lambda}_{\nu}(t) \in \mathbb{N}[t]$ for all $\lambda,\nu$. Exploiting this positivity, there are many combinatorial formulae for $K^{\lambda}_{\nu}(t)$, of which the most well-known are due to Lascoux and Sch\"utzenberger \cite{las-sch} and Kirillov and Reshetikhin \cite{kir-res}. In the present work we will not have more to say about the Kostka polynomials, but will consider rather their inverses: 
\begin{align}
\label{invK}
P_{\nu}(x;t)
=
\sum_{\lambda}
\b{K}^{\lambda}_{\nu}(t)
s_{\lambda}(x).
\end{align}
Here it is clear that $\b{K}^{\lambda}_{\nu}(t) \in \mathbb{Z}[t]$, since $P_{\nu}(x;t) \in \mathbb{Z}[x,t]$ and the Schur polynomials form a $\mathbb{Z}$-basis for the ring of symmetric functions, but the coefficients of $\b{K}^{\lambda}_{\nu}(t)$ are no longer positive. Nevertheless, the second main goal of this paper is to derive a combinatorial expression for these polynomials. Before proceeding to this result, we first generalize them.

\subsection{Generalized inverse Kostka polynomials}
\label{gen-inv-kost}

We start from the branching formula \eqref{S-branch} for the $t$-Schur polynomials:
\begin{align*}
S_{\lambda}(x,y;t)
=
\sum_{\mu}
S_{\lambda/\mu}(x;t)
S_{\mu}(y;t).
\end{align*}
Multiplying both sides by $s_{\lambda}(z)$ and summing over $\lambda$, from \eqref{cauchy-Ss} we find that
\begin{align*}
\prod_{i,j}
\frac{1-tx_i z_j}{1-x_i z_j}
\prod_{i,j}
\frac{1-ty_i z_j}{1-y_i z_j}
=
\sum_{\lambda,\mu}
S_{\lambda/\mu}(x;t)
S_{\mu}(y;t)
s_{\lambda}(z),
\end{align*}
or alternatively, re-expanding the left hand side in a different way,
\begin{align*}
\sum_{\mu,\nu}
Q_{\nu}(x;t)
P_{\nu}(z;t)
S_{\mu}(y;t)
s_{\mu}(z)
=
\sum_{\lambda,\mu}
S_{\lambda/\mu}(x;t)
S_{\mu}(y;t)
s_{\lambda}(z).
\end{align*}
Since the $t$-Schur polynomials are a $\mathbb{Z}$-basis for the ring of symmetric functions \cite{mac}, we can extract coefficients of $S_{\mu}(y;t)$ from both sides, giving
\begin{align}
\label{...}
\sum_{\nu}
Q_{\nu}(x;t)
P_{\nu}(z;t)
s_{\mu}(z)
=
\sum_{\lambda}
S_{\lambda/\mu}(x;t)
s_{\lambda}(z).
\end{align}
Now define $\b{K}^{\lambda}_{\mu\nu}(t)$ as the expansion coefficients of the product of a Schur and Hall--Littlewood polynomial, over the Schur basis:
\begin{align*}
s_{\mu}(z)
P_{\nu}(z;t)
:=
\sum_{\lambda}
\b{K}^{\lambda}_{\mu\nu}(t)
s_{\lambda}(z),
\end{align*}
which clearly generalizes \eqref{invK} to include a third non-trivial partition. Substituting this into \eqref{...} and comparing coefficients of $s_{\lambda}(z)$ on both sides, one finds that
\begin{align*}
S_{\lambda/\mu}(x;t)
=
\sum_{\nu}
\b{K}^{\lambda}_{\mu\nu}(t)
Q_{\nu}(x;t).
\end{align*}
Hence to calculate $\b{K}^{\lambda}_{\mu\nu}(t)$ we need to expand a skew $t$-Schur polynomial in the Hall--Littlewood basis. Applying the results of Section \ref{CT}, we can write the generalized inverse Kostka polynomials as a constant term expression:
\begin{align}
\label{CT-invK}
{\rm Coeff}\left[
\prod_{1 \leq i<j \leq n}
\left(
\frac{1-z_j/z_i}{1- t z_j/z_i}
\right)
S_{\lambda/\mu}(z_1,\dots,z_n;t),
z_1^{\nu_1} \dots z_n^{\nu_n}
\right]
=
\b{K}^{\lambda}_{\mu\nu}(t)
b_{\nu}(t).
\end{align}

\subsection{Inverse Kostka polynomials from action of divided-difference operators}

\begin{thm}
\label{thm-kost}
{\rm
Let $\lambda, \mu, \nu$ be three partitions with respective lengths $\ell,m,n$ such that $\ell = m$, and largest parts $L,M,N$ such that $M \geq L$ and $m+M > N$. Let $\b\lambda,\b\mu$ denote the partitions obtained by complementing each part of $\lambda,\mu$ by $M$, $\b\nu$ the partition obtained by complementing each part of $\nu$ by $m+M-1$, and assume that $|\b\lambda| = |\b\mu|+|\b\nu|$. Then
\begin{align}
\label{K-exp}
\mathcal{K}^{\lambda}_{\mu\nu}(x;t)
=
\bra{{\re\cev \mu}}
\otimes
\bra{\ }
\prod_{i=1}^{n+|\nu|}
\left\{
\begin{array}{ll}
\Te(x_i),
&
i \in k(\nu)
\\
\Tb(x_i),
&
i \not\in k(\nu)
\end{array}
\right\}
\ket{{\re\cev \lambda}}
\otimes
\ket{{\cev 0}}
\end{align}
is a homogeneous polynomial in $\{x_1,\dots,x_{n+|\nu|}\}$ of total degree $n+2|\nu|$, and the generalized inverse Kostka polynomial $\b{K}^{\b\lambda}_{\b\mu\b\nu}(t)$ is given by the repeated action of divided-difference operators on this partition function:
\begin{align}
\label{main-formula2}
\b{K}^{\b\lambda}_{\b\mu\b\nu}(t)
=
\frac{t^{-(m+1)|\nu|}}{b_{\b\nu}(t)}
\times
\left(
\prod_{j \not\in k(\nu)}
\Delta_j
\right)
\left(
\prod_{i=1}^{n+|\nu|}
\b{x}_i
\right)
\mathcal{K}^{\lambda}_{\mu\nu}(x;t),
\end{align}
where the product of divided-difference operators is ordered from left to right, starting at $i=1$ and finishing at $i = n+|\nu|$, and omitting $i \in k(\nu)$.
}
\end{thm}

The proof of Theorem \ref{thm-kost} is very similar to that of Theorem 
\ref{thm-hall}. To limit direct repetition of ideas, we give this proof in Appendix \ref{B}.

\subsection{Combinatorial interpretation of equation \eqref{main-formula2}}

In order to make a precise statement of Theorem \ref{thm-kost-puzzle}, we begin by adapting each of the Definitions \ref{def-dipole}--\ref{weight} to fit the fermion-boson model that we are currently considering. The proof of Theorem \ref{thm-kost-puzzle} is omitted, since it follows by direct analogy with the proof of Theorem \ref{thm-hall-puzzle}.

\noindent \textbf{Definition 2$\bm{'}$.}
Consider a rectangular grid $G$ of light and dark squares, \begin{tikzpicture}[scale=0.5,baseline=0.1cm] \blight{0}{0} \end{tikzpicture} and \begin{tikzpicture}[scale=0.5,baseline=0.1cm] \bdark{0}{0} \end{tikzpicture}\ . A dipole tiling of $G$ is a placement of horizontal dipoles  
\begin{tikzpicture}[scale=0.5,baseline=0.1cm] 
\filldraw[fill=black!20!red,draw=black!20!red] (0.5,0.5) circle (0.25cm); 
\draw[rline] (0.5,0.5) -- (2,0.5); 
\filldraw[fill=black!20!red,draw=black!20!red] (2,0.5) circle (0.25cm);
\node[text centered] at (0.5,0.5) {\color{white} $+$};  
\node[text centered] at (2,0.5) {\color{white} $-$};
\end{tikzpicture} 
and
\begin{tikzpicture}[scale=0.5,baseline=0.1cm] 
\filldraw[fill=black,draw=black] (0.5,0.5) circle (0.25cm); 
\draw[bline] (0.5,0.5) -- (2,0.5); 
\filldraw[fill=black,draw=black] (2,0.5) circle (0.25cm);
\node[text centered] at (0.5,0.5) {\color{white} $+$};  
\node[text centered] at (2,0.5) {\color{white} $-$};
\end{tikzpicture}
of any length on $G$, such that red (black) dipoles begin and end on light (dark) squares, and dipoles cannot overlap.

\noindent \textbf{Definition 3$\bm{'}$.}
Fix a non-negative integer $N \geq 0$. An $N$--puzzle is a dipole tiling of the grid
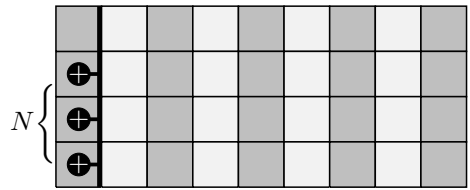
\begin{align*}
\begin{tikzpicture}[scale=0.6]
\node at (-0.5,5.4) {$N \Bigg\{$ };
\plusb{0}{4}
\plusb{0}{5}
\plusb{0}{6}
\bdark{0}{7}
\blight{1}{7}
\bdark{2}{7}
\blight{3}{7}
\bdark{4}{7}
\blight{5}{7}
\bdark{6}{7}
\blight{7}{7}
\bdark{8}{7}
\blight{1}{6}
\bdark{2}{6}
\blight{3}{6}
\bdark{4}{6}
\blight{5}{6}
\bdark{6}{6}
\blight{7}{6}
\bdark{8}{6}
\blight{1}{5}
\bdark{2}{5}
\blight{3}{5}
\bdark{4}{5}
\blight{5}{5}
\bdark{6}{5}
\blight{7}{5}
\bdark{8}{5}
\blight{1}{4}
\bdark{2}{4}
\blight{3}{4}
\bdark{4}{4}
\blight{5}{4}
\bdark{6}{4}
\blight{7}{4}
\bdark{8}{4}
\draw[ultra thick] (0.95,4) -- (0.95,8);
\end{tikzpicture}
\end{align*}
in which the left column is frozen as indicated: its top square must be unoccupied, \begin{tikzpicture}[scale=0.5,baseline=0.1cm] \bdark{0}{0} \end{tikzpicture}\ , and the remaining $N$ squares are required to be of the form \begin{tikzpicture}[scale=0.5,baseline=0.1cm] \plusb{0}{0} \end{tikzpicture}\ . Numbering the rows from top to bottom by $\{0,1,\dots,N\}$, let $r_i$ denote the total number of tiles of the form \begin{tikzpicture}[scale=0.5,baseline=0.1cm] \bdark{0}{0} \end{tikzpicture} or \begin{tikzpicture}[scale=0.5,baseline=0.1cm] \minb{0}{0} \end{tikzpicture} in the $i$-th row. Then we require that $(r_1,\dots,r_N)$ is a partition and 
$\sum_{i=0}^{N} r_i = 2N+2$. The length of the $N$--puzzle is the total number of rows $i$ for which $r_i > 1$.

\noindent \textbf{Definition 4$\bm{'}$.}
Let $\lambda,\mu,\nu$ be three partitions. A {\it $\nu$--puzzle} with frame $(\mu,\lambda)$ is a dipole tiling of the form
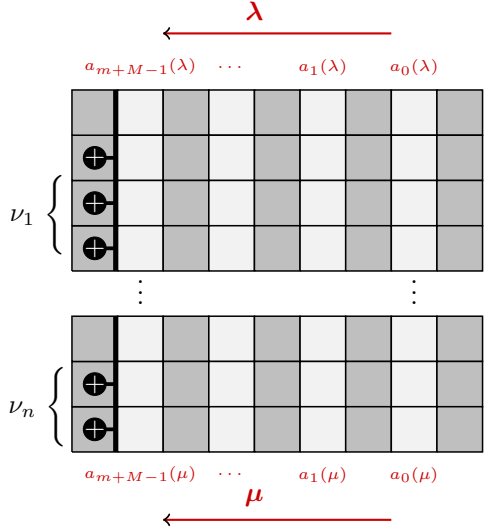
\begin{align*}
\begin{tikzpicture}[scale=0.6]
\node at (-0.75,9.25) {$\nu_1 \  \Bigg\{$ };
\node at (-0.75,5) {$\nu_n \  \Bigg\{$ };
\plusb{0}{4}
\plusb{0}{5}
\bdark{0}{6}
\plusb{0}{8}
\plusb{0}{9}
\plusb{0}{10}
\bdark{0}{11}
\blight{1}{11}
\bdark{2}{11}
\blight{3}{11}
\bdark{4}{11}
\blight{5}{11}
\bdark{6}{11}
\blight{7}{11}
\bdark{8}{11}
\blight{1}{10}
\bdark{2}{10}
\blight{3}{10}
\bdark{4}{10}
\blight{5}{10}
\bdark{6}{10}
\blight{7}{10}
\bdark{8}{10}
\blight{1}{9}
\bdark{2}{9}
\blight{3}{9}
\bdark{4}{9}
\blight{5}{9}
\bdark{6}{9}
\blight{7}{9}
\bdark{8}{9}
\blight{1}{8}
\bdark{2}{8}
\blight{3}{8}
\bdark{4}{8}
\blight{5}{8}
\bdark{6}{8}
\blight{7}{8}
\bdark{8}{8}
\node at (1.5,7.7) {$\vdots$};
\node at (7.5,7.7) {$\vdots$};
\blight{1}{6}
\bdark{2}{6}
\blight{3}{6}
\bdark{4}{6}
\blight{5}{6}
\bdark{6}{6}
\blight{7}{6}
\bdark{8}{6}
\blight{1}{5}
\bdark{2}{5}
\blight{3}{5}
\bdark{4}{5}
\blight{5}{5}
\bdark{6}{5}
\blight{7}{5}
\bdark{8}{5}
\blight{1}{4}
\bdark{2}{4}
\blight{3}{4}
\bdark{4}{4}
\blight{5}{4}
\bdark{6}{4}
\blight{7}{4}
\bdark{8}{4}
\draw[arrow=1,draw=black!20!red,thick] (7,2.5) -- (2,2.5);
\node at (4,3) {$\re{\bm \mu}$};
\draw[arrow=1,draw=black!20!red,thick] (7,13.25) -- (2,13.25);
\node at (4,13.75) {$\re{\bm \lambda}$};
\node[text centered] at (1.5,12.5) {\re \tiny $a_{m+M-1}(\lambda)$};
\node[text centered] at (3.5,12.5) {\re \tiny $\cdots$};
\node[text centered] at (5.5,12.5) {\re \tiny $a_1(\lambda)$};
\node[text centered] at (7.5,12.5) {\re \tiny $a_0(\lambda)$};
\node[text centered] at (1.5,3.5) {\re \tiny $a_{m+M-1}(\mu)$};
\node[text centered] at (3.5,3.5) {\re \tiny $\cdots$};
\node[text centered] at (5.5,3.5) {\re \tiny $a_1(\mu)$};
\node[text centered] at (7.5,3.5) {\re \tiny $a_0(\mu)$};
\draw[ultra thick] (0.95,4) -- (0.95,7); \draw[ultra thick] (0.95,8) -- (0.95,12);
\end{tikzpicture}
\end{align*}
obtained by vertically concatenating $\nu_i$--puzzles for $1 \leq i \leq n$, such that the total charge of the $j$-th light column from the right is $a_j(\mu) - a_j(\lambda)$, and all internal dark columns have total charge 0. Necessarily, the left and rightmost dark columns will have total charge $|\nu|$ and $-|\nu|$, respectively. The length of the $\nu$--puzzle is the sum of the lengths of its constituent $\nu_i$--puzzles.

\noindent \textbf{Definition 5$\bm{'}$.}
The cumulative charge of a tile is the sum (reading from top downward) of all charges in its column, up to and including the tile itself, plus the occupation number at the top of the column. 

\noindent \textbf{Definition 6$\bm{'}$.}
The Boltzmann weight of a puzzle $P$, denoted $W(P)$, is the product of local weights assigned to each tile in the lattice (again, excluding the tiles in the trivial left column). We indicate the weight of each tile below:
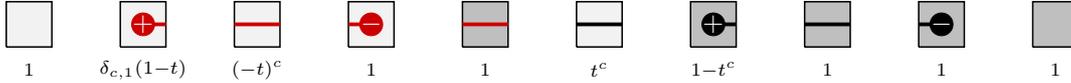
\begin{align*}
\begin{tikzpicture}[scale=0.6]
\blight{1}{1};
\plusr{3.5}{1};
\flatr{6}{1};
\minr{8.5}{1};
\flatR{11}{1};
\flatB{13.5}{1};
\plusb{16}{1};
\flatb{18.5}{1};
\minb{21}{1};
\bdark{23.5}{1};
\node[text centered] at (1.5,0.5) {$\ss 1$};
\node[text centered] at (4,0.5) {$\ss \delta_{c,1}(1- t)$};
\node[text centered] at (6.5,0.5) {$\ss (-t)^c$};
\node[text centered] at (9,0.5) {$\ss 1$};
\node[text centered] at (11.5,0.5) {$\ss 1$};
\node[text centered] at (14,0.5) {$\ss t^c$};
\node[text centered] at (16.5,0.5) {$\ss 1-t^{c}$};
\node[text centered] at (19,0.5) {$\ss 1$};
\node[text centered] at (21.5,0.5) {$\ss 1$};
\node[text centered] at (24,0.5) {$\ss 1$};
\end{tikzpicture} 
\end{align*}
where $c$ denotes the cumulative charge of the tile. From these weights, we see that the cumulative charge of any light tile is forced to be either 0 or 1, while the cumulative charge of any dark tile must be non-negative.

\begin{thm}
\label{thm-kost-puzzle}
{\rm

Let $\lambda,\mu,\nu$ be three partitions with the same definitions as in Theorem \ref{thm-kost}, and $\b\lambda,\b\mu,\b\nu$ their complements. Then the generalized inverse Kostka polynomial $\b{K}^{\b\lambda}_{\b\mu\b\nu}(t)$ is obtained by summing over all $\nu$--puzzles with frame $(\mu,\lambda)$:
\begin{align}
\label{invK-comb}
K^{\b\lambda}_{\b\mu\b\nu}(t)
=
\frac{t^{-(m+1)|\nu|}}{b_{\b\nu}(t)}
\times
\sum_{P \in \mathbb{P}_{\nu}(\mu,\lambda)}
(-1)^{L(P)}
W(P),
\end{align}
where $L(P)$ denotes the length of $P$. An example of \eqref{invK-comb} is given in Appendix \ref{app:c}.

}
\end{thm}

\section{Knutson--Tao puzzles from integrability}

In this section we calculate a certain partition function in the model \eqref{Lmat-bf} at $t=0$, and prove that it is proportional to the Littlewood--Richardson coefficient $c^{\lambda}_{\mu\nu}$. We then give a simple bijection between lattice configurations in the model under consideration, and Knutson--Tao puzzles \cite{knu-tao,knu-tao-woo}.

\subsection{Fermion-boson model at \texorpdfstring{$t=0$}{t=0}}

We begin by taking the limit $t \rightarrow 0$ of the model \eqref{Lmat-bf}. The operators of the model acquire the simplified actions
\begin{align*}
&
\phid \ket{m}
=
\ket{m+1},
\quad
\phi \ket{m}
=
(1-\delta_{m,0})
\ket{m-1}
\quad
\text{on}\ \mathcal{B},
\\
&
{\re \psid} \ket{{\re 0}}
=
\ket{{\re 1}},
\quad
{\re \psid} \ket{{\re 1}}
=
0,
\quad
{\re \psi} \ket{{\re 0}}
=
0,
\quad
{\re \psi} \ket{{\re 1}}
=
\ket{{\re 0}},
\quad
{\re N} \ket{{\re 0}}
=
0,
\quad
{\re N} \ket{{\re 1}}
=
\ket{{\re 1}}
\quad
\text{on}\ \mathcal{F}.
\end{align*}
In particular, we notice that $(\pm t)^{\re N} \mapsto {\re \pi} = \ket{{\re 0}}\bra{{\re 0}}$ when $t=0$, the projector onto the vacuum. Finally, inverting the spectral parameter in \eqref{Lmat-bf} (this is purely a convenience) we find that
\begin{align}
\label{t=0Lmat}
L^{*}_a(x| {\re \mathfrak{f}} \otimes \mathfrak{b})
=
\begin{pmatrix}
1 \otimes 1 & x [ 1 \otimes \phid ] & x [{\re \psid} \otimes 1]
\\ \\ 
{\re \pi} \otimes \phi & x [ {\re \pi} \otimes 1 ] & 0
\\ \\
{\re \psi} \otimes 1 & x [ {\re \psi} \otimes \phid ] & x [{\re \pi} \otimes 1 ]
\end{pmatrix}_a
=
\left(
\begin{array}{ccc}
\begin{tikzpicture}[scale=0.6,baseline=0.25cm]
\rblank{0}{0}
\bblank{1}{0}
\end{tikzpicture} 
& 
\begin{tikzpicture}[scale=0.6,baseline=0.25cm]
\rblank{0}{0}
\bplus{1}{0}
\end{tikzpicture} 
&
\begin{tikzpicture}[scale=0.6,baseline=0.25cm]
\rplus{0}{0}
\darkrfull{1}{0}
\end{tikzpicture} 
\\ \\
\begin{tikzpicture}[scale=0.6,baseline=0.25cm]
\lightbfull{0}{0}
\bminus{1}{0}
\end{tikzpicture} 
& 
\begin{tikzpicture}[scale=0.6,baseline=0.25cm]
\lightbfull{0}{0}
\bfull{1}{0}
\end{tikzpicture} 
&
0
\\ \\
\begin{tikzpicture}[scale=0.6,baseline=0.25cm]
\rminus{0}{0}
\bblank{1}{0}
\end{tikzpicture} 
& 
\begin{tikzpicture}[scale=0.6,baseline=0.25cm]
\rminus{0}{0}
\bplus{1}{0}
\end{tikzpicture} 
&
\begin{tikzpicture}[scale=0.6,baseline=0.25cm]
\rfull{0}{0}
\darkrfull{1}{0}
\end{tikzpicture} 
\end{array}
\right)_a
\quad
\text{when}\ t=0.
\end{align}
Because of the projectors which are now present in the $L$-matrix, this model is combinatorially much simpler than those considered earlier in the paper. In particular, the following expectation value vanishes:
\begin{align}
\label{forbid}
\bra{{\re 1}}
{\re \pi}
\ket{{\re 1}}
=
\begin{tikzpicture}[scale=0.8,baseline=1cm]
\lightbfullr[1]{1}{1}[1]  
\end{tikzpicture}
=
\begin{tikzpicture}[scale=0.8,baseline=1cm]
\rfull[1]{1}{1}[1]  
\end{tikzpicture}
=
0,
\end{align}
meaning that the two tiles shown above are forbidden in lattice configurations of the model.

We construct a monodromy matrix in the standard way:
\begin{align*}
T^{*}_a(x) 
=
L^{*}_a(x|{\re\mathfrak f_0} \otimes {\mathfrak b_0})
L^{*}_a(x|{\re\mathfrak f_1} \otimes {\mathfrak b_1})
\cdots
=
\begin{pmatrix}
\Tes(x) & \star & \star
\\
\Tbs(x) & \star & \star
\\
\Trs(x) & \star & \star
\end{pmatrix}_a.
\end{align*}
The commutation relations that we shall require are
\begin{align}
\label{schur-com1}
\Trs(x) \Tes(y)
&=
\frac{\b{y}}{\b{y}-\b{x}}
\Tes(y) \Trs(x)
+
\frac{\b{x}}{\b{x}-\b{y}}
\Tes(x) \Trs(y),
\\
\label{schur-com2}
\b{x} \Trs(x) \Trs(y) 
&= 
\b{y} \Trs(y) \Trs(x),
\end{align}
which (up to inversion of the spectral parameters) are the $t=0$ specialization of \eqref{T-re} and \eqref{T-ee}.

\subsection{Littlewood--Richardson coefficients from a partition function}
\label{sec:LR-schur}

\begin{thm}
\label{thm-LR}
{\rm
Let $\lambda$ be a partition of length $\ell$ with largest part $L$. Fix two other partitions $\mu$, $\nu$, also both of length $\ell$, which satisfy $|\lambda| = |\mu| + |\nu|$. Let $\b\nu$ denote the complement of $\nu$ by $L+1$, and define the set $\kappa(\b\nu) = \{\kappa_1,\dots,\kappa_\ell\}$ where $\kappa_i = i+\b\nu_{\ell-i+1} = i+L+1-\nu_i$. Define the partition function
\begin{align}
\label{LR-PF-schur}
\mathcal{C}^{\lambda}_{\mu\nu}(x)
=
\bra{\hole[0.1]}
\otimes
\bra{{\vec\mu}}
\prod_{i=1}^{\ell+L+1}
\left\{
\begin{array}{ll}
\Te^{*}(x_i),
&
i \in \kappa(\b\nu)
\\
\Tr^{*}(x_i),
&
i \not\in \kappa(\b\nu)
\end{array}
\right\}
\ket{\rpart[0.1]}
\otimes
\ket{{\vec \lambda}},
\end{align}
where $\ket{{\vec \lambda}}$ and $\ket{{\vec \mu}}$ are partition states in $\mathcal{B}$, while $\ket{\rpart[0.1]}$ and $\ket{\hole[0.1]}$ denote completely occupied and empty states in $\mathcal{F}$. Then 
$\mathcal{C}^{\lambda}_{\mu\nu}(x)$ is a polynomial in $x_1,\dots,x_{\ell+L+1}$ consisting of a single monomial, 
$\prod_{i=1}^{\ell+L+1} x_i^{p_i}$, where $p_i = \#\{j > i | j \not\in \kappa(\b\nu) \}$. The Littlewood--Richardson coefficient $c^{\lambda}_{\mu\nu}$ is the coefficient (possibly, zero) of that monomial:
\begin{align}
\label{LR-schur}
c^{\lambda}_{\mu\nu}
=
\left(
\prod_{i=1}^{\ell+L+1} 
\b{x}_i^{p_i}
\right)
\mathcal{C}^{\lambda}_{\mu\nu}(x).
\end{align}

}
\end{thm}
The rest of this section will be devoted to the proof. It proceeds along similar lines to the proofs of Theorems \ref{thm-hall} and \ref{thm-kost}, apart from some key differences which make it worthwhile to sketch this proof in its own right. 

\begin{ex}
\label{ex-C}
{\rm
Let $\lambda = (4,4,2,1)$, $\mu = (3,3,1,0)$, $\nu = (2,1,1,0)$. Then $\ell = 4, L = 4, \b\nu = (5,4,4,3)$ and
\begin{align*}
\mathcal{C}^{\lambda}_{\mu\nu}(x)
=
\bra{\hole[0.1]}
\otimes
\bra{{\vec\mu}}
\Trs(x_1) \Trs(x_2) \Trs(x_3)
\Tes(x_4)
\Trs(x_5)
\Tes(x_6) \Tes(x_7)
\Trs(x_8)
\Tes(x_9)
\ket{\rpart[0.1]}
\otimes
\ket{{\vec \lambda}},
\end{align*} 
which is graphically represented as
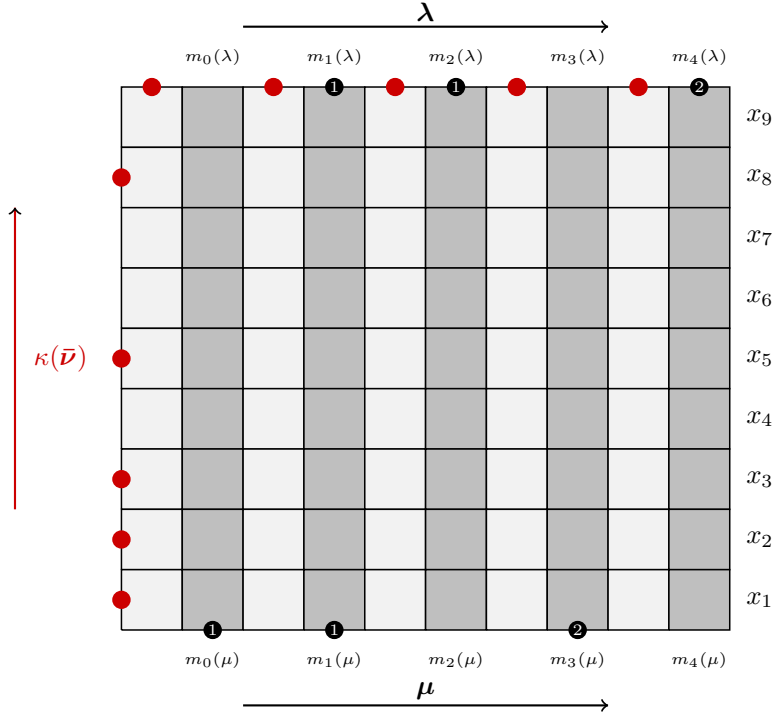
\begin{align*}
\begin{tikzpicture}[scale=0.8]
\foreach\x in {1,...,9}{\node at (11.5,\x-1.5) {$x_\x$};}
\draw[arrow=1,draw=black,thick] (3,-2.25) -- (9,-2.25);
\node at (6,-2) {${\bm \mu}$};
\node[text centered] at (2.5,-1.5) {\tiny $m_0(\mu)$};
\node[text centered] at (4.5,-1.5) {\tiny $m_1(\mu)$};
\node[text centered] at (6.5,-1.5) {\tiny $m_2(\mu)$};
\node[text centered] at (8.5,-1.5) {\tiny $m_3(\mu)$};
\node[text centered] at (10.5,-1.5) {\tiny $m_4(\mu)$};
\draw[arrow=1,draw=black,thick] (3,9) -- (9,9);
\node at (6,9.25) {${\bm \lambda}$};
\node[text centered] at (2.5,8.5) {\tiny $m_0(\lambda)$};
\node[text centered] at (4.5,8.5) {\tiny $m_1(\lambda)$};
\node[text centered] at (6.5,8.5) {\tiny $m_2(\lambda)$};
\node[text centered] at (8.5,8.5) {\tiny $m_3(\lambda)$};
\node[text centered] at (10.5,8.5) {\tiny $m_4(\lambda)$};
\draw[arrow=1,draw=black!20!red,thick] (-0.75,1) -- (-0.75,6);
\node at (0,3.5) {$\re \kappa(\bm {\b\nu})$};
\redoc{1}{-0.5}{}
\redoc{1}{0.5}{}
\redoc{1}{1.5}{}
\redoc{1}{3.5}{}
\redoc{1}{6.5}{}
\rblank[1]{1}{7}
\bblank[0]{2}{7}
\rblank[1]{3}{7}
\bblank[1]{4}{7}
\rblank[1]{5}{7}
\bblank[1]{6}{7}
\rblank[1]{7}{7}
\bblank[0]{8}{7}
\rblank[1]{9}{7}
\bblank[2]{10}{7}
\rblank{1}{6}
\bblank{2}{6}
\rblank{3}{6}
\bblank{4}{6}
\rblank{5}{6}
\bblank{6}{6}
\rblank{7}{6}
\bblank{8}{6}
\rblank{9}{6}
\bblank{10}{6}
\rblank{1}{5}
\bblank{2}{5}
\rblank{3}{5}
\bblank{4}{5}
\rblank{5}{5}
\bblank{6}{5}
\rblank{7}{5}
\bblank{8}{5}
\rblank{9}{5}
\bblank{10}{5}
\rblank{1}{4}
\bblank{2}{4}
\rblank{3}{4}
\bblank{4}{4}
\rblank{5}{4}
\bblank{6}{4}
\rblank{7}{4}
\bblank{8}{4}
\rblank{9}{4}
\bblank{10}{4}
\rblank{1}{3}
\bblank{2}{3}
\rblank{3}{3}
\bblank{4}{3}
\rblank{5}{3}
\bblank{6}{3}
\rblank{7}{3}
\bblank{8}{3}
\rblank{9}{3}
\bblank{10}{3}
\rblank{1}{2}
\bblank{2}{2}
\rblank{3}{2}
\bblank{4}{2}
\rblank{5}{2}
\bblank{6}{2}
\rblank{7}{2}
\bblank{8}{2}
\rblank{9}{2}
\bblank{10}{2}
\rblank{1}{1}
\bblank{2}{1}
\rblank{3}{1}
\bblank{4}{1}
\rblank{5}{1}
\bblank{6}{1}
\rblank{7}{1}
\bblank{8}{1}
\rblank{9}{1}
\bblank{10}{1}
\rblank{1}{0}
\bblank{2}{0}
\rblank{3}{0}
\bblank{4}{0}
\rblank{5}{0}
\bblank{6}{0}
\rblank{7}{0}
\bblank{8}{0}
\rblank{9}{0}
\bblank{10}{0}
\rblank{1}{-1}[0]
\bblank{2}{-1}[1]
\rblank{3}{-1}[0]
\bblank{4}{-1}[1]
\rblank{5}{-1}[0]
\bblank{6}{-1}[0]
\rblank{7}{-1}[0]
\bblank{8}{-1}[2]
\rblank{9}{-1}[0]
\bblank{10}{-1}[0]
\end{tikzpicture}
\end{align*}

}
\end{ex}

\subsubsection{Degree counting}
\label{deg-count}

Writing the $L$-matrix \eqref{t=0Lmat} in its graphical representation, we clearly obtain an $x$ weight for every occupied right edge (ORE). Studying the lattice representation of $\mathcal{C}^{\lambda}_{\mu\nu}(x)$, in every legal configuration the black (bosonic) lines give rise to $|\lambda| - |\mu| = |\nu|$ OREs, while the red (fermionic) lines produce a total of $0+1+\cdots+ L = L(L+1)/2$ OREs. This is a total of $|\nu| + L(L+1)/2$ OREs, meaning that $\mathcal{C}^{\lambda}_{\mu\nu}(x)$ is a homogeneous polynomial of degree $|\nu| + L(L+1)/2$.

On the other hand, it is immediate that the sum of the exponents $p_i$ gives 
$\sum_{i=1}^{\ell+L+1} p_i = |\nu| + L(L+1)/2$. Hence to establish the first claim of the theorem (that 
$\mathcal{C}^{\lambda}_{\mu\nu}(x)$ consists of the single monomial 
$\prod_{i=1}^{\ell+L+1}  x_i^{p_i}$), we need to show that the limits
\begin{align*}
\lim_{x_j \rightarrow \infty}
\left(
\prod_{i=1}^{\ell+L+1} 
\b{x}_i^{p_i}
\right)
\mathcal{C}^{\lambda}_{\mu\nu}(x)
\end{align*}
exist for all $1 \leq j \leq \ell+L+1$. This ensures that $\mathcal{C}^{\lambda}_{\mu\nu}(x)$ contains no other degree $|\nu| + L(L+1)/2$ monomials.

\subsubsection{Multiple sum expression for $\mathcal{C}^{\lambda}_{\mu\nu}(x)$}

By repeated use of the commutation relations \eqref{schur-com1} and \eqref{schur-com2}, one can completely order the operators appearing in \eqref{LR-PF-schur} (so that all $\Tes$ operators are on the left). The result of the computation is given below:
\begin{multline}
\label{n-iterat3}
\prod_{i=1}^{\ell+L+1} 
\b{x}_i^{p_i}
\times
\mathcal{C}^{\lambda}_{\mu\nu}(x)
=
\sum_{1 \leq i_1 \leq \kappa_1}
\cdots
\sum_{\substack{1 \leq i_{\ell} \leq \kappa_{\ell} \\ i_\ell \not=i_1,\dots,i_{\ell-1}}}
\prod_{j_1\not=i_1}^{\kappa_1}
\frac{1}{(\b{x}_{i_1}-\b{x}_{j_1})}
\cdots
\prod_{j_\ell\not=i_1,\dots,i_\ell}^{\kappa_\ell}
\frac{1}{(\b{x}_{i_\ell}-\b{x}_{j_\ell})}
\\
\times
\prod_{a=1}^{\ell}
\b{x}_{i_a}^{L+1}
\times
\prod_{b=1}^{L+1}
\b{x}_{\hat\imath_b}^{L+1-b}
\times
\bra{\hole[0.1]}
\otimes
\bra{{\vec\mu}}
\Tes(x_{i_1})
\dots
\Tes(x_{i_\ell})
\Trs(x_{\hat\imath_1})
\dots
\Trs(x_{\hat\imath_{L+1}})
\ket{\rpart[0.1]}
\otimes
\ket{{\vec \lambda}}
\end{multline}
where the summation is over distinct integers $\{i_1,\dots,i_\ell\}$ such that $i_j \leq \kappa_j$ for all $1 \leq j \leq \ell$, and where $\{\hat\imath_1,\dots,\hat\imath_{L+1}\}$ denotes the complement of $\{i_1,\dots,i_\ell\}$ in 
$\{1,\dots,\ell+L+1\}$.

\subsubsection{Trivial action}

\begin{lem}{\rm
Let $\ket{{\vec \lambda}}$ be an arbitrary partition state in $\mathcal{B}$ with largest part $L$, and $\ket{\rpart[0.1]}$/$\ket{\hole[0.1]}$ the completely occupied/unoccupied state in ${\re \mathcal{F}}$. Then
\begin{align}
\label{schur-trivial}
\Trs(y_1)
\dots
\Trs(y_{L+1})
\ket{\rpart[0.1]}
\otimes
\ket{{\vec \lambda}}
=
\prod_{i=1}^{L+1}
y_i^{L+1-i}
\ket{\hole[0.1]}
\otimes
\ket{{\vec \lambda}}.
\end{align}
}
\end{lem}

\begin{proof}
We study the lattice representation of the left hand side of \eqref{schur-trivial}, shown here for the example $\lambda = (4,4,2,1)$:
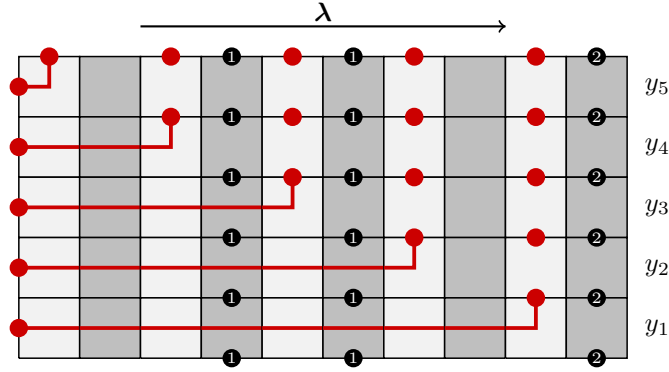
\begin{align*}
\begin{tikzpicture}[scale=0.8]
\foreach\x in {1,...,5}{\node at (11.5,\x+2.5) {$y_\x$};}
\draw[arrow=1,draw=black,thick] (3,8.5) -- (9,8.5);
\node at (6,8.75) {${\bm \lambda}$};
\redoc{1}{3.5}{}
\redoc{1}{4.5}{}
\redoc{1}{5.5}{}
\redoc{1}{6.5}{}
\redoc{1}{7.5}{}
\rminus[1]{1}{7}[0]
\bblank[0]{2}{7}[0]
\rblank[1]{3}{7}[1]
\bblank[1]{4}{7}[1]
\rblank[1]{5}{7}[1]
\bblank[1]{6}{7}[1]
\rblank[1]{7}{7}[1]
\bblank[0]{8}{7}[0]
\rblank[1]{9}{7}[1]
\bblank[2]{10}{7}[2]
\rfull{1}{6}[0]
\darkrfull{2}{6}[0]
\rminus{3}{6}[0]
\bblank{4}{6}[1]
\rblank{5}{6}[1]
\bblank{6}{6}[1]
\rblank{7}{6}[1]
\bblank{8}{6}[0]
\rblank{9}{6}[1]
\bblank{10}{6}[2]
\rfull{1}{5}[0]
\darkrfull{2}{5}[0]
\rfull{3}{5}[0]
\darkrfull{4}{5}[1]
\rminus{5}{5}[0]
\bblank{6}{5}[1]
\rblank{7}{5}[1]
\bblank{8}{5}[0]
\rblank{9}{5}[1]
\bblank{10}{5}[2]
\rfull{1}{4}[0]
\darkrfull{2}{4}[0]
\rfull{3}{4}[0]
\darkrfull{4}{4}[1]
\rfull{5}{4}[0]
\darkrfull{6}{4}[1]
\rminus{7}{4}[0]
\bblank{8}{4}[0]
\rblank{9}{4}[1]
\bblank{10}{4}[2]
\rfull{1}{3}
\darkrfull{2}{3}[0]
\rfull{3}{3}
\darkrfull{4}{3}[1]
\rfull{5}{3}
\darkrfull{6}{3}[1]
\rfull{7}{3}
\darkrfull{8}{3}[0]
\rminus{9}{3}
\bblank{10}{3}[2]
\end{tikzpicture}
\end{align*}
Because the configurations \eqref{forbid} are disallowed, it immediately follows that the lattice is frozen, with black particles propagating in straight lines. The red line in the $i$-th row gives rise to $L+1-i$ OREs, and hence produces the factor 
$y_i^{L+1-i}$, for all $1 \leq i \leq L+1$. All red particles exit from the left edge of the lattice, meaning that the outgoing state is constrained to be 
$\ket{\hole[0.1]} \otimes \ket{{\vec \lambda}}$.

\end{proof}

\subsubsection{Multiple integral expression for $\mathcal{C}^{\lambda}_{\mu\nu}(x)$}

Because of the trivial action \eqref{schur-trivial} of the $\Trs$ operators, the multiple sum \eqref{n-iterat3} can be simplified to
\begin{multline*}
\prod_{i=1}^{\ell+L+1} 
\b{x}_i^{p_i}
\times
\mathcal{C}^{\lambda}_{\mu\nu}(x)
=
\sum_{1 \leq i_1 \leq \kappa_1}
\cdots
\sum_{\substack{1 \leq i_{\ell} \leq \kappa_{\ell} \\ i_\ell \not=i_1,\dots,i_{\ell-1}}}
\prod_{j_1\not=i_1}^{\kappa_1}
\frac{1}{(\b{x}_{i_1}-\b{x}_{j_1})}
\cdots
\prod_{j_\ell\not=i_1,\dots,i_\ell}^{\kappa_\ell}
\frac{1}{(\b{x}_{i_\ell}-\b{x}_{j_\ell})}
\\
\times
\prod_{a=1}^{\ell}
\b{x}_{i_a}^{L+1}
\times
\bra{\hole[0.1]}
\otimes
\bra{{\vec\mu}}
\Tes(x_{i_1})
\dots
\Tes(x_{i_\ell})
\ket{\hole[0.1]}
\otimes
\ket{{\vec \lambda}}.
\end{multline*}
The expectation value that remains, 
$\bra{\hole[0.1]}
\otimes
\bra{{\vec\mu}}
\Tes(x_{i_1})
\dots
\Tes(x_{i_\ell})
\ket{\hole[0.1]}
\otimes
\ket{{\vec \lambda}}$,
has no further dependence on red particles and is therefore purely bosonic. In fact, by taking the $t=0$ specialization of \eqref{skew-HL}, it is simply a skew Schur polynomial:
\begin{align*}
\bra{\hole[0.1]}
\otimes
\bra{{\vec\mu}}
\Tes(x_{i_1})
\dots
\Tes(x_{i_\ell})
\ket{\hole[0.1]}
\otimes
\ket{{\vec \lambda}}
=
s_{\lambda/\mu}(x_{i_1},\dots,x_{i_\ell}).
\end{align*}
Substituting this result into the above sum, we then convert it into a multiple integral:
\begin{align}
\label{schur-mult-int}
\prod_{i=1}^{\ell+L+1} 
\b{x}_i^{p_i}
\times
\mathcal{C}^{\lambda}_{\mu\nu}(x)
=
\oint_{w_1}
\cdots
\oint_{w_\ell}
\frac{
\prod_{1 \leq i<j \leq \ell}
(w_j - w_i)
}
{
\prod_{i=1}^{\ell}
\prod_{j=1}^{\kappa_i}
(w_i - \b{x}_j)
}
\prod_{i=1}^{\ell}
w_i^{L+1}
s_{\lambda/\mu}(\b{w}_1,\dots,\b{w}_\ell),
\end{align}
where each contour of integration is a circle centred on the origin and surrounding the points 
$(\b{x}_1,\dots,\b{x}_{\ell+L+1})$, which are the only poles of the integrand.

\subsubsection{Homogeneous limit}

Now we consider the limits $x_j \rightarrow \infty$ of \eqref{schur-mult-int}. This causes poles to coalesce at the origin, but since the origin is contained within the contours of integration, these limits are clearly non-singular. Together with the arguments of Section \ref{deg-count}, this proves that $\mathcal{C}^{\lambda}_{\mu\nu}(x)$ is proportional to the monomial $\prod_{i=1}^{\ell+L+1} x_i^{p_i}$, and accordingly the right hand side of \eqref{schur-mult-int} is in fact a constant. To evaluate it, we take all such limits $x_j \rightarrow \infty$:
\begin{align*}
\lim_{x \rightarrow \infty}
\prod_{i=1}^{\ell+L+1} 
\b{x}_i^{p_i}
\times
\mathcal{C}^{\lambda}_{\mu\nu}(x)
&=
\oint_{w_1}
\cdots
\oint_{w_\ell}
\prod_{1 \leq i<j \leq \ell}
(w_j - w_i)
\prod_{i=1}^{\ell}
w_i^{L+1-\kappa_i}
s_{\lambda/\mu}(\b{w}_1,\dots,\b{w}_\ell)
\\
&=
{\rm Coeff}
\left[
\prod_{1 \leq i<j \leq \ell}
(w_j-w_i)
s_{\lambda/\mu}(\b{w}_1,\dots,\b{w}_\ell),
\prod_{i=1}^{\ell}
w_i^{\kappa_i-L-2}
\right]
\\
&=
{\rm Coeff}
\left[
\prod_{1 \leq i<j \leq \ell}
(z_i-z_j)
s_{\lambda/\mu}(z_1,\dots,z_\ell),
\prod_{i=1}^{\ell}
z_i^{\ell+L+1-\kappa_i}
\right].
\end{align*}
Finally, recalling that $L+1 - \kappa_i = \nu_i - i$ for all $1 \leq i \leq \ell$, we obtain
\begin{align*}
\lim_{x \rightarrow \infty}
\prod_{i=1}^{\ell+L+1} 
\b{x}_i^{p_i}
\times
\mathcal{C}^{\lambda}_{\mu\nu}(x)
=
{\rm Coeff}
\left[
\prod_{1 \leq i<j \leq \ell}
(z_i-z_j)
s_{\lambda/\mu}(z_1,\dots,z_\ell),
\prod_{i=1}^{\ell}
z_i^{\ell+\nu_i-i}
\right]
=
c^{\lambda}_{\mu\nu},
\end{align*}
with the final equality being simply the $t=0$ specialization of \eqref{CT-Hall}.

\subsection{Right-triangle puzzles}

We now analyse the combinatorial implication of Theorem \ref{thm-LR}. This is a much simpler task than in previous sections, since \eqref{LR-schur} contains no divided-difference operators and is manifestly positive. The partition function \eqref{LR-PF-schur} consists of exactly $c^{\lambda}_{\mu\nu}$ configurations. For instance, returning to Example \ref{ex-C}, one finds that only two lattice configurations are possible:
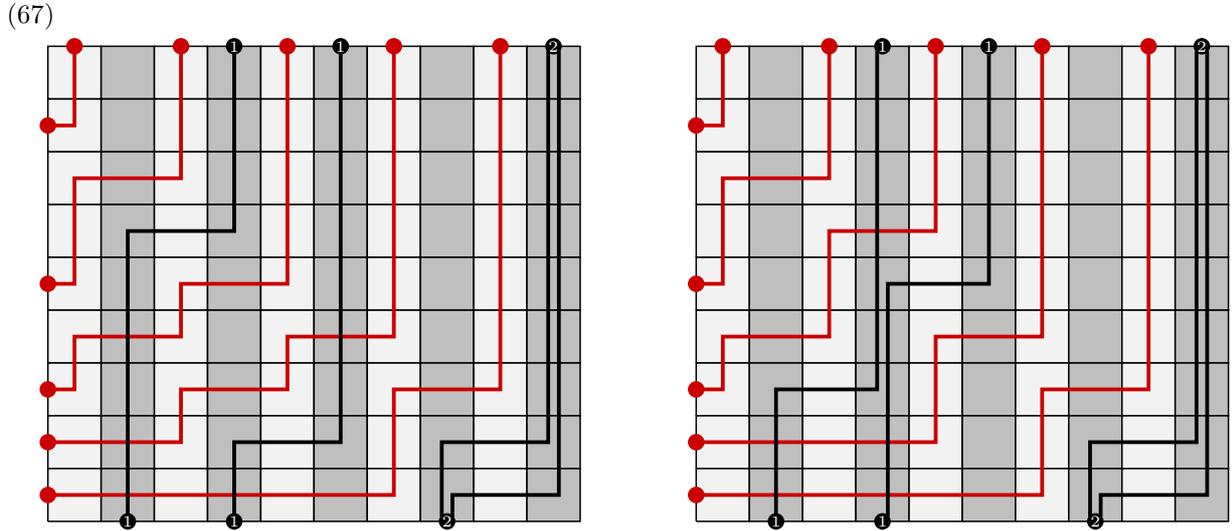
\begin{align}
\label{2configs}
\begin{tikzpicture}[scale=0.7]
\redoc{1}{-0.5}{}
\redoc{1}{0.5}{}
\redoc{1}{1.5}{}
\redoc{1}{3.5}{}
\redoc{1}{6.5}{}
\rblank[1]{1}{7}
\bblank[0]{2}{7}
\rblank[1]{3}{7}
\bblank[1]{4}{7}
\rblank[1]{5}{7}
\bblank[1]{6}{7}
\rblank[1]{7}{7}
\bblank[0]{8}{7}
\rblank[1]{9}{7}
\bblank[2]{10}{7}
\rblank{1}{6}
\bblank{2}{6}
\rblank{3}{6}
\bblank{4}{6}
\rblank{5}{6}
\bblank{6}{6}
\rblank{7}{6}
\bblank{8}{6}
\rblank{9}{6}
\bblank{10}{6}
\rblank{1}{5}
\bblank{2}{5}
\rblank{3}{5}
\bblank{4}{5}
\rblank{5}{5}
\bblank{6}{5}
\rblank{7}{5}
\bblank{8}{5}
\rblank{9}{5}
\bblank{10}{5}
\rblank{1}{4}
\bblank{2}{4}
\rblank{3}{4}
\bblank{4}{4}
\rblank{5}{4}
\bblank{6}{4}
\rblank{7}{4}
\bblank{8}{4}
\rblank{9}{4}
\bblank{10}{4}
\rblank{1}{3}
\bblank{2}{3}
\rblank{3}{3}
\bblank{4}{3}
\rblank{5}{3}
\bblank{6}{3}
\rblank{7}{3}
\bblank{8}{3}
\rblank{9}{3}
\bblank{10}{3}
\rblank{1}{2}
\bblank{2}{2}
\rblank{3}{2}
\bblank{4}{2}
\rblank{5}{2}
\bblank{6}{2}
\rblank{7}{2}
\bblank{8}{2}
\rblank{9}{2}
\bblank{10}{2}
\rblank{1}{1}
\bblank{2}{1}
\rblank{3}{1}
\bblank{4}{1}
\rblank{5}{1}
\bblank{6}{1}
\rblank{7}{1}
\bblank{8}{1}
\rblank{9}{1}
\bblank{10}{1}
\rblank{1}{0}
\bblank{2}{0}
\rblank{3}{0}
\bblank{4}{0}
\rblank{5}{0}
\bblank{6}{0}
\rblank{7}{0}
\bblank{8}{0}
\rblank{9}{0}
\bblank{10}{0}
\rblank{1}{-1}[0]
\bblank{2}{-1}[1]
\rblank{3}{-1}[0]
\bblank{4}{-1}[1]
\rblank{5}{-1}[0]
\bblank{6}{-1}[0]
\rblank{7}{-1}[0]
\bblank{8}{-1}[2]
\rblank{9}{-1}[0]
\bblank{10}{-1}[0]
\draw[rline] (1,6.5) -- (1.5,6.5) -- (1.5,8);
\draw[rline] (1,3.5) -- (1.5,3.5) -- (1.5,5.5) -- (3.5,5.5) -- (3.5,8);
\draw[rline] (1,1.5) -- (1.5,1.5) -- (1.5,2.5) -- (3.5,2.5) -- (3.5,3.5) -- (5.5,3.5) -- (5.5,8);
\draw[rline] (1,0.5) -- (3.5,0.5) -- (3.5,1.5) -- (5.5,1.5) -- (5.5,2.5) -- (7.5,2.5) -- (7.5,8);
\draw[rline] (1,-0.5) -- (7.5,-0.5) -- (7.5,1.5) -- (9.5,1.5) -- (9.5,8);
\draw[bline] (2.5,-1+0.12) -- (2.5,4.5) -- (4.5,4.5) -- (4.5,8-0.12);
\draw[bline] (4.5,-1+0.12) -- (4.5,0.5) -- (6.5,0.5) -- (6.5,8-0.12);
\draw[bline] (8.4,-1+0.09) -- (8.4,0.5) -- (10.4,0.5) -- (10.4,8-0.09);
\draw[bline] (8.6,-1+0.09) -- (8.6,-0.5) -- (10.6,-0.5) -- (10.6,8-0.09);
\end{tikzpicture}
\quad\quad\quad\quad
\begin{tikzpicture}[scale=0.7]
\redoc{1}{-0.5}{}
\redoc{1}{0.5}{}
\redoc{1}{1.5}{}
\redoc{1}{3.5}{}
\redoc{1}{6.5}{}
\rblank[1]{1}{7}
\bblank[0]{2}{7}
\rblank[1]{3}{7}
\bblank[1]{4}{7}
\rblank[1]{5}{7}
\bblank[1]{6}{7}
\rblank[1]{7}{7}
\bblank[0]{8}{7}
\rblank[1]{9}{7}
\bblank[2]{10}{7}
\rblank{1}{6}
\bblank{2}{6}
\rblank{3}{6}
\bblank{4}{6}
\rblank{5}{6}
\bblank{6}{6}
\rblank{7}{6}
\bblank{8}{6}
\rblank{9}{6}
\bblank{10}{6}
\rblank{1}{5}
\bblank{2}{5}
\rblank{3}{5}
\bblank{4}{5}
\rblank{5}{5}
\bblank{6}{5}
\rblank{7}{5}
\bblank{8}{5}
\rblank{9}{5}
\bblank{10}{5}
\rblank{1}{4}
\bblank{2}{4}
\rblank{3}{4}
\bblank{4}{4}
\rblank{5}{4}
\bblank{6}{4}
\rblank{7}{4}
\bblank{8}{4}
\rblank{9}{4}
\bblank{10}{4}
\rblank{1}{3}
\bblank{2}{3}
\rblank{3}{3}
\bblank{4}{3}
\rblank{5}{3}
\bblank{6}{3}
\rblank{7}{3}
\bblank{8}{3}
\rblank{9}{3}
\bblank{10}{3}
\rblank{1}{2}
\bblank{2}{2}
\rblank{3}{2}
\bblank{4}{2}
\rblank{5}{2}
\bblank{6}{2}
\rblank{7}{2}
\bblank{8}{2}
\rblank{9}{2}
\bblank{10}{2}
\rblank{1}{1}
\bblank{2}{1}
\rblank{3}{1}
\bblank{4}{1}
\rblank{5}{1}
\bblank{6}{1}
\rblank{7}{1}
\bblank{8}{1}
\rblank{9}{1}
\bblank{10}{1}
\rblank{1}{0}
\bblank{2}{0}
\rblank{3}{0}
\bblank{4}{0}
\rblank{5}{0}
\bblank{6}{0}
\rblank{7}{0}
\bblank{8}{0}
\rblank{9}{0}
\bblank{10}{0}
\rblank{1}{-1}[0]
\bblank{2}{-1}[1]
\rblank{3}{-1}[0]
\bblank{4}{-1}[1]
\rblank{5}{-1}[0]
\bblank{6}{-1}[0]
\rblank{7}{-1}[0]
\bblank{8}{-1}[2]
\rblank{9}{-1}[0]
\bblank{10}{-1}[0]
\draw[rline] (1,6.5) -- (1.5,6.5) -- (1.5,8);
\draw[rline] (1,3.5) -- (1.5,3.5) -- (1.5,5.5) -- (3.5,5.5) -- (3.5,8);
\draw[rline] (1,1.5) -- (1.5,1.5) -- (1.5,2.5) -- (3.5,2.5) -- (3.5,4.5) -- (5.5,4.5) -- (5.5,8);
\draw[rline] (1,0.5) -- (5.5,0.5) -- (5.5,2.5) -- (7.5,2.5) -- (7.5,8);
\draw[rline] (1,-0.5) -- (7.5,-0.5) -- (7.5,1.5) -- (9.5,1.5) -- (9.5,8);
\draw[bline] (2.5,-1+0.09) -- (2.5,1.5) -- (4.4,1.5) -- (4.4,8-0.09);
\draw[bline] (4.6,-1+0.09) -- (4.6,3.5) -- (6.5,3.5) -- (6.5,8-0.12);
\draw[bline] (8.4,-1+0.09) -- (8.4,0.5) -- (10.4,0.5) -- (10.4,8-0.09);
\draw[bline] (8.6,-1+0.09) -- (8.6,-0.5) -- (10.6,-0.5) -- (10.6,8-0.09);
\end{tikzpicture}
\end{align}
and accordingly $c^{\lambda}_{\mu\nu} = 2$ in this case. The fact that the lattice configurations of 
$\mathcal{C}^{\lambda}_{\mu\nu}(x)$ enumerate Littlewood--Richardson coefficients suggests that the former can be put in one-to-one correspondence with Knutson--Tao puzzles. To demonstrate that this is indeed the case, we first define an intermediate class of combinatorial objects, which we call \textit{right-triangle puzzles}.

\begin{defn}{\rm

Consider tilings of the square lattice by the following set of tiles:
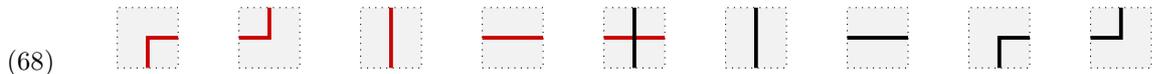
\begin{align}
\label{puzzle-pieces}
\begin{tikzpicture}[scale=0.8]
\rp{0}{0}
\rmin{2}{0}
\rver{4}{0}
\rhor{6}{0}
\rb{8}{0}
\bver{10}{0}
\bhor{12}{0}
\bp{14}{0}
\bmin{16}{0}
\end{tikzpicture}
\end{align}
Let $\lambda,\mu,\nu$ be three partitions as defined in Theorem \ref{thm-LR}. A right-triangle puzzle with frame 
$(\lambda,\mu,\nu)$ is a tiling of an $(\ell+L+1) \times (\ell+L+1)$ lattice by the tiles \eqref{puzzle-pieces}, where the boundary conditions are fixed as follows:
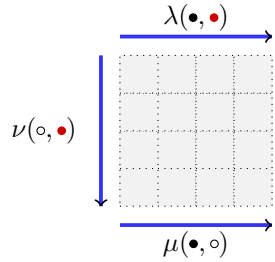
\begin{align*}
\begin{tikzpicture}[scale=0.5]
\draw[bbline,arrow=1] (0.5,5) -- (0.5,1); \node at (-1,3) {$\nu(\hole[0.1],\rpart[0.1])$};
\draw[bbline,arrow=1] (1,0.5) -- (5,0.5); \node at (3,6) {$\lambda(\bpart[0.1],\rpart[0.1])$};
\draw[bbline,arrow=1] (1,5.5) -- (5,5.5); \node at (3,0) {$\mu(\bpart[0.1],\hole[0.1])$};
\clear{1}{1}
\clear{2}{1}
\clear{3}{1}
\clear{4}{1}
\clear{1}{2}
\clear{2}{2}
\clear{3}{2}
\clear{4}{2}
\clear{1}{3}
\clear{2}{3}
\clear{3}{3}
\clear{4}{3}
\clear{1}{4}
\clear{2}{4}
\clear{3}{4}
\clear{4}{4}
\end{tikzpicture}
\end{align*}
In other words, the top, bottom and left boundaries are truncated Maya diagrams formed by the two types of particles in parentheses, representing the partition indicated. These puzzles take their name from the fact that, when studying configurations of the lattice above, only the region below the NW--SE diagonal is non-trivial (the region above this diagonal is necessarily frozen).
}
\end{defn}

\begin{prop}
\label{thm-LR-puzzle}
{\rm
The Littlewood--Richardson coefficient $c^{\lambda}_{\mu\nu}$ is equal to the number of right-triangle puzzles with frame $(\lambda,\mu,\nu)$.
}
\end{prop}

\begin{proof}
There is a simple bijection between the lattice configurations encoded by \eqref{LR-PF-schur} and right-triangle puzzles. We will sketch the general idea here, but avoid a rigorous proof, since this would require a long detour. 

We begin by noting that, for the two classes of objects, the underlying lattice has a different size. The lattice configurations of $\mathcal{C}^{\lambda}_{\mu\nu}(x)$ are $(\ell+L+1) \times (2L+2)$, while right-triangle puzzles are $(\ell+L+1) \times (\ell+L+1)$. Both objects have the same number of rows, but configurations of $\mathcal{C}^{\lambda}_{\mu\nu}(x)$ may contain more or less columns than their counterpart right-triangle puzzle. The bijection that we seek should therefore be viewed as contraction or expansion in the horizontal dimension, while no transformation takes place in the vertical dimension. 

To go from lattice configurations in $\mathcal{C}^{\lambda}_{\mu\nu}(x)$, for example those of \eqref{2configs}, to right-triangle puzzles, we allow all vertical line segments the freedom to be translated left or right by moves. The horizontal line segments are free to contract or expand during this process, such that the connectivity between vertical line segments is preserved. A move of a vertical line segment is permitted {\it only} if it does not introduce new crossings between lattice lines, or undo a crossing already present. Using these moves, given a configuration in $\mathcal{C}^{\lambda}_{\mu\nu}(x)$ there is a unique way to reposition vertical line segments within an $(\ell+L+1) \times (\ell+L+1)$ grid, such that only the tiles \eqref{puzzle-pieces} are used. When this procedure is applied to the configurations in \eqref{2configs}, the result is
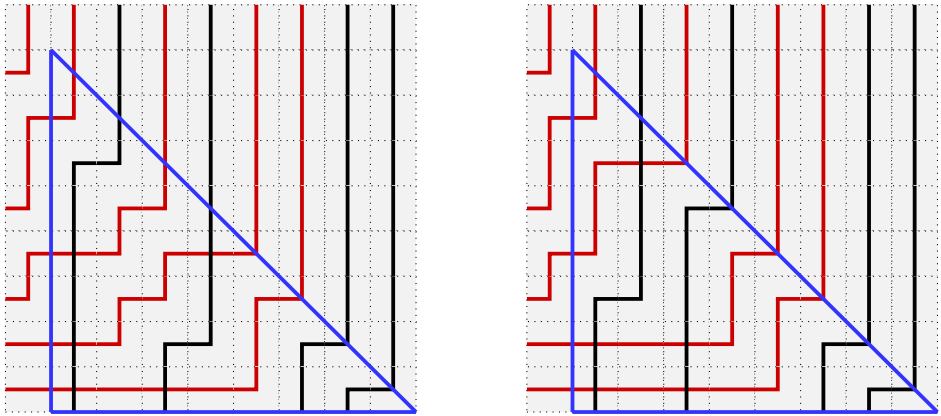
\begin{align*}
\begin{tikzpicture}[scale=0.6]
\rver{1}{9}
\rver{2}{9}
\bver{3}{9}
\rver{4}{9}
\bver{5}{9}
\rver{6}{9}
\rver{7}{9}
\bver{8}{9}
\bver{9}{9}
\rmin{1}{8}
\rver{2}{8}
\bver{3}{8}
\rver{4}{8}
\bver{5}{8}
\rver{6}{8}
\rver{7}{8}
\bver{8}{8}
\bver{9}{8}
\rp{1}{7}
\rmin{2}{7}
\bver{3}{7}
\rver{4}{7}
\bver{5}{7}
\rver{6}{7}
\rver{7}{7}
\bver{8}{7}
\bver{9}{7}
\rver{1}{6}
\bp{2}{6}
\bmin{3}{6}
\rver{4}{6}
\bver{5}{6}
\rver{6}{6}
\rver{7}{6}
\bver{8}{6}
\bver{9}{6}
\rmin{1}{5}
\bver{2}{5}
\rp{3}{5}
\rmin{4}{5}
\bver{5}{5}
\rver{6}{5}
\rver{7}{5}
\bver{8}{5}
\bver{9}{5}
\rp{1}{4}
\rb{2}{4}
\rmin{3}{4}
\rp{4}{4}
\rb{5}{4}
\rmin{6}{4}
\rver{7}{4}
\bver{8}{4}
\bver{9}{4}
\rmin{1}{3}
\bver{2}{3}
\rp{3}{3}
\rmin{4}{3}
\bver{5}{3}
\rp{6}{3}
\rmin{7}{3}
\bver{8}{3}
\bver{9}{3}
\rhor{1}{2}
\rb{2}{2}
\rmin{3}{2}
\bp{4}{2}
\bmin{5}{2}
\rver{6}{2}
\bp{7}{2}
\bmin{8}{2}
\bver{9}{2}
\rhor{1}{1}
\rb{2}{1}
\rhor{3}{1}
\rb{4}{1}
\rhor{5}{1}
\rmin{6}{1}
\bver{7}{1}
\bp{8}{1}
\bmin{9}{1}
\draw[bbline] (2,9) -- (10,1);
\draw[bbline] (2,1) -- (10,1);
\draw[bbline] (2,1) -- (2,9);
\end{tikzpicture}
\quad\quad\quad\quad
\begin{tikzpicture}[scale=0.6]
\rver{1}{9}
\rver{2}{9}
\bver{3}{9}
\rver{4}{9}
\bver{5}{9}
\rver{6}{9}
\rver{7}{9}
\bver{8}{9}
\bver{9}{9}
\rmin{1}{8}
\rver{2}{8}
\bver{3}{8}
\rver{4}{8}
\bver{5}{8}
\rver{6}{8}
\rver{7}{8}
\bver{8}{8}
\bver{9}{8}
\rp{1}{7}
\rmin{2}{7}
\bver{3}{7}
\rver{4}{7}
\bver{5}{7}
\rver{6}{7}
\rver{7}{7}
\bver{8}{7}
\bver{9}{7}
\rver{1}{6}
\rp{2}{6}
\rb{3}{6}
\rmin{4}{6}
\bver{5}{6}
\rver{6}{6}
\rver{7}{6}
\bver{8}{6}
\bver{9}{6}
\rmin{1}{5}
\rver{2}{5}
\bver{3}{5}
\bp{4}{5}
\bmin{5}{5}
\rver{6}{5}
\rver{7}{5}
\bver{8}{5}
\bver{9}{5}
\rp{1}{4}
\rmin{2}{4}
\bver{3}{4}
\bver{4}{4}
\rp{5}{4}
\rmin{6}{4}
\rver{7}{4}
\bver{8}{4}
\bver{9}{4}
\rmin{1}{3}
\bp{2}{3}
\bmin{3}{3}
\bver{4}{3}
\rver{5}{3}
\rp{6}{3}
\rmin{7}{3}
\bver{8}{3}
\bver{9}{3}
\rhor{1}{2}
\rb{2}{2}
\rhor{3}{2}
\rb{4}{2}
\rmin{5}{2}
\rver{6}{2}
\bp{7}{2}
\bmin{8}{2}
\bver{9}{2}
\rhor{1}{1}
\rb{2}{1}
\rhor{3}{1}
\rb{4}{1}
\rhor{5}{1}
\rmin{6}{1}
\bver{7}{1}
\bp{8}{1}
\bmin{9}{1}
\draw[bbline] (2,9) -- (10,1);
\draw[bbline] (2,1) -- (10,1);
\draw[bbline] (2,1) -- (2,9);
\end{tikzpicture}
\end{align*}
which are the corresponding right-triangle puzzles. The non-trivial region of each puzzle is indicated by the blue triangle; outside of this triangle, the lattice paths are frozen.

It is straightforward to proceed in the opposite sense, from right-triangle puzzles to configurations of $\mathcal{C}^{\lambda}_{\mu\nu}(x)$, subject to the same rules of movement. Given any right-triangle puzzle, there is again a unique way to perform these moves, such that a configuration of 
$\mathcal{C}^{\lambda}_{\mu\nu}(x)$ on a $(\ell+L+1) \times (2L+2)$ grid is obtained.
\end{proof}

\subsection{Knutson--Tao puzzles}

Finally, we give the trivial correspondence between right-triangle puzzles and the ones introduced in \cite{knu-tao,knu-tao-woo}. One extracts the non-trivial region of a right-triangle puzzle and shears it so that it takes the shape of an equilateral triangle. This shearing, when applied to the right-triangle puzzles above, produces
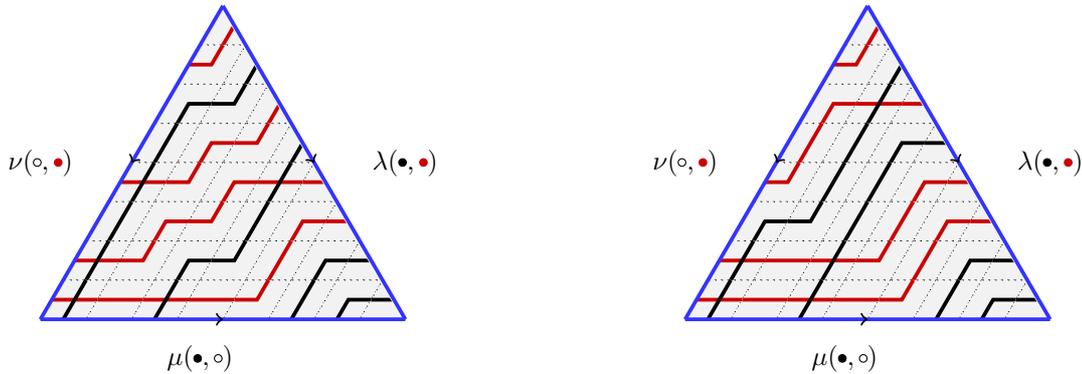
\begin{align*}
\quad\quad\quad\quad
\begin{tikzpicture}[scale=0.6,distort]
\rver{2}{8}
\bver{3}{8}
\rver{4}{8}
\bver{5}{8}
\rver{6}{8}
\rver{7}{8}
\bver{8}{8}
\bver{9}{8}
\rmin{2}{7}
\bver{3}{7}
\rver{4}{7}
\bver{5}{7}
\rver{6}{7}
\rver{7}{7}
\bver{8}{7}
\bver{9}{7}
\bp{2}{6}
\bmin{3}{6}
\rver{4}{6}
\bver{5}{6}
\rver{6}{6}
\rver{7}{6}
\bver{8}{6}
\bver{9}{6}
\bver{2}{5}
\rp{3}{5}
\rmin{4}{5}
\bver{5}{5}
\rver{6}{5}
\rver{7}{5}
\bver{8}{5}
\bver{9}{5}
\rb{2}{4}
\rmin{3}{4}
\rp{4}{4}
\rb{5}{4}
\rmin{6}{4}
\rver{7}{4}
\bver{8}{4}
\bver{9}{4}
\bver{2}{3}
\rp{3}{3}
\rmin{4}{3}
\bver{5}{3}
\rp{6}{3}
\rmin{7}{3}
\bver{8}{3}
\bver{9}{3}
\rb{2}{2}
\rmin{3}{2}
\bp{4}{2}
\bmin{5}{2}
\rver{6}{2}
\bp{7}{2}
\bmin{8}{2}
\bver{9}{2}
\rb{2}{1}
\rhor{3}{1}
\rb{4}{1}
\rhor{5}{1}
\rmin{6}{1}
\bver{7}{1}
\bp{8}{1}
\bmin{9}{1}
\filldraw[fill=white,draw=white,line width=0.7mm] (2,9) -- (10,1) -- (10,9) -- (2,9);
\draw[bbline,arrow=0.5] (2,9) -- (10,1);
\draw[bbline,arrow=0.5] (2,1) -- (10,1);
\draw[bbline,arrow=0.5] (2,9) -- (2,1);
\node[text centered] at (0,5) {$\nu(\hole[0.1],\rpart[0.1])$};
\node[text centered] at (8,5) {$\lambda(\bpart[0.1],\rpart[0.1])$};
\node[text centered] at (6,0) {$\mu(\bpart[0.1],\hole[0.1])$};
\end{tikzpicture}
\quad\quad
\begin{tikzpicture}[scale=0.6,distort]
\rver{2}{8}
\bver{3}{8}
\rver{4}{8}
\bver{5}{8}
\rver{6}{8}
\rver{7}{8}
\bver{8}{8}
\bver{9}{8}
\rmin{2}{7}
\bver{3}{7}
\rver{4}{7}
\bver{5}{7}
\rver{6}{7}
\rver{7}{7}
\bver{8}{7}
\bver{9}{7}
\rp{2}{6}
\rb{3}{6}
\rmin{4}{6}
\bver{5}{6}
\rver{6}{6}
\rver{7}{6}
\bver{8}{6}
\bver{9}{6}
\rver{2}{5}
\bver{3}{5}
\bp{4}{5}
\bmin{5}{5}
\rver{6}{5}
\rver{7}{5}
\bver{8}{5}
\bver{9}{5}
\rmin{2}{4}
\bver{3}{4}
\bver{4}{4}
\rp{5}{4}
\rmin{6}{4}
\rver{7}{4}
\bver{8}{4}
\bver{9}{4}
\bp{2}{3}
\bmin{3}{3}
\bver{4}{3}
\rver{5}{3}
\rp{6}{3}
\rmin{7}{3}
\bver{8}{3}
\bver{9}{3}
\rb{2}{2}
\rhor{3}{2}
\rb{4}{2}
\rmin{5}{2}
\rver{6}{2}
\bp{7}{2}
\bmin{8}{2}
\bver{9}{2}
\rb{2}{1}
\rhor{3}{1}
\rb{4}{1}
\rhor{5}{1}
\rmin{6}{1}
\bver{7}{1}
\bp{8}{1}
\bmin{9}{1}
\filldraw[fill=white,draw=white,line width=0.7mm] (2,9) -- (10,1) -- (10,9) -- (2,9);
\draw[bbline,arrow=0.5] (2,9) -- (10,1);
\draw[bbline,arrow=0.5] (2,1) -- (10,1);
\draw[bbline,arrow=0.5] (2,9) -- (2,1);
\node[text centered] at (0,5) {$\nu(\hole[0.1],\rpart[0.1])$};
\node[text centered] at (8,5) {$\lambda(\bpart[0.1],\rpart[0.1])$};
\node[text centered] at (6,0) {$\mu(\bpart[0.1],\hole[0.1])$};
\end{tikzpicture}
\end{align*}
which are essentially of the same form as the puzzles in \cite{z-j}, except that our tiles are decorated in a different way. Right-triangle puzzles are thus equivalent to tilings of an equilateral triangle by the tiles
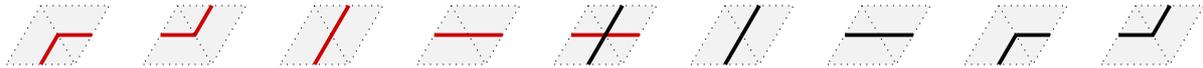
\begin{align}
\begin{tikzpicture}[scale=0.9,distort,baseline=0.5cm]
\rp{0}{0} \draw[dotted] (0,1)--(1,0);
\rmin{2}{0}  \draw[dotted] (2,1)--(3,0);
\rver{4}{0}  \draw[dotted] (4,1)--(5,0);
\rhor{6}{0}  \draw[dotted] (6,1)--(7,0);
\rb{8}{0}  \draw[dotted] (8,1)--(9,0);
\bver{10}{0}  \draw[dotted] (10,1)--(11,0);
\bhor{12}{0} \draw[dotted] (12,1)--(13,0);
\bp{14}{0}  \draw[dotted] (14,1)--(15,0);
\bmin{16}{0}  \draw[dotted] (16,1)--(17,0);
\end{tikzpicture}
\end{align}
where we divide each rhombus into its constituent equilateral triangles. Since the internal decoration of each tile is irrelevant, we may suppress it and instead label the edge states of the tiles by $+,0,-$ (as long as this is done in a self-consistent way). Doing so, we obtain the tiles 
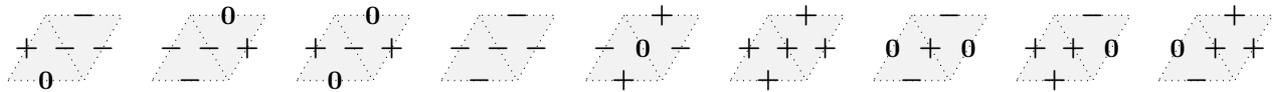
\begin{align}
\begin{tikzpicture}[scale=1,distort,baseline=0.5cm]
\wwlight{0}{0}{0}{-}{-}{+} \draw[dotted] (0,1) -- node{${\bm-}$} (1,0);
\end{tikzpicture}
\quad%
\begin{tikzpicture}[scale=1,distort,baseline=0.5cm]
\wwlight{2}{0}{-}{+}{0}{-} \draw[dotted] (2,1) -- node{${\bm-}$} (3,0);
\end{tikzpicture}
\quad%
\begin{tikzpicture}[scale=1,distort,baseline=0.5cm]
\wwlight{4}{0}{0}{+}{0}{+} \draw[dotted] (4,1) -- node{${\bm-}$} (5,0);
\end{tikzpicture}
\quad%
\begin{tikzpicture}[scale=1,distort,baseline=0.5cm]
\wwlight{6}{0}{-}{-}{-}{-} \draw[dotted] (6,1) -- node{${\bm-}$} (7,0);
\end{tikzpicture}
\quad%
\begin{tikzpicture}[scale=1,distort,baseline=0.5cm]
\wwlight{8}{0}{+}{-}{+}{-} \draw[dotted] (8,1) -- node{${\bm0}$} (9,0);
\end{tikzpicture}
\quad%
\begin{tikzpicture}[scale=1,distort,baseline=0.5cm]
\wwlight{10}{0}{+}{+}{+}{+} \draw[dotted] (10,1) -- node{${\bm+}$} (11,0);
\end{tikzpicture}
\quad%
\begin{tikzpicture}[scale=1,distort,baseline=0.5cm]
\wwlight{12}{0}{-}{0}{-}{0} \draw[dotted] (12,1) -- node{${\bm+}$} (13,0);
\end{tikzpicture}
\quad%
\begin{tikzpicture}[scale=1,distort,baseline=0.5cm]
\wwlight{16}{0}{+}{0}{-}{+} \draw[dotted] (16,1) -- node{${\bm+}$} (17,0);
\end{tikzpicture}
\quad%
\begin{tikzpicture}[scale=1,distort,baseline=0.5cm]
\wwlight{14}{0}{-}{+}{+}{0} \draw[dotted] (14,1) -- node{${\bm+}$} (15,0);
\end{tikzpicture}
\end{align}
By studying the possible connectivities of the tiles above (and ensuring that $0$ states are made internal), we find that this is equivalent to a tiling by only seven tiles:
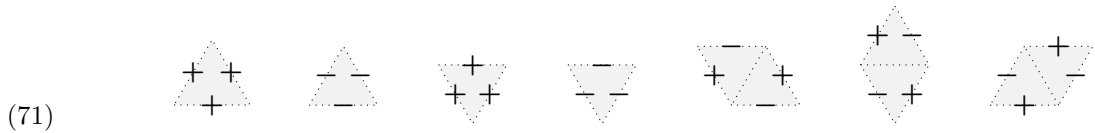
\begin{align}
\begin{tikzpicture}[scale=1,distort]
\draw[dotted,bgplaq] (1,1) -- node{${\bm+}$} (2,1) -- node{${\bm+}$} (1,2) -- node{${\bm+}$} (1,1);
\end{tikzpicture}
\qquad%
\begin{tikzpicture}[scale=0.9,distort]
\draw[dotted,bgplaq] (3,1) -- node{${\bm-}$} (4,1) -- node{${\bm-}$} (3,2) -- node{${\bm-}$} (3,1);
\end{tikzpicture}
\qquad%
\begin{tikzpicture}[scale=0.9,distort]
\draw[dotted,bgplaq] (5,2) -- node{${\bm+}$} (6,2) -- node{${\bm+}$} (6,1) -- node{${\bm+}$} (5,2);
\end{tikzpicture}
\qquad%
\begin{tikzpicture}[scale=0.9,distort]
\draw[dotted,bgplaq] (7,2) -- node{${\bm-}$} (8,2) -- node{${\bm-}$} (8,1) -- node{${\bm-}$} (7,2);
\end{tikzpicture}
\qquad%
\begin{tikzpicture}[scale=0.9,distort]
\draw[dotted,bgplaq] (9,2)  -- node{${\bm-}$} (10,2)  -- node{${\bm+}$} (11,1)  -- node{${\bm-}$} (10,1)  -- node{${\bm+}$} (9,2);
\draw[dotted] (10,2) -- (10,1);
\end{tikzpicture}
\qquad%
\begin{tikzpicture}[scale=0.9,distort]
\draw[dotted,bgplaq] (11.5,2.5) -- node{${\bm-}$} (12.5,1.5) -- node{${\bm+}$} (12.5,0.5) -- node{${\bm-}$} (11.5,1.5) -- node{${\bm+}$} (11.5,2.5);
\draw[dotted] (11.5,1.5) -- (12.5,1.5);
\end{tikzpicture}
\qquad%
\begin{tikzpicture}[scale=0.9,distort]
\draw[dotted,bgplaq] (14,1) -- node{${\bm+}$} (15,1) -- node{${\bm-}$} (15,2) -- node{${\bm+}$} (14,2) -- node{${\bm-}$} (14,1);
\draw[dotted] (15,1) -- (14,2);
\end{tikzpicture}
\end{align}
which recovers the result of \cite{knu-tao,knu-tao-woo}.

\section{Discussion}

We conclude with a number of comments and open questions.

{\bf 1.}
In this paper we have restricted our attention to the coproduct rules \eqref{coproducts} as a means of calculating various structure constants. It would be even more satisfying to have a direct proof of the product rules 
\eqref{HL-prod} and \eqref{schur-HL-prod} themselves, using nothing more than the rank-two integrable models in Sections \ref{model-2-bosons} and \ref{fer-bos-model} and repeated use of the associated intertwining equations \eqref{inter-2-bosons} and \eqref{int-bf}. Such a goal was already achieved in \cite{z-j} for the Schur product rule \eqref{schur-prod}, allowing it to be generalized to factorial Schur polynomials, and thus leading to a solution of the Molev--Sagan problem \cite{mol-sag}.  

{\bf 2.}
The connection between the Bethe Ansatz and expansion coefficients of symmetric functions is certainly not an original one. In the case of the Kostka polynomials, such an idea dates back to \cite{kir-res}. More conceptually related to the present work, Korff has used the $t$-boson model to obtain both the Kostka polynomials and their inverses as expectation values of non-commutative symmetric functions \cite{korff}. Our work also appears closely related with the recent preprint of Borodin and Petrov \cite{bor-pet}, in which appear a host of multiple integrals of very similar type to those in the present paper; it would be very rewarding to connect the two works.   

{\bf 3.}
It is mysterious that the natural $t$-generalization of $\mathcal{C}^{\lambda}_{\mu\nu}(x)$, obtained by simply retaining a non-zero $t$ parameter in the fermion-boson model, does not seem to admit any clear interpretation. It is tempting to speculate, however, that it should be related to the last of the possible product rules between Hall--Littlewood and Schur polynomials:
\begin{align}
\label{generalized-kost}
P_{\mu} s_{\nu} = \sum_{\lambda} K^{\lambda}_{\mu\nu}(t) P_{\lambda}.
\end{align}
We base this speculation on the fact that {\bf a)} The coproduct version of \eqref{generalized-kost} is $Q_{\lambda/\mu} = \sum_{\nu} K^{\lambda}_{\mu\nu}(t) S_{\nu}$ (this can be derived from similar reasoning as in Section \ref{gen-inv-kost}). Hence $K^{\lambda}_{\mu\nu}(t)$ can be calculated by expanding a skew Hall--Littlewood polynomial in the $t$-Schur basis; {\bf b)} Lifting the partition function \eqref{LR-PF-schur} to generic $t$, it clearly has the correct structure to perform such an expansion. The bosonic (black) particles will give rise to $Q_{\lambda/\mu}$, while the fermionic (red) particles should in turn allow the expansion in terms of $S_{\nu}$.

{\bf 4.}
It would be very desirable to extend our results to the level of Macdonald polynomials, and calculate (for example) the expansion coefficients in the equations
\begin{align*}
P_{\mu}(x;q,t) P_{\nu}(x;q,t) = \sum_{\lambda} f^{\lambda}_{\mu\nu}(q,t) P_{\lambda}(x;q,t),
\quad
c_{\nu}(q,t) P_{\nu}(x;q,t) = \sum_{\lambda} K^{\lambda}_{\nu}(q,t) S_{\lambda}(x;t).
\end{align*}
Finding a combinatorial rule for $f^{\lambda}_{\mu\nu}(q,t)$ and proving that 
$K^{\lambda}_{\nu}(q,t) \in \mathbb{N}[q,t]$ were celebrated problems in Macdonald theory, which were ultimately resolved in \cite{yip} and \cite{hai}. Can the methods of integrable models and the Bethe Ansatz be applied to these problems? For the moment, this remains a formidable challenge, considering that even the Macdonald polynomials themselves are difficult to construct within the framework of quantum integrability, although recently progress has been made in this direction \cite{cd-gw}.

{\bf 5.}
Because of the Macdonald involution, it can be shown that the structure constants $f^{\lambda}_{\mu\nu}(q,t)$ possess the following symmetry (see Equation (7.3) in Section 7, Chapter VI of \cite{mac}):
\begin{align*}
f^{\lambda}_{\mu\nu}(q,t)
=
f^{\lambda'}_{\mu'\nu'}(t,q)
\frac{b_{\lambda}(q,t)}{b_{\mu}(q,t)b_{\nu}(q,t)},
\end{align*}
where $\lambda',\mu',\nu'$ denote the conjugate of the original partitions. By taking $q \rightarrow 0$, this reduces to
\begin{align*}
f^{\lambda}_{\mu\nu}(t)
=
f^{\lambda'}_{\mu'\nu'}(t,0)
\frac{b_{\lambda}(t)}{b_{\mu}(t)b_{\nu}(t)},
\end{align*}
directly relating the Hall polynomials and the structure constants of the $t$--Whittaker polynomials. In view of this, \eqref{main-formula} can be dually interpreted as a formula for $f^{\lambda'}_{\mu'\nu'}(t,0)$.

\section*{Acknowledgments}

MW is supported by the ARC grant DE160100958 and the ARC Centre of Excellence for Mathematical and Statistical Frontiers (ACEMS). MW would like to thank Eric Ragoucy for kind hospitality and stimulating discussions at the Laboratoire d'Annecy-le-Vieux de Physique Th\'eorique (LAPTh), while this manuscript was in preparation. PZJ is supported by ERC grant ``LIC'' 278124 and ARC grant DP140102201.

\appendix

\section{Examples of Theorem \ref{thm-hall-puzzle}}
\label{A}

\subsection{The case \texorpdfstring{$\b\lambda = (4,1,1,1), \b\mu = (3,1,1,0), \b\nu = (2)$}{bar lambda=(4,1,1,1), bar mu=(3,1,1,0), bar nu=(2)}}

The partitions $\b\lambda,\b\mu,\b\nu$ are the complements of $\lambda = (3,3,3,0), \mu = (4,3,3,1), \nu = (2)$ by 4. The playing field is thus
\begin{align*}
\begin{tikzpicture}[scale=0.5]
\plusb{0}{4}
\plusb{0}{5}
\bdark{0}{6}
\blight{1}{6}
\bdark{2}{6}
\blight{3}{6}
\bdark{4}{6}
\blight{5}{6}
\bdark{6}{6}
\blight{7}{6}
\bdark{8}{6}
\blight{9}{6}
\bdark{10}{6}
\blight{1}{5}
\bdark{2}{5}
\blight{3}{5}
\bdark{4}{5}
\blight{5}{5}
\bdark{6}{5}
\blight{7}{5}
\bdark{8}{5}
\blight{9}{5}
\bdark{10}{5}
\blight{1}{4}
\bdark{2}{4}
\blight{3}{4}
\bdark{4}{4}
\blight{5}{4}
\bdark{6}{4}
\blight{7}{4}
\bdark{8}{4}
\blight{9}{4}
\bdark{10}{4}
\node[text centered] at (1.5,7.5) {\gr \tiny $0$};
\node[text centered] at (3.5,7.5) {\gr \tiny $3$};
\node[text centered] at (5.5,7.5) {\gr \tiny $0$};
\node[text centered] at (7.5,7.5) {\gr \tiny $0$};
\node[text centered] at (9.5,7.5) {\gr \tiny $1$};
\node[text centered] at (1.5,3.5) {\gr \tiny $1$};
\node[text centered] at (3.5,3.5) {\gr \tiny $2$};
\node[text centered] at (5.5,3.5) {\gr \tiny $0$};
\node[text centered] at (7.5,3.5) {\gr \tiny $1$};
\node[text centered] at (9.5,3.5) {\gr \tiny $0$};
\draw[ultra thick] (0.95,4) -- (0.95,7);
\end{tikzpicture}
\end{align*}
One finds a total of 4 dipole puzzles in $\mathbb{P}_{\nu}(\mu,\lambda)$, giving rise to the sum 
\begin{align*}
\sum_{P \in \mathbb{P}_{\nu}(\mu,\lambda)} (-1)^{L(P)} W(P) 
&
=
\begin{tikzpicture}[scale=0.5,baseline=2.6cm]
\plusb{0}{4}
\plusb{0}{5}
\bdark{0}{6}
\plusg{1}{6}
\flatG{2}{6}
\ming{3}{6}
\bdark{4}{6}
\blight{5}{6}
\bdark{6}{6}
\plusg{7}{6}
\flatG{8}{6}
\ming{9}{6}
\bdark{10}{6}
\flatB{1}{5}
\flatb{2}{5}
\flatB{3}{5}
\flatb{4}{5}
\flatB{5}{5}
\flatb{6}{5}
\flatB{7}{5}
\flatb{8}{5}
\flatB{9}{5}
\minb{10}{5}
\flatB{1}{4}
\flatb{2}{4}
\flatB{3}{4}
\flatb{4}{4}
\flatB{5}{4}
\flatb{6}{4}
\flatB{7}{4}
\flatb{8}{4}
\flatB{9}{4}
\minb{10}{4}
\node[text centered] at (1.5,7.5) {\gr \tiny $0$};
\node[text centered] at (3.5,7.5) {\gr \tiny $3$};
\node[text centered] at (5.5,7.5) {\gr \tiny $0$};
\node[text centered] at (7.5,7.5) {\gr \tiny $0$};
\node[text centered] at (9.5,7.5) {\gr \tiny $1$};
\node[text centered] at (1.5,3.5) {\gr \tiny $1$};
\node[text centered] at (3.5,3.5) {\gr \tiny $2$};
\node[text centered] at (5.5,3.5) {\gr \tiny $0$};
\node[text centered] at (7.5,3.5) {\gr \tiny $1$};
\node[text centered] at (9.5,3.5) {\gr \tiny $0$};
\draw[ultra thick] (0.95,4) -- (0.95,7);
\node[text centered] at (5.5,2.5) {$t^8 (1-t)^2$};
\end{tikzpicture}
\
-
\
\begin{tikzpicture}[scale=0.5,baseline=2.6cm]
\plusb{0}{4}
\plusb{0}{5}
\bdark{0}{6}
\plusg{1}{6}
\flatG{2}{6}
\ming{3}{6}
\bdark{4}{6}
\blight{5}{6}
\plusb{6}{6}
\flatB{7}{6}
\flatb{8}{6}
\flatB{9}{6}
\minb{10}{6}
\flatB{1}{5}
\flatb{2}{5}
\flatB{3}{5}
\flatb{4}{5}
\flatB{5}{5}
\minb{6}{5}
\plusg{7}{5}
\flatG{8}{5}
\ming{9}{5}
\bdark{10}{5}
\flatB{1}{4}
\flatb{2}{4}
\flatB{3}{4}
\flatb{4}{4}
\flatB{5}{4}
\flatb{6}{4}
\flatB{7}{4}
\flatb{8}{4}
\flatB{9}{4}
\minb{10}{4}
\node[text centered] at (1.5,7.5) {\gr \tiny $0$};
\node[text centered] at (3.5,7.5) {\gr \tiny $3$};
\node[text centered] at (5.5,7.5) {\gr \tiny $0$};
\node[text centered] at (7.5,7.5) {\gr \tiny $0$};
\node[text centered] at (9.5,7.5) {\gr \tiny $1$};
\node[text centered] at (1.5,3.5) {\gr \tiny $1$};
\node[text centered] at (3.5,3.5) {\gr \tiny $2$};
\node[text centered] at (5.5,3.5) {\gr \tiny $0$};
\node[text centered] at (7.5,3.5) {\gr \tiny $1$};
\node[text centered] at (9.5,3.5) {\gr \tiny $0$};
\draw[ultra thick] (0.95,4) -- (0.95,7);
\node[text centered] at (5.5,2.5) {$t^8 (1-t)^3$};
\end{tikzpicture}
\\
&-\ 
\begin{tikzpicture}[scale=0.5,baseline=2.6cm]
\plusb{0}{4}
\plusb{0}{5}
\bdark{0}{6}
\plusg{1}{6}
\flatG{2}{6}
\ming{3}{6}
\plusb{4}{6}
\flatB{5}{6}
\flatb{6}{6}
\flatB{7}{6}
\flatb{8}{6}
\flatB{9}{6}
\minb{10}{6}
\flatB{1}{5}
\flatb{2}{5}
\flatB{3}{5}
\minb{4}{5}
\blight{5}{5}
\bdark{6}{5}
\plusg{7}{5}
\flatG{8}{5}
\ming{9}{5}
\bdark{10}{5}
\flatB{1}{4}
\flatb{2}{4}
\flatB{3}{4}
\flatb{4}{4}
\flatB{5}{4}
\flatb{6}{4}
\flatB{7}{4}
\flatb{8}{4}
\flatB{9}{4}
\minb{10}{4}
\node[text centered] at (1.5,7.5) {\gr \tiny $0$};
\node[text centered] at (3.5,7.5) {\gr \tiny $3$};
\node[text centered] at (5.5,7.5) {\gr \tiny $0$};
\node[text centered] at (7.5,7.5) {\gr \tiny $0$};
\node[text centered] at (9.5,7.5) {\gr \tiny $1$};
\node[text centered] at (1.5,3.5) {\gr \tiny $1$};
\node[text centered] at (3.5,3.5) {\gr \tiny $2$};
\node[text centered] at (5.5,3.5) {\gr \tiny $0$};
\node[text centered] at (7.5,3.5) {\gr \tiny $1$};
\node[text centered] at (9.5,3.5) {\gr \tiny $0$};
\draw[ultra thick] (0.95,4) -- (0.95,7);
\node[text centered] at (5.5,2.5) {$t^8 (1-t)^3$};
\end{tikzpicture}
\
+
\
\begin{tikzpicture}[scale=0.5,baseline=2.6cm]
\plusb{0}{4}
\plusb{0}{5}
\bdark{0}{6}
\plusg{1}{6}
\flatG{2}{6}
\ming{3}{6}
\plusb{4}{6}
\flatB{5}{6}
\flatb{6}{6}
\flatB{7}{6}
\flatb{8}{6}
\flatB{9}{6}
\minb{10}{6}
\flatB{1}{5}
\flatb{2}{5}
\flatB{3}{5}
\minb{4}{5}
\blight{5}{5}
\plusb{6}{5}
\flatB{7}{5}
\flatb{8}{5}
\flatB{9}{5}
\minb{10}{5}
\flatB{1}{4}
\flatb{2}{4}
\flatB{3}{4}
\flatb{4}{4}
\flatB{5}{4}
\minb{6}{4}
\plusg{7}{4}
\flatG{8}{4}
\ming{9}{4}
\bdark{10}{4}
\node[text centered] at (1.5,7.5) {\gr \tiny $0$};
\node[text centered] at (3.5,7.5) {\gr \tiny $3$};
\node[text centered] at (5.5,7.5) {\gr \tiny $0$};
\node[text centered] at (7.5,7.5) {\gr \tiny $0$};
\node[text centered] at (9.5,7.5) {\gr \tiny $1$};
\node[text centered] at (1.5,3.5) {\gr \tiny $1$};
\node[text centered] at (3.5,3.5) {\gr \tiny $2$};
\node[text centered] at (5.5,3.5) {\gr \tiny $0$};
\node[text centered] at (7.5,3.5) {\gr \tiny $1$};
\node[text centered] at (9.5,3.5) {\gr \tiny $0$};
\draw[ultra thick] (0.95,4) -- (0.95,7);
\node[text centered] at (5.5,2.5) {$t^8 (1-t)^4$};
\end{tikzpicture}
\end{align*}
Summing the resulting Boltzmann weights, we find that $\sum_{P \in \mathbb{P}_{\nu}(\mu,\lambda)} (-1)^{L(P)} W(P) = t^{10} (1-t)^2$. The Hall polynomial is obtained by normalizing this sum:
\begin{align*}
f^{\b\lambda}_{\b\mu\b\nu}(t)
&=
t^{-(m+1)|\nu|}
\times
\frac{B_{\b\lambda}(t)}{B_{\b\mu}(t) b_{\b\nu}(t)}
\times
\sum_{P \in \mathbb{P}_{\nu}(\mu,\lambda)} (-1)^{L(P)} W(P) 
\\
&=
\frac{t^{-10}(1-t)^2 (1-t^2) (1-t^3)}{(1-t)^3 (1-t^2) (1-t)}
\times
t^{10} (1-t)^2
=
1-t^3. 
\end{align*}

\subsection{The case \texorpdfstring{$\b\lambda = (3,2,1), \b\mu = (2,1,0), \b\nu = (2,1)$}{bar lambda=(3,2,1), bar mu=(2,1,0), bar nu=(2,1)}}

The partitions $\b\lambda,\b\mu,\b\nu$ are the complements of $\lambda = (2,1,0), \mu = (3,2,1), \nu = (2,1)$ by 3. The relevant puzzles hence have the frame
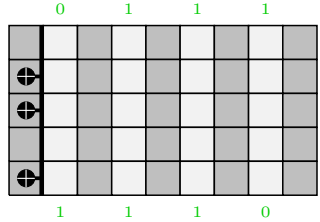
\begin{align*}
\begin{tikzpicture}[scale=0.45]
\plusb{0}{2}
\bdark{0}{3}
\plusb{0}{4}
\plusb{0}{5}
\bdark{0}{6}
\blight{1}{6}
\bdark{2}{6}
\blight{3}{6}
\bdark{4}{6}
\blight{5}{6}
\bdark{6}{6}
\blight{7}{6}
\bdark{8}{6}
\blight{1}{5}
\bdark{2}{5}
\blight{3}{5}
\bdark{4}{5}
\blight{5}{5}
\bdark{6}{5}
\blight{7}{5}
\bdark{8}{5}
\blight{1}{4}
\bdark{2}{4}
\blight{3}{4}
\bdark{4}{4}
\blight{5}{4}
\bdark{6}{4}
\blight{7}{4}
\bdark{8}{4}
\blight{1}{3}
\bdark{2}{3}
\blight{3}{3}
\bdark{4}{3}
\blight{5}{3}
\bdark{6}{3}
\blight{7}{3}
\bdark{8}{3}
\blight{1}{2}
\bdark{2}{2}
\blight{3}{2}
\bdark{4}{2}
\blight{5}{2}
\bdark{6}{2}
\blight{7}{2}
\bdark{8}{2}
\node[text centered] at (1.5,7.5) {\gr \tiny $0$};
\node[text centered] at (3.5,7.5) {\gr \tiny $1$};
\node[text centered] at (5.5,7.5) {\gr \tiny $1$};
\node[text centered] at (7.5,7.5) {\gr \tiny $1$};
\node[text centered] at (1.5,1.5) {\gr \tiny $1$};
\node[text centered] at (3.5,1.5) {\gr \tiny $1$};
\node[text centered] at (5.5,1.5) {\gr \tiny $1$};
\node[text centered] at (7.5,1.5) {\gr \tiny $0$};
\draw[ultra thick] (0.95,2) -- (0.95,7);
\end{tikzpicture}
\end{align*}
In this case there are a total of 28 dipole puzzles in $\mathbb{P}_{\nu}(\mu,\lambda)$, giving rise to the sum
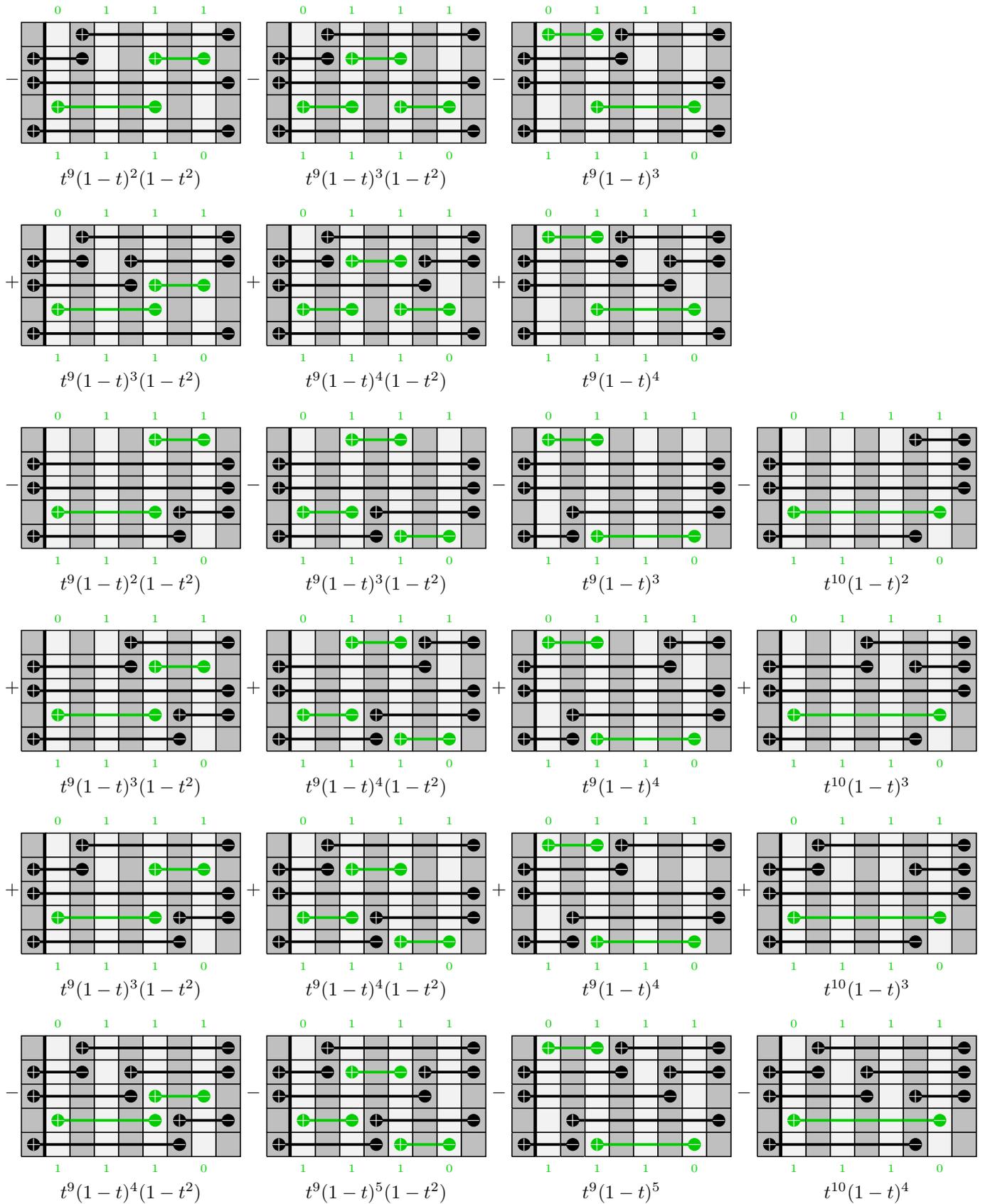
\begin{align*}
&
\begin{tikzpicture}[scale=0.45,baseline=2cm]
\plusb{0}{2}
\bdark{0}{3}
\plusb{0}{4}
\plusb{0}{5}
\bdark{0}{6}
\blight{1}{6}
\bdark{2}{6}
\blight{3}{6}
\bdark{4}{6}
\plusg{5}{6}
\flatG{6}{6}
\ming{7}{6}
\bdark{8}{6}
\flatB{1}{5}
\flatb{2}{5}
\flatB{3}{5}
\flatb{4}{5}
\flatB{5}{5}
\flatb{6}{5}
\flatB{7}{5}
\minb{8}{5}
\flatB{1}{4}
\flatb{2}{4}
\flatB{3}{4}
\flatb{4}{4}
\flatB{5}{4}
\flatb{6}{4}
\flatB{7}{4}
\minb{8}{4}
\plusg{1}{3}
\flatG{2}{3}
\flatg{3}{3}
\flatG{4}{3}
\ming{5}{3}
\bdark{6}{3}
\blight{7}{3}
\bdark{8}{3}
\flatB{1}{2}
\flatb{2}{2}
\flatB{3}{2}
\flatb{4}{2}
\flatB{5}{2}
\flatb{6}{2}
\flatB{7}{2}
\minb{8}{2}
\node[text centered] at (1.5,7.5) {\gr \tiny $0$};
\node[text centered] at (3.5,7.5) {\gr \tiny $1$};
\node[text centered] at (5.5,7.5) {\gr \tiny $1$};
\node[text centered] at (7.5,7.5) {\gr \tiny $1$};
\node[text centered] at (1.5,1.5) {\gr \tiny $1$};
\node[text centered] at (3.5,1.5) {\gr \tiny $1$};
\node[text centered] at (5.5,1.5) {\gr \tiny $1$};
\node[text centered] at (7.5,1.5) {\gr \tiny $0$};
\draw[ultra thick] (0.95,2) -- (0.95,7);
\node[text centered] at (4.5,0.5) {$t^9 (1-t)(1-t^2)$};
\end{tikzpicture}
+
\begin{tikzpicture}[scale=0.45,baseline=2cm]
\plusb{0}{2}
\bdark{0}{3}
\plusb{0}{4}
\plusb{0}{5}
\bdark{0}{6}
\blight{1}{6}
\bdark{2}{6}
\plusg{3}{6}
\flatG{4}{6}
\ming{5}{6}
\bdark{6}{6}
\blight{7}{6}
\bdark{8}{6}
\flatB{1}{5}
\flatb{2}{5}
\flatB{3}{5}
\flatb{4}{5}
\flatB{5}{5}
\flatb{6}{5}
\flatB{7}{5}
\minb{8}{5}
\flatB{1}{4}
\flatb{2}{4}
\flatB{3}{4}
\flatb{4}{4}
\flatB{5}{4}
\flatb{6}{4}
\flatB{7}{4}
\minb{8}{4}
\plusg{1}{3}
\flatG{2}{3}
\ming{3}{3}
\bdark{4}{3}
\plusg{5}{3}
\flatG{6}{3}
\ming{7}{3}
\bdark{8}{3}
\flatB{1}{2}
\flatb{2}{2}
\flatB{3}{2}
\flatb{4}{2}
\flatB{5}{2}
\flatb{6}{2}
\flatB{7}{2}
\minb{8}{2}
\node[text centered] at (1.5,7.5) {\gr \tiny $0$};
\node[text centered] at (3.5,7.5) {\gr \tiny $1$};
\node[text centered] at (5.5,7.5) {\gr \tiny $1$};
\node[text centered] at (7.5,7.5) {\gr \tiny $1$};
\node[text centered] at (1.5,1.5) {\gr \tiny $1$};
\node[text centered] at (3.5,1.5) {\gr \tiny $1$};
\node[text centered] at (5.5,1.5) {\gr \tiny $1$};
\node[text centered] at (7.5,1.5) {\gr \tiny $0$};
\draw[ultra thick] (0.95,2) -- (0.95,7);
\node[text centered] at (4.5,0.5) {$t^9 (1-t)^2 (1-t^2)$};
\end{tikzpicture}
+
\begin{tikzpicture}[scale=0.45,baseline=2cm]
\plusb{0}{2}
\bdark{0}{3}
\plusb{0}{4}
\plusb{0}{5}
\bdark{0}{6}
\plusg{1}{6}
\flatG{2}{6}
\ming{3}{6}
\bdark{4}{6}
\blight{5}{6}
\bdark{6}{6}
\blight{7}{6}
\bdark{8}{6}
\flatB{1}{5}
\flatb{2}{5}
\flatB{3}{5}
\flatb{4}{5}
\flatB{5}{5}
\flatb{6}{5}
\flatB{7}{5}
\minb{8}{5}
\flatB{1}{4}
\flatb{2}{4}
\flatB{3}{4}
\flatb{4}{4}
\flatB{5}{4}
\flatb{6}{4}
\flatB{7}{4}
\minb{8}{4}
\blight{1}{3}
\bdark{2}{3}
\plusg{3}{3}
\flatG{4}{3}
\flatg{5}{3}
\flatG{6}{3}
\ming{7}{3}
\bdark{8}{3}
\flatB{1}{2}
\flatb{2}{2}
\flatB{3}{2}
\flatb{4}{2}
\flatB{5}{2}
\flatb{6}{2}
\flatB{7}{2}
\minb{8}{2}
\node[text centered] at (1.5,7.5) {\gr \tiny $0$};
\node[text centered] at (3.5,7.5) {\gr \tiny $1$};
\node[text centered] at (5.5,7.5) {\gr \tiny $1$};
\node[text centered] at (7.5,7.5) {\gr \tiny $1$};
\node[text centered] at (1.5,1.5) {\gr \tiny $1$};
\node[text centered] at (3.5,1.5) {\gr \tiny $1$};
\node[text centered] at (5.5,1.5) {\gr \tiny $1$};
\node[text centered] at (7.5,1.5) {\gr \tiny $0$};
\draw[ultra thick] (0.95,2) -- (0.95,7);
\node[text centered] at (4.5,0.5) {$t^9 (1-t)^2 $};
\end{tikzpicture}
\\
-
&
\begin{tikzpicture}[scale=0.45,baseline=2cm]
\plusb{0}{2}
\bdark{0}{3}
\plusb{0}{4}
\plusb{0}{5}
\bdark{0}{6}
\blight{1}{6}
\bdark{2}{6}
\blight{3}{6}
\plusb{4}{6}
\flatB{5}{6}
\flatb{6}{6}
\flatB{7}{6}
\minb{8}{6}
\flatB{1}{5}
\flatb{2}{5}
\flatB{3}{5}
\minb{4}{5}
\plusg{5}{5}
\flatG{6}{5}
\ming{7}{5}
\bdark{8}{5}
\flatB{1}{4}
\flatb{2}{4}
\flatB{3}{4}
\flatb{4}{4}
\flatB{5}{4}
\flatb{6}{4}
\flatB{7}{4}
\minb{8}{4}
\plusg{1}{3}
\flatG{2}{3}
\flatg{3}{3}
\flatG{4}{3}
\ming{5}{3}
\bdark{6}{3}
\blight{7}{3}
\bdark{8}{3}
\flatB{1}{2}
\flatb{2}{2}
\flatB{3}{2}
\flatb{4}{2}
\flatB{5}{2}
\flatb{6}{2}
\flatB{7}{2}
\minb{8}{2}
\node[text centered] at (1.5,7.5) {\gr \tiny $0$};
\node[text centered] at (3.5,7.5) {\gr \tiny $1$};
\node[text centered] at (5.5,7.5) {\gr \tiny $1$};
\node[text centered] at (7.5,7.5) {\gr \tiny $1$};
\node[text centered] at (1.5,1.5) {\gr \tiny $1$};
\node[text centered] at (3.5,1.5) {\gr \tiny $1$};
\node[text centered] at (5.5,1.5) {\gr \tiny $1$};
\node[text centered] at (7.5,1.5) {\gr \tiny $0$};
\draw[ultra thick] (0.95,2) -- (0.95,7);
\node[text centered] at (4.5,0.5) {$t^9 (1-t)^2(1-t^2)$};
\end{tikzpicture}
-
\begin{tikzpicture}[scale=0.45,baseline=2cm]
\plusb{0}{2}
\bdark{0}{3}
\plusb{0}{4}
\plusb{0}{5}
\bdark{0}{6}
\blight{1}{6}
\bdark{2}{6}
\plusg{3}{6}
\flatG{4}{6}
\ming{5}{6}
\plusb{6}{6}
\flatB{7}{6}
\minb{8}{6}
\flatB{1}{5}
\flatb{2}{5}
\flatB{3}{5}
\flatb{4}{5}
\flatB{5}{5}
\minb{6}{5}
\blight{7}{5}
\bdark{8}{5}
\flatB{1}{4}
\flatb{2}{4}
\flatB{3}{4}
\flatb{4}{4}
\flatB{5}{4}
\flatb{6}{4}
\flatB{7}{4}
\minb{8}{4}
\plusg{1}{3}
\flatG{2}{3}
\ming{3}{3}
\bdark{4}{3}
\plusg{5}{3}
\flatG{6}{3}
\ming{7}{3}
\bdark{8}{3}
\flatB{1}{2}
\flatb{2}{2}
\flatB{3}{2}
\flatb{4}{2}
\flatB{5}{2}
\flatb{6}{2}
\flatB{7}{2}
\minb{8}{2}
\node[text centered] at (1.5,7.5) {\gr \tiny $0$};
\node[text centered] at (3.5,7.5) {\gr \tiny $1$};
\node[text centered] at (5.5,7.5) {\gr \tiny $1$};
\node[text centered] at (7.5,7.5) {\gr \tiny $1$};
\node[text centered] at (1.5,1.5) {\gr \tiny $1$};
\node[text centered] at (3.5,1.5) {\gr \tiny $1$};
\node[text centered] at (5.5,1.5) {\gr \tiny $1$};
\node[text centered] at (7.5,1.5) {\gr \tiny $0$};
\draw[ultra thick] (0.95,2) -- (0.95,7);
\node[text centered] at (4.5,0.5) {$t^9 (1-t)^3 (1-t^2)$};
\end{tikzpicture}
-
\begin{tikzpicture}[scale=0.45,baseline=2cm]
\plusb{0}{2}
\bdark{0}{3}
\plusb{0}{4}
\plusb{0}{5}
\bdark{0}{6}
\plusg{1}{6}
\flatG{2}{6}
\ming{3}{6}
\bdark{4}{6}
\blight{5}{6}
\plusb{6}{6}
\flatB{7}{6}
\minb{8}{6}
\flatB{1}{5}
\flatb{2}{5}
\flatB{3}{5}
\flatb{4}{5}
\flatB{5}{5}
\minb{6}{5}
\blight{7}{5}
\bdark{8}{5}
\flatB{1}{4}
\flatb{2}{4}
\flatB{3}{4}
\flatb{4}{4}
\flatB{5}{4}
\flatb{6}{4}
\flatB{7}{4}
\minb{8}{4}
\blight{1}{3}
\bdark{2}{3}
\plusg{3}{3}
\flatG{4}{3}
\flatg{5}{3}
\flatG{6}{3}
\ming{7}{3}
\bdark{8}{3}
\flatB{1}{2}
\flatb{2}{2}
\flatB{3}{2}
\flatb{4}{2}
\flatB{5}{2}
\flatb{6}{2}
\flatB{7}{2}
\minb{8}{2}
\node[text centered] at (1.5,7.5) {\gr \tiny $0$};
\node[text centered] at (3.5,7.5) {\gr \tiny $1$};
\node[text centered] at (5.5,7.5) {\gr \tiny $1$};
\node[text centered] at (7.5,7.5) {\gr \tiny $1$};
\node[text centered] at (1.5,1.5) {\gr \tiny $1$};
\node[text centered] at (3.5,1.5) {\gr \tiny $1$};
\node[text centered] at (5.5,1.5) {\gr \tiny $1$};
\node[text centered] at (7.5,1.5) {\gr \tiny $0$};
\draw[ultra thick] (0.95,2) -- (0.95,7);
\node[text centered] at (4.5,0.5) {$t^9 (1-t)^3 $};
\end{tikzpicture}
\\
-
&
\begin{tikzpicture}[scale=0.45,baseline=2cm]
\plusb{0}{2}
\bdark{0}{3}
\plusb{0}{4}
\plusb{0}{5}
\bdark{0}{6}
\blight{1}{6}
\plusb{2}{6}
\flatB{3}{6}
\flatb{4}{6}
\flatB{5}{6}
\flatb{6}{6}
\flatB{7}{6}
\minb{8}{6}
\flatB{1}{5}
\minb{2}{5}
\blight{3}{5}
\bdark{4}{5}
\plusg{5}{5}
\flatG{6}{5}
\ming{7}{5}
\bdark{8}{5}
\flatB{1}{4}
\flatb{2}{4}
\flatB{3}{4}
\flatb{4}{4}
\flatB{5}{4}
\flatb{6}{4}
\flatB{7}{4}
\minb{8}{4}
\plusg{1}{3}
\flatG{2}{3}
\flatg{3}{3}
\flatG{4}{3}
\ming{5}{3}
\bdark{6}{3}
\blight{7}{3}
\bdark{8}{3}
\flatB{1}{2}
\flatb{2}{2}
\flatB{3}{2}
\flatb{4}{2}
\flatB{5}{2}
\flatb{6}{2}
\flatB{7}{2}
\minb{8}{2}
\node[text centered] at (1.5,7.5) {\gr \tiny $0$};
\node[text centered] at (3.5,7.5) {\gr \tiny $1$};
\node[text centered] at (5.5,7.5) {\gr \tiny $1$};
\node[text centered] at (7.5,7.5) {\gr \tiny $1$};
\node[text centered] at (1.5,1.5) {\gr \tiny $1$};
\node[text centered] at (3.5,1.5) {\gr \tiny $1$};
\node[text centered] at (5.5,1.5) {\gr \tiny $1$};
\node[text centered] at (7.5,1.5) {\gr \tiny $0$};
\draw[ultra thick] (0.95,2) -- (0.95,7);
\node[text centered] at (4.5,0.5) {$t^9 (1-t)^2(1-t^2)$};
\end{tikzpicture}
-
\begin{tikzpicture}[scale=0.45,baseline=2cm]
\plusb{0}{2}
\bdark{0}{3}
\plusb{0}{4}
\plusb{0}{5}
\bdark{0}{6}
\blight{1}{6}
\plusb{2}{6}
\flatB{3}{6}
\flatb{4}{6}
\flatB{5}{6}
\flatb{6}{6}
\flatB{7}{6}
\minb{8}{6}
\flatB{1}{5}
\minb{2}{5}
\plusg{3}{5}
\flatG{4}{5}
\ming{5}{5}
\bdark{6}{5}
\blight{7}{5}
\bdark{8}{5}
\flatB{1}{4}
\flatb{2}{4}
\flatB{3}{4}
\flatb{4}{4}
\flatB{5}{4}
\flatb{6}{4}
\flatB{7}{4}
\minb{8}{4}
\plusg{1}{3}
\flatG{2}{3}
\ming{3}{3}
\bdark{4}{3}
\plusg{5}{3}
\flatG{6}{3}
\ming{7}{3}
\bdark{8}{3}
\flatB{1}{2}
\flatb{2}{2}
\flatB{3}{2}
\flatb{4}{2}
\flatB{5}{2}
\flatb{6}{2}
\flatB{7}{2}
\minb{8}{2}
\node[text centered] at (1.5,7.5) {\gr \tiny $0$};
\node[text centered] at (3.5,7.5) {\gr \tiny $1$};
\node[text centered] at (5.5,7.5) {\gr \tiny $1$};
\node[text centered] at (7.5,7.5) {\gr \tiny $1$};
\node[text centered] at (1.5,1.5) {\gr \tiny $1$};
\node[text centered] at (3.5,1.5) {\gr \tiny $1$};
\node[text centered] at (5.5,1.5) {\gr \tiny $1$};
\node[text centered] at (7.5,1.5) {\gr \tiny $0$};
\draw[ultra thick] (0.95,2) -- (0.95,7);
\node[text centered] at (4.5,0.5) {$t^9 (1-t)^3 (1-t^2)$};
\end{tikzpicture}
-
\begin{tikzpicture}[scale=0.45,baseline=2cm]
\plusb{0}{2}
\bdark{0}{3}
\plusb{0}{4}
\plusb{0}{5}
\bdark{0}{6}
\plusg{1}{6}
\flatG{2}{6}
\ming{3}{6}
\plusb{4}{6}
\flatB{5}{6}
\flatb{6}{6}
\flatB{7}{6}
\minb{8}{6}
\flatB{1}{5}
\flatb{2}{5}
\flatB{3}{5}
\minb{4}{5}
\blight{5}{5}
\bdark{6}{5}
\blight{7}{5}
\bdark{8}{5}
\flatB{1}{4}
\flatb{2}{4}
\flatB{3}{4}
\flatb{4}{4}
\flatB{5}{4}
\flatb{6}{4}
\flatB{7}{4}
\minb{8}{4}
\blight{1}{3}
\bdark{2}{3}
\plusg{3}{3}
\flatG{4}{3}
\flatg{5}{3}
\flatG{6}{3}
\ming{7}{3}
\bdark{8}{3}
\flatB{1}{2}
\flatb{2}{2}
\flatB{3}{2}
\flatb{4}{2}
\flatB{5}{2}
\flatb{6}{2}
\flatB{7}{2}
\minb{8}{2}
\node[text centered] at (1.5,7.5) {\gr \tiny $0$};
\node[text centered] at (3.5,7.5) {\gr \tiny $1$};
\node[text centered] at (5.5,7.5) {\gr \tiny $1$};
\node[text centered] at (7.5,7.5) {\gr \tiny $1$};
\node[text centered] at (1.5,1.5) {\gr \tiny $1$};
\node[text centered] at (3.5,1.5) {\gr \tiny $1$};
\node[text centered] at (5.5,1.5) {\gr \tiny $1$};
\node[text centered] at (7.5,1.5) {\gr \tiny $0$};
\draw[ultra thick] (0.95,2) -- (0.95,7);
\node[text centered] at (4.5,0.5) {$t^9 (1-t)^3 $};
\end{tikzpicture}
\\
+
&
\begin{tikzpicture}[scale=0.45,baseline=2cm]
\plusb{0}{2}
\bdark{0}{3}
\plusb{0}{4}
\plusb{0}{5}
\bdark{0}{6}
\blight{1}{6}
\plusb{2}{6}
\flatB{3}{6}
\flatb{4}{6}
\flatB{5}{6}
\flatb{6}{6}
\flatB{7}{6}
\minb{8}{6}
\flatB{1}{5}
\minb{2}{5}
\blight{3}{5}
\plusb{4}{5}
\flatB{5}{5}
\flatb{6}{5}
\flatB{7}{5}
\minb{8}{5}
\flatB{1}{4}
\flatb{2}{4}
\flatB{3}{4}
\minb{4}{4}
\plusg{5}{4}
\flatG{6}{4}
\ming{7}{4}
\bdark{8}{4}
\plusg{1}{3}
\flatG{2}{3}
\flatg{3}{3}
\flatG{4}{3}
\ming{5}{3}
\bdark{6}{3}
\blight{7}{3}
\bdark{8}{3}
\flatB{1}{2}
\flatb{2}{2}
\flatB{3}{2}
\flatb{4}{2}
\flatB{5}{2}
\flatb{6}{2}
\flatB{7}{2}
\minb{8}{2}
\node[text centered] at (1.5,7.5) {\gr \tiny $0$};
\node[text centered] at (3.5,7.5) {\gr \tiny $1$};
\node[text centered] at (5.5,7.5) {\gr \tiny $1$};
\node[text centered] at (7.5,7.5) {\gr \tiny $1$};
\node[text centered] at (1.5,1.5) {\gr \tiny $1$};
\node[text centered] at (3.5,1.5) {\gr \tiny $1$};
\node[text centered] at (5.5,1.5) {\gr \tiny $1$};
\node[text centered] at (7.5,1.5) {\gr \tiny $0$};
\draw[ultra thick] (0.95,2) -- (0.95,7);
\node[text centered] at (4.5,0.5) {$t^9 (1-t)^3(1-t^2)$};
\end{tikzpicture}
+
\begin{tikzpicture}[scale=0.45,baseline=2cm]
\plusb{0}{2}
\bdark{0}{3}
\plusb{0}{4}
\plusb{0}{5}
\bdark{0}{6}
\blight{1}{6}
\plusb{2}{6}
\flatB{3}{6}
\flatb{4}{6}
\flatB{5}{6}
\flatb{6}{6}
\flatB{7}{6}
\minb{8}{6}
\flatB{1}{5}
\minb{2}{5}
\plusg{3}{5}
\flatG{4}{5}
\ming{5}{5}
\plusb{6}{5}
\flatB{7}{5}
\minb{8}{5}
\flatB{1}{4}
\flatb{2}{4}
\flatB{3}{4}
\flatb{4}{4}
\flatB{5}{4}
\minb{6}{4}
\blight{7}{4}
\bdark{8}{4}
\plusg{1}{3}
\flatG{2}{3}
\ming{3}{3}
\bdark{4}{3}
\plusg{5}{3}
\flatG{6}{3}
\ming{7}{3}
\bdark{8}{3}
\flatB{1}{2}
\flatb{2}{2}
\flatB{3}{2}
\flatb{4}{2}
\flatB{5}{2}
\flatb{6}{2}
\flatB{7}{2}
\minb{8}{2}
\node[text centered] at (1.5,7.5) {\gr \tiny $0$};
\node[text centered] at (3.5,7.5) {\gr \tiny $1$};
\node[text centered] at (5.5,7.5) {\gr \tiny $1$};
\node[text centered] at (7.5,7.5) {\gr \tiny $1$};
\node[text centered] at (1.5,1.5) {\gr \tiny $1$};
\node[text centered] at (3.5,1.5) {\gr \tiny $1$};
\node[text centered] at (5.5,1.5) {\gr \tiny $1$};
\node[text centered] at (7.5,1.5) {\gr \tiny $0$};
\draw[ultra thick] (0.95,2) -- (0.95,7);
\node[text centered] at (4.5,0.5) {$t^9 (1-t)^4 (1-t^2)$};
\end{tikzpicture}
+
\begin{tikzpicture}[scale=0.45,baseline=2cm]
\plusb{0}{2}
\bdark{0}{3}
\plusb{0}{4}
\plusb{0}{5}
\bdark{0}{6}
\plusg{1}{6}
\flatG{2}{6}
\ming{3}{6}
\plusb{4}{6}
\flatB{5}{6}
\flatb{6}{6}
\flatB{7}{6}
\minb{8}{6}
\flatB{1}{5}
\flatb{2}{5}
\flatB{3}{5}
\minb{4}{5}
\blight{5}{5}
\plusb{6}{5}
\flatB{7}{5}
\minb{8}{5}
\flatB{1}{4}
\flatb{2}{4}
\flatB{3}{4}
\flatb{4}{4}
\flatB{5}{4}
\minb{6}{4}
\blight{7}{4}
\bdark{8}{4}
\blight{1}{3}
\bdark{2}{3}
\plusg{3}{3}
\flatG{4}{3}
\flatg{5}{3}
\flatG{6}{3}
\ming{7}{3}
\bdark{8}{3}
\flatB{1}{2}
\flatb{2}{2}
\flatB{3}{2}
\flatb{4}{2}
\flatB{5}{2}
\flatb{6}{2}
\flatB{7}{2}
\minb{8}{2}
\node[text centered] at (1.5,7.5) {\gr \tiny $0$};
\node[text centered] at (3.5,7.5) {\gr \tiny $1$};
\node[text centered] at (5.5,7.5) {\gr \tiny $1$};
\node[text centered] at (7.5,7.5) {\gr \tiny $1$};
\node[text centered] at (1.5,1.5) {\gr \tiny $1$};
\node[text centered] at (3.5,1.5) {\gr \tiny $1$};
\node[text centered] at (5.5,1.5) {\gr \tiny $1$};
\node[text centered] at (7.5,1.5) {\gr \tiny $0$};
\draw[ultra thick] (0.95,2) -- (0.95,7);
\node[text centered] at (4.5,0.5) {$t^9 (1-t)^4 $};
\end{tikzpicture}
\\
-
&
\begin{tikzpicture}[scale=0.45,baseline=2cm]
\plusb{0}{2}
\bdark{0}{3}
\plusb{0}{4}
\plusb{0}{5}
\bdark{0}{6}
\blight{1}{6}
\bdark{2}{6}
\blight{3}{6}
\bdark{4}{6}
\plusg{5}{6}
\flatG{6}{6}
\ming{7}{6}
\bdark{8}{6}
\flatB{1}{5}
\flatb{2}{5}
\flatB{3}{5}
\flatb{4}{5}
\flatB{5}{5}
\flatb{6}{5}
\flatB{7}{5}
\minb{8}{5}
\flatB{1}{4}
\flatb{2}{4}
\flatB{3}{4}
\flatb{4}{4}
\flatB{5}{4}
\flatb{6}{4}
\flatB{7}{4}
\minb{8}{4}
\plusg{1}{3}
\flatG{2}{3}
\flatg{3}{3}
\flatG{4}{3}
\ming{5}{3}
\plusb{6}{3}
\flatB{7}{3}
\minb{8}{3}
\flatB{1}{2}
\flatb{2}{2}
\flatB{3}{2}
\flatb{4}{2}
\flatB{5}{2}
\minb{6}{2}
\blight{7}{2}
\bdark{8}{2}
\node[text centered] at (1.5,7.5) {\gr \tiny $0$};
\node[text centered] at (3.5,7.5) {\gr \tiny $1$};
\node[text centered] at (5.5,7.5) {\gr \tiny $1$};
\node[text centered] at (7.5,7.5) {\gr \tiny $1$};
\node[text centered] at (1.5,1.5) {\gr \tiny $1$};
\node[text centered] at (3.5,1.5) {\gr \tiny $1$};
\node[text centered] at (5.5,1.5) {\gr \tiny $1$};
\node[text centered] at (7.5,1.5) {\gr \tiny $0$};
\draw[ultra thick] (0.95,2) -- (0.95,7);
\node[text centered] at (4.5,0.5) {$t^9 (1-t)^2(1-t^2)$};
\end{tikzpicture}
-
\begin{tikzpicture}[scale=0.45,baseline=2cm]
\plusb{0}{2}
\bdark{0}{3}
\plusb{0}{4}
\plusb{0}{5}
\bdark{0}{6}
\blight{1}{6}
\bdark{2}{6}
\plusg{3}{6}
\flatG{4}{6}
\ming{5}{6}
\bdark{6}{6}
\blight{7}{6}
\bdark{8}{6}
\flatB{1}{5}
\flatb{2}{5}
\flatB{3}{5}
\flatb{4}{5}
\flatB{5}{5}
\flatb{6}{5}
\flatB{7}{5}
\minb{8}{5}
\flatB{1}{4}
\flatb{2}{4}
\flatB{3}{4}
\flatb{4}{4}
\flatB{5}{4}
\flatb{6}{4}
\flatB{7}{4}
\minb{8}{4}
\plusg{1}{3}
\flatG{2}{3}
\ming{3}{3}
\plusb{4}{3}
\flatB{5}{3}
\flatb{6}{3}
\flatB{7}{3}
\minb{8}{3}
\flatB{1}{2}
\flatb{2}{2}
\flatB{3}{2}
\minb{4}{2}
\plusg{5}{2}
\flatG{6}{2}
\ming{7}{2}
\bdark{8}{2}
\node[text centered] at (1.5,7.5) {\gr \tiny $0$};
\node[text centered] at (3.5,7.5) {\gr \tiny $1$};
\node[text centered] at (5.5,7.5) {\gr \tiny $1$};
\node[text centered] at (7.5,7.5) {\gr \tiny $1$};
\node[text centered] at (1.5,1.5) {\gr \tiny $1$};
\node[text centered] at (3.5,1.5) {\gr \tiny $1$};
\node[text centered] at (5.5,1.5) {\gr \tiny $1$};
\node[text centered] at (7.5,1.5) {\gr \tiny $0$};
\draw[ultra thick] (0.95,2) -- (0.95,7);
\node[text centered] at (4.5,0.5) {$t^9 (1-t)^3 (1-t^2)$};
\end{tikzpicture}
-
\begin{tikzpicture}[scale=0.45,baseline=2cm]
\plusb{0}{2}
\bdark{0}{3}
\plusb{0}{4}
\plusb{0}{5}
\bdark{0}{6}
\plusg{1}{6}
\flatG{2}{6}
\ming{3}{6}
\bdark{4}{6}
\blight{5}{6}
\bdark{6}{6}
\blight{7}{6}
\bdark{8}{6}
\flatB{1}{5}
\flatb{2}{5}
\flatB{3}{5}
\flatb{4}{5}
\flatB{5}{5}
\flatb{6}{5}
\flatB{7}{5}
\minb{8}{5}
\flatB{1}{4}
\flatb{2}{4}
\flatB{3}{4}
\flatb{4}{4}
\flatB{5}{4}
\flatb{6}{4}
\flatB{7}{4}
\minb{8}{4}
\blight{1}{3}
\plusb{2}{3}
\flatB{3}{3}
\flatb{4}{3}
\flatB{5}{3}
\flatb{6}{3}
\flatB{7}{3}
\minb{8}{3}
\flatB{1}{2}
\minb{2}{2}
\plusg{3}{2}
\flatG{4}{2}
\flatg{5}{2}
\flatG{6}{2}
\ming{7}{2}
\bdark{8}{2}
\node[text centered] at (1.5,7.5) {\gr \tiny $0$};
\node[text centered] at (3.5,7.5) {\gr \tiny $1$};
\node[text centered] at (5.5,7.5) {\gr \tiny $1$};
\node[text centered] at (7.5,7.5) {\gr \tiny $1$};
\node[text centered] at (1.5,1.5) {\gr \tiny $1$};
\node[text centered] at (3.5,1.5) {\gr \tiny $1$};
\node[text centered] at (5.5,1.5) {\gr \tiny $1$};
\node[text centered] at (7.5,1.5) {\gr \tiny $0$};
\draw[ultra thick] (0.95,2) -- (0.95,7);
\node[text centered] at (4.5,0.5) {$t^9 (1-t)^3 $};
\end{tikzpicture}
-
\begin{tikzpicture}[scale=0.45,baseline=2cm]
\plusb{0}{2}
\bdark{0}{3}
\plusb{0}{4}
\plusb{0}{5}
\bdark{0}{6}
\blight{1}{6}
\bdark{2}{6}
\blight{3}{6}
\bdark{4}{6}
\blight{5}{6}
\plusb{6}{6}
\flatB{7}{6}
\minb{8}{6}
\flatB{1}{5}
\flatb{2}{5}
\flatB{3}{5}
\flatb{4}{5}
\flatB{5}{5}
\flatb{6}{5}
\flatB{7}{5}
\minb{8}{5}
\flatB{1}{4}
\flatb{2}{4}
\flatB{3}{4}
\flatb{4}{4}
\flatB{5}{4}
\flatb{6}{4}
\flatB{7}{4}
\minb{8}{4}
\plusg{1}{3}
\flatG{2}{3}
\flatg{3}{3}
\flatG{4}{3}
\flatg{5}{3}
\flatG{6}{3}
\ming{7}{3}
\bdark{8}{3}
\flatB{1}{2}
\flatb{2}{2}
\flatB{3}{2}
\flatb{4}{2}
\flatB{5}{2}
\minb{6}{2}
\blight{7}{2}
\bdark{8}{2}
\node[text centered] at (1.5,7.5) {\gr \tiny $0$};
\node[text centered] at (3.5,7.5) {\gr \tiny $1$};
\node[text centered] at (5.5,7.5) {\gr \tiny $1$};
\node[text centered] at (7.5,7.5) {\gr \tiny $1$};
\node[text centered] at (1.5,1.5) {\gr \tiny $1$};
\node[text centered] at (3.5,1.5) {\gr \tiny $1$};
\node[text centered] at (5.5,1.5) {\gr \tiny $1$};
\node[text centered] at (7.5,1.5) {\gr \tiny $0$};
\draw[ultra thick] (0.95,2) -- (0.95,7);
\node[text centered] at (4.5,0.5) {$t^{10} (1-t)^2 $};
\end{tikzpicture}
\\
+
&
\begin{tikzpicture}[scale=0.45,baseline=2cm]
\plusb{0}{2}
\bdark{0}{3}
\plusb{0}{4}
\plusb{0}{5}
\bdark{0}{6}
\blight{1}{6}
\bdark{2}{6}
\blight{3}{6}
\plusb{4}{6}
\flatB{5}{6}
\flatb{6}{6}
\flatB{7}{6}
\minb{8}{6}
\flatB{1}{5}
\flatb{2}{5}
\flatB{3}{5}
\minb{4}{5}
\plusg{5}{5}
\flatG{6}{5}
\ming{7}{5}
\bdark{8}{5}
\flatB{1}{4}
\flatb{2}{4}
\flatB{3}{4}
\flatb{4}{4}
\flatB{5}{4}
\flatb{6}{4}
\flatB{7}{4}
\minb{8}{4}
\plusg{1}{3}
\flatG{2}{3}
\flatg{3}{3}
\flatG{4}{3}
\ming{5}{3}
\plusb{6}{3}
\flatB{7}{3}
\minb{8}{3}
\flatB{1}{2}
\flatb{2}{2}
\flatB{3}{2}
\flatb{4}{2}
\flatB{5}{2}
\minb{6}{2}
\blight{7}{2}
\bdark{8}{2}
\node[text centered] at (1.5,7.5) {\gr \tiny $0$};
\node[text centered] at (3.5,7.5) {\gr \tiny $1$};
\node[text centered] at (5.5,7.5) {\gr \tiny $1$};
\node[text centered] at (7.5,7.5) {\gr \tiny $1$};
\node[text centered] at (1.5,1.5) {\gr \tiny $1$};
\node[text centered] at (3.5,1.5) {\gr \tiny $1$};
\node[text centered] at (5.5,1.5) {\gr \tiny $1$};
\node[text centered] at (7.5,1.5) {\gr \tiny $0$};
\draw[ultra thick] (0.95,2) -- (0.95,7);
\node[text centered] at (4.5,0.5) {$t^9 (1-t)^3(1-t^2)$};
\end{tikzpicture}
+
\begin{tikzpicture}[scale=0.45,baseline=2cm]
\plusb{0}{2}
\bdark{0}{3}
\plusb{0}{4}
\plusb{0}{5}
\bdark{0}{6}
\blight{1}{6}
\bdark{2}{6}
\plusg{3}{6}
\flatG{4}{6}
\ming{5}{6}
\plusb{6}{6}
\flatB{7}{6}
\minb{8}{6}
\flatB{1}{5}
\flatb{2}{5}
\flatB{3}{5}
\flatb{4}{5}
\flatB{5}{5}
\minb{6}{5}
\blight{7}{5}
\bdark{8}{5}
\flatB{1}{4}
\flatb{2}{4}
\flatB{3}{4}
\flatb{4}{4}
\flatB{5}{4}
\flatb{6}{4}
\flatB{7}{4}
\minb{8}{4}
\plusg{1}{3}
\flatG{2}{3}
\ming{3}{3}
\plusb{4}{3}
\flatB{5}{3}
\flatb{6}{3}
\flatB{7}{3}
\minb{8}{3}
\flatB{1}{2}
\flatb{2}{2}
\flatB{3}{2}
\minb{4}{2}
\plusg{5}{2}
\flatG{6}{2}
\ming{7}{2}
\bdark{8}{2}
\node[text centered] at (1.5,7.5) {\gr \tiny $0$};
\node[text centered] at (3.5,7.5) {\gr \tiny $1$};
\node[text centered] at (5.5,7.5) {\gr \tiny $1$};
\node[text centered] at (7.5,7.5) {\gr \tiny $1$};
\node[text centered] at (1.5,1.5) {\gr \tiny $1$};
\node[text centered] at (3.5,1.5) {\gr \tiny $1$};
\node[text centered] at (5.5,1.5) {\gr \tiny $1$};
\node[text centered] at (7.5,1.5) {\gr \tiny $0$};
\draw[ultra thick] (0.95,2) -- (0.95,7);
\node[text centered] at (4.5,0.5) {$t^9 (1-t)^4 (1-t^2)$};
\end{tikzpicture}
+
\begin{tikzpicture}[scale=0.45,baseline=2cm]
\plusb{0}{2}
\bdark{0}{3}
\plusb{0}{4}
\plusb{0}{5}
\bdark{0}{6}
\plusg{1}{6}
\flatG{2}{6}
\ming{3}{6}
\bdark{4}{6}
\blight{5}{6}
\plusb{6}{6}
\flatB{7}{6}
\minb{8}{6}
\flatB{1}{5}
\flatb{2}{5}
\flatB{3}{5}
\flatb{4}{5}
\flatB{5}{5}
\minb{6}{5}
\blight{7}{5}
\bdark{8}{5}
\flatB{1}{4}
\flatb{2}{4}
\flatB{3}{4}
\flatb{4}{4}
\flatB{5}{4}
\flatb{6}{4}
\flatB{7}{4}
\minb{8}{4}
\blight{1}{3}
\plusb{2}{3}
\flatB{3}{3}
\flatb{4}{3}
\flatB{5}{3}
\flatb{6}{3}
\flatB{7}{3}
\minb{8}{3}
\flatB{1}{2}
\minb{2}{2}
\plusg{3}{2}
\flatG{4}{2}
\flatg{5}{2}
\flatG{6}{2}
\ming{7}{2}
\bdark{8}{2}
\node[text centered] at (1.5,7.5) {\gr \tiny $0$};
\node[text centered] at (3.5,7.5) {\gr \tiny $1$};
\node[text centered] at (5.5,7.5) {\gr \tiny $1$};
\node[text centered] at (7.5,7.5) {\gr \tiny $1$};
\node[text centered] at (1.5,1.5) {\gr \tiny $1$};
\node[text centered] at (3.5,1.5) {\gr \tiny $1$};
\node[text centered] at (5.5,1.5) {\gr \tiny $1$};
\node[text centered] at (7.5,1.5) {\gr \tiny $0$};
\draw[ultra thick] (0.95,2) -- (0.95,7);
\node[text centered] at (4.5,0.5) {$t^9 (1-t)^4 $};
\end{tikzpicture}
+
\begin{tikzpicture}[scale=0.45,baseline=2cm]
\plusb{0}{2}
\bdark{0}{3}
\plusb{0}{4}
\plusb{0}{5}
\bdark{0}{6}
\blight{1}{6}
\bdark{2}{6}
\blight{3}{6}
\plusb{4}{6}
\flatB{5}{6}
\flatb{6}{6}
\flatB{7}{6}
\minb{8}{6}
\flatB{1}{5}
\flatb{2}{5}
\flatB{3}{5}
\minb{4}{5}
\blight{5}{5}
\plusb{6}{5}
\flatB{7}{5}
\minb{8}{5}
\flatB{1}{4}
\flatb{2}{4}
\flatB{3}{4}
\flatb{4}{4}
\flatB{5}{4}
\flatb{6}{4}
\flatB{7}{4}
\minb{8}{4}
\plusg{1}{3}
\flatG{2}{3}
\flatg{3}{3}
\flatG{4}{3}
\flatg{5}{3}
\flatG{6}{3}
\ming{7}{3}
\bdark{8}{3}
\flatB{1}{2}
\flatb{2}{2}
\flatB{3}{2}
\flatb{4}{2}
\flatB{5}{2}
\minb{6}{2}
\blight{7}{2}
\bdark{8}{2}
\node[text centered] at (1.5,7.5) {\gr \tiny $0$};
\node[text centered] at (3.5,7.5) {\gr \tiny $1$};
\node[text centered] at (5.5,7.5) {\gr \tiny $1$};
\node[text centered] at (7.5,7.5) {\gr \tiny $1$};
\node[text centered] at (1.5,1.5) {\gr \tiny $1$};
\node[text centered] at (3.5,1.5) {\gr \tiny $1$};
\node[text centered] at (5.5,1.5) {\gr \tiny $1$};
\node[text centered] at (7.5,1.5) {\gr \tiny $0$};
\draw[ultra thick] (0.95,2) -- (0.95,7);
\node[text centered] at (4.5,0.5) {$t^{10} (1-t)^3 $};
\end{tikzpicture}
\\
+
&
\begin{tikzpicture}[scale=0.45,baseline=2cm]
\plusb{0}{2}
\bdark{0}{3}
\plusb{0}{4}
\plusb{0}{5}
\bdark{0}{6}
\blight{1}{6}
\plusb{2}{6}
\flatB{3}{6}
\flatb{4}{6}
\flatB{5}{6}
\flatb{6}{6}
\flatB{7}{6}
\minb{8}{6}
\flatB{1}{5}
\minb{2}{5}
\blight{3}{5}
\bdark{4}{5}
\plusg{5}{5}
\flatG{6}{5}
\ming{7}{5}
\bdark{8}{5}
\flatB{1}{4}
\flatb{2}{4}
\flatB{3}{4}
\flatb{4}{4}
\flatB{5}{4}
\flatb{6}{4}
\flatB{7}{4}
\minb{8}{4}
\plusg{1}{3}
\flatG{2}{3}
\flatg{3}{3}
\flatG{4}{3}
\ming{5}{3}
\plusb{6}{3}
\flatB{7}{3}
\minb{8}{3}
\flatB{1}{2}
\flatb{2}{2}
\flatB{3}{2}
\flatb{4}{2}
\flatB{5}{2}
\minb{6}{2}
\blight{7}{2}
\bdark{8}{2}
\node[text centered] at (1.5,7.5) {\gr \tiny $0$};
\node[text centered] at (3.5,7.5) {\gr \tiny $1$};
\node[text centered] at (5.5,7.5) {\gr \tiny $1$};
\node[text centered] at (7.5,7.5) {\gr \tiny $1$};
\node[text centered] at (1.5,1.5) {\gr \tiny $1$};
\node[text centered] at (3.5,1.5) {\gr \tiny $1$};
\node[text centered] at (5.5,1.5) {\gr \tiny $1$};
\node[text centered] at (7.5,1.5) {\gr \tiny $0$};
\draw[ultra thick] (0.95,2) -- (0.95,7);
\node[text centered] at (4.5,0.5) {$t^9 (1-t)^3 (1-t^2)$};
\end{tikzpicture}
+
\begin{tikzpicture}[scale=0.45,baseline=2cm]
\plusb{0}{2}
\bdark{0}{3}
\plusb{0}{4}
\plusb{0}{5}
\bdark{0}{6}
\blight{1}{6}
\plusb{2}{6}
\flatB{3}{6}
\flatb{4}{6}
\flatB{5}{6}
\flatb{6}{6}
\flatB{7}{6}
\minb{8}{6}
\flatB{1}{5}
\minb{2}{5}
\plusg{3}{5}
\flatG{4}{5}
\ming{5}{5}
\bdark{6}{5}
\blight{7}{5}
\bdark{8}{5}
\flatB{1}{4}
\flatb{2}{4}
\flatB{3}{4}
\flatb{4}{4}
\flatB{5}{4}
\flatb{6}{4}
\flatB{7}{4}
\minb{8}{4}
\plusg{1}{3}
\flatG{2}{3}
\ming{3}{3}
\plusb{4}{3}
\flatB{5}{3}
\flatb{6}{3}
\flatB{7}{3}
\minb{8}{3}
\flatB{1}{2}
\flatb{2}{2}
\flatB{3}{2}
\minb{4}{2}
\plusg{5}{2}
\flatG{6}{2}
\ming{7}{2}
\bdark{8}{2}
\node[text centered] at (1.5,7.5) {\gr \tiny $0$};
\node[text centered] at (3.5,7.5) {\gr \tiny $1$};
\node[text centered] at (5.5,7.5) {\gr \tiny $1$};
\node[text centered] at (7.5,7.5) {\gr \tiny $1$};
\node[text centered] at (1.5,1.5) {\gr \tiny $1$};
\node[text centered] at (3.5,1.5) {\gr \tiny $1$};
\node[text centered] at (5.5,1.5) {\gr \tiny $1$};
\node[text centered] at (7.5,1.5) {\gr \tiny $0$};
\draw[ultra thick] (0.95,2) -- (0.95,7);
\node[text centered] at (4.5,0.5) {$t^9 (1-t)^4 (1-t^2)$};
\end{tikzpicture}
+
\begin{tikzpicture}[scale=0.45,baseline=2cm]
\plusb{0}{2}
\bdark{0}{3}
\plusb{0}{4}
\plusb{0}{5}
\bdark{0}{6}
\plusg{1}{6}
\flatG{2}{6}
\ming{3}{6}
\plusb{4}{6}
\flatB{5}{6}
\flatb{6}{6}
\flatB{7}{6}
\minb{8}{6}
\flatB{1}{5}
\flatb{2}{5}
\flatB{3}{5}
\minb{4}{5}
\blight{5}{5}
\bdark{6}{5}
\blight{7}{5}
\bdark{8}{5}
\flatB{1}{4}
\flatb{2}{4}
\flatB{3}{4}
\flatb{4}{4}
\flatB{5}{4}
\flatb{6}{4}
\flatB{7}{4}
\minb{8}{4}
\blight{1}{3}
\plusb{2}{3}
\flatB{3}{3}
\flatb{4}{3}
\flatB{5}{3}
\flatb{6}{3}
\flatB{7}{3}
\minb{8}{3}
\flatB{1}{2}
\minb{2}{2}
\plusg{3}{2}
\flatG{4}{2}
\flatg{5}{2}
\flatG{6}{2}
\ming{7}{2}
\bdark{8}{2}
\node[text centered] at (1.5,7.5) {\gr \tiny $0$};
\node[text centered] at (3.5,7.5) {\gr \tiny $1$};
\node[text centered] at (5.5,7.5) {\gr \tiny $1$};
\node[text centered] at (7.5,7.5) {\gr \tiny $1$};
\node[text centered] at (1.5,1.5) {\gr \tiny $1$};
\node[text centered] at (3.5,1.5) {\gr \tiny $1$};
\node[text centered] at (5.5,1.5) {\gr \tiny $1$};
\node[text centered] at (7.5,1.5) {\gr \tiny $0$};
\draw[ultra thick] (0.95,2) -- (0.95,7);
\node[text centered] at (4.5,0.5) {$t^9 (1-t)^4 $};
\end{tikzpicture}
+
\begin{tikzpicture}[scale=0.45,baseline=2cm]
\plusb{0}{2}
\bdark{0}{3}
\plusb{0}{4}
\plusb{0}{5}
\bdark{0}{6}
\blight{1}{6}
\plusb{2}{6}
\flatB{3}{6}
\flatb{4}{6}
\flatB{5}{6}
\flatb{6}{6}
\flatB{7}{6}
\minb{8}{6}
\flatB{1}{5}
\minb{2}{5}
\blight{3}{5}
\bdark{4}{5}
\blight{5}{5}
\plusb{6}{5}
\flatB{7}{5}
\minb{8}{5}
\flatB{1}{4}
\flatb{2}{4}
\flatB{3}{4}
\flatb{4}{4}
\flatB{5}{4}
\flatb{6}{4}
\flatB{7}{4}
\minb{8}{4}
\plusg{1}{3}
\flatG{2}{3}
\flatg{3}{3}
\flatG{4}{3}
\flatg{5}{3}
\flatG{6}{3}
\ming{7}{3}
\bdark{8}{3}
\flatB{1}{2}
\flatb{2}{2}
\flatB{3}{2}
\flatb{4}{2}
\flatB{5}{2}
\minb{6}{2}
\blight{7}{2}
\bdark{8}{2}
\node[text centered] at (1.5,7.5) {\gr \tiny $0$};
\node[text centered] at (3.5,7.5) {\gr \tiny $1$};
\node[text centered] at (5.5,7.5) {\gr \tiny $1$};
\node[text centered] at (7.5,7.5) {\gr \tiny $1$};
\node[text centered] at (1.5,1.5) {\gr \tiny $1$};
\node[text centered] at (3.5,1.5) {\gr \tiny $1$};
\node[text centered] at (5.5,1.5) {\gr \tiny $1$};
\node[text centered] at (7.5,1.5) {\gr \tiny $0$};
\draw[ultra thick] (0.95,2) -- (0.95,7);
\node[text centered] at (4.5,0.5) {$t^{10} (1-t)^3 $};
\end{tikzpicture}
\\
-
&
\begin{tikzpicture}[scale=0.45,baseline=2cm]
\plusb{0}{2}
\bdark{0}{3}
\plusb{0}{4}
\plusb{0}{5}
\bdark{0}{6}
\blight{1}{6}
\plusb{2}{6}
\flatB{3}{6}
\flatb{4}{6}
\flatB{5}{6}
\flatb{6}{6}
\flatB{7}{6}
\minb{8}{6}
\flatB{1}{5}
\minb{2}{5}
\blight{3}{5}
\plusb{4}{5}
\flatB{5}{5}
\flatb{6}{5}
\flatB{7}{5}
\minb{8}{5}
\flatB{1}{4}
\flatb{2}{4}
\flatB{3}{4}
\minb{4}{4}
\plusg{5}{4}
\flatG{6}{4}
\ming{7}{4}
\bdark{8}{4}
\plusg{1}{3}
\flatG{2}{3}
\flatg{3}{3}
\flatG{4}{3}
\ming{5}{3}
\plusb{6}{3}
\flatB{7}{3}
\minb{8}{3}
\flatB{1}{2}
\flatb{2}{2}
\flatB{3}{2}
\flatb{4}{2}
\flatB{5}{2}
\minb{6}{2}
\blight{7}{2}
\bdark{8}{2}
\node[text centered] at (1.5,7.5) {\gr \tiny $0$};
\node[text centered] at (3.5,7.5) {\gr \tiny $1$};
\node[text centered] at (5.5,7.5) {\gr \tiny $1$};
\node[text centered] at (7.5,7.5) {\gr \tiny $1$};
\node[text centered] at (1.5,1.5) {\gr \tiny $1$};
\node[text centered] at (3.5,1.5) {\gr \tiny $1$};
\node[text centered] at (5.5,1.5) {\gr \tiny $1$};
\node[text centered] at (7.5,1.5) {\gr \tiny $0$};
\draw[ultra thick] (0.95,2) -- (0.95,7);
\node[text centered] at (4.5,0.5) {$t^9 (1-t)^4(1-t^2)$};
\end{tikzpicture}
-
\begin{tikzpicture}[scale=0.45,baseline=2cm]
\plusb{0}{2}
\bdark{0}{3}
\plusb{0}{4}
\plusb{0}{5}
\bdark{0}{6}
\blight{1}{6}
\plusb{2}{6}
\flatB{3}{6}
\flatb{4}{6}
\flatB{5}{6}
\flatb{6}{6}
\flatB{7}{6}
\minb{8}{6}
\flatB{1}{5}
\minb{2}{5}
\plusg{3}{5}
\flatG{4}{5}
\ming{5}{5}
\plusb{6}{5}
\flatB{7}{5}
\minb{8}{5}
\flatB{1}{4}
\flatb{2}{4}
\flatB{3}{4}
\flatb{4}{4}
\flatB{5}{4}
\minb{6}{4}
\blight{7}{4}
\bdark{8}{4}
\plusg{1}{3}
\flatG{2}{3}
\ming{3}{3}
\plusb{4}{3}
\flatB{5}{3}
\flatb{6}{3}
\flatB{7}{3}
\minb{8}{3}
\flatB{1}{2}
\flatb{2}{2}
\flatB{3}{2}
\minb{4}{2}
\plusg{5}{2}
\flatG{6}{2}
\ming{7}{2}
\bdark{8}{2}
\node[text centered] at (1.5,7.5) {\gr \tiny $0$};
\node[text centered] at (3.5,7.5) {\gr \tiny $1$};
\node[text centered] at (5.5,7.5) {\gr \tiny $1$};
\node[text centered] at (7.5,7.5) {\gr \tiny $1$};
\node[text centered] at (1.5,1.5) {\gr \tiny $1$};
\node[text centered] at (3.5,1.5) {\gr \tiny $1$};
\node[text centered] at (5.5,1.5) {\gr \tiny $1$};
\node[text centered] at (7.5,1.5) {\gr \tiny $0$};
\draw[ultra thick] (0.95,2) -- (0.95,7);
\node[text centered] at (4.5,0.5) {$t^9 (1-t)^5 (1-t^2)$};
\end{tikzpicture}
-
\begin{tikzpicture}[scale=0.45,baseline=2cm]
\plusb{0}{2}
\bdark{0}{3}
\plusb{0}{4}
\plusb{0}{5}
\bdark{0}{6}
\plusg{1}{6}
\flatG{2}{6}
\ming{3}{6}
\plusb{4}{6}
\flatB{5}{6}
\flatb{6}{6}
\flatB{7}{6}
\minb{8}{6}
\flatB{1}{5}
\flatb{2}{5}
\flatB{3}{5}
\minb{4}{5}
\blight{5}{5}
\plusb{6}{5}
\flatB{7}{5}
\minb{8}{5}
\flatB{1}{4}
\flatb{2}{4}
\flatB{3}{4}
\flatb{4}{4}
\flatB{5}{4}
\minb{6}{4}
\blight{7}{4}
\bdark{8}{4}
\blight{1}{3}
\plusb{2}{3}
\flatB{3}{3}
\flatb{4}{3}
\flatB{5}{3}
\flatb{6}{3}
\flatB{7}{3}
\minb{8}{3}
\flatB{1}{2}
\minb{2}{2}
\plusg{3}{2}
\flatG{4}{2}
\flatg{5}{2}
\flatG{6}{2}
\ming{7}{2}
\bdark{8}{2}
\node[text centered] at (1.5,7.5) {\gr \tiny $0$};
\node[text centered] at (3.5,7.5) {\gr \tiny $1$};
\node[text centered] at (5.5,7.5) {\gr \tiny $1$};
\node[text centered] at (7.5,7.5) {\gr \tiny $1$};
\node[text centered] at (1.5,1.5) {\gr \tiny $1$};
\node[text centered] at (3.5,1.5) {\gr \tiny $1$};
\node[text centered] at (5.5,1.5) {\gr \tiny $1$};
\node[text centered] at (7.5,1.5) {\gr \tiny $0$};
\draw[ultra thick] (0.95,2) -- (0.95,7);
\node[text centered] at (4.5,0.5) {$t^9 (1-t)^5 $};
\end{tikzpicture}
-
\begin{tikzpicture}[scale=0.45,baseline=2cm]
\plusb{0}{2}
\bdark{0}{3}
\plusb{0}{4}
\plusb{0}{5}
\bdark{0}{6}
\blight{1}{6}
\plusb{2}{6}
\flatB{3}{6}
\flatb{4}{6}
\flatB{5}{6}
\flatb{6}{6}
\flatB{7}{6}
\minb{8}{6}
\flatB{1}{5}
\minb{2}{5}
\blight{3}{5}
\plusb{4}{5}
\flatB{5}{5}
\flatb{6}{5}
\flatB{7}{5}
\minb{8}{5}
\flatB{1}{4}
\flatb{2}{4}
\flatB{3}{4}
\minb{4}{4}
\blight{5}{4}
\plusb{6}{4}
\flatB{7}{4}
\minb{8}{4}
\plusg{1}{3}
\flatG{2}{3}
\flatg{3}{3}
\flatG{4}{3}
\flatg{5}{3}
\flatG{6}{3}
\ming{7}{3}
\bdark{8}{3}
\flatB{1}{2}
\flatb{2}{2}
\flatB{3}{2}
\flatb{4}{2}
\flatB{5}{2}
\minb{6}{2}
\blight{7}{2}
\bdark{8}{2}
\node[text centered] at (1.5,7.5) {\gr \tiny $0$};
\node[text centered] at (3.5,7.5) {\gr \tiny $1$};
\node[text centered] at (5.5,7.5) {\gr \tiny $1$};
\node[text centered] at (7.5,7.5) {\gr \tiny $1$};
\node[text centered] at (1.5,1.5) {\gr \tiny $1$};
\node[text centered] at (3.5,1.5) {\gr \tiny $1$};
\node[text centered] at (5.5,1.5) {\gr \tiny $1$};
\node[text centered] at (7.5,1.5) {\gr \tiny $0$};
\draw[ultra thick] (0.95,2) -- (0.95,7);
\node[text centered] at (4.5,0.5) {$t^{10} (1-t)^4 $};
\end{tikzpicture}
\end{align*}
Although this is a large sum, we have arranged the puzzles so that terms in each column combine in a neat way. In fact, one can see that the puzzles in the first three columns are ``factorized'' with respect to configurations that occur in their first three and last two rows (\textit{i.e.} there are 4 possible configurations of their top three rows, and 2 of their last two rows, for a total of 8 overall configurations). For this reason, each of these columns sum to produce a single factorized expression. The fourth column contains examples of puzzles which are not factorized in this way (\textit{i.e.} the configurations of the top three rows and the last two rows are not independent of each other). In general, we cannot expect such puzzles to combine nicely (although they do in this particular case).

Summing the Boltzmann weights in each column, we find that 
\begin{align*}
\sum_{P \in \mathbb{P}_{\nu}(\mu,\lambda)} (-1)^{L(P)} W(P) 
&= 
t^{12}(1-t)(1-t^2)+t^{12}(1-t)^2(1-t^2)+t^{12}(1-t)^2-t^{12}(1-t)^2
\\
&=
t^{12}(2-t)(1-t)(1-t^2).
\end{align*}
Finally, taking care of the normalization for this example, we obtain
\begin{align*}
f^{\b\lambda}_{\b\mu\b\nu}(t)
&=
t^{-(m+1)|\nu|}
\times
\frac{B_{\b\lambda}(t)}{B_{\b\mu}(t) b_{\b\nu}(t)}
\times
\sum_{P \in \mathbb{P}_{\nu}(\mu,\lambda)} (-1)^{L(P)} W(P) 
\\
&=
\frac{t^{-12}(1-t)^3}{(1-t)^3 (1-t)^2}
\times
t^{12}(2-t)(1-t)(1-t^2)
=
(2-t)(1+t). 
\end{align*}

\section{Proof of Theorem \ref{thm-kost}}
\label{B}

\subsection{Proof of homogeneity and degree}

The $L$ matrix \eqref{Lmat-bf} is used in the construction of $\mathcal{K}^{\lambda}_{\mu\nu}(x;t)$, which gives rise to an $x$ weight every time the right edge of a darkly-shaded tile is vacant. No $x$ weight is produced when such an edge is occupied. We abbreviate these edges by VREs and OREs (vacant and occupied right edges).

In any legal lattice configuration, the red particles give rise to 
$|\mu|-|\lambda| = |\b\nu| = (m+M-1)n - |\nu|$ OREs, while the black particles give rise to $(m+M-1)|\nu|$. This is a combined $(m+M-1)(n+|\nu|)-|\nu|$ OREs, so in any lattice configuration there are exactly 
$(m+M)(n+|\nu|)-(m+M-1)(n+|\nu|)+|\nu| = n+2|\nu|$ VREs, each with an associated $x$ weight. This establishes the homogeneity and degree of $\mathcal{K}^{\lambda}_{\mu\nu}(x;t)$. Considering that the entire right edge of the lattice is unoccupied, this produces $n+|\nu|$ trivial VREs, meaning that $\mathcal{K}^{\lambda}_{\mu\nu}(x;t)$ has an obvious common factor of $\prod_{i=1}^{n+|\nu|} x_i$.

\subsection{Multiple sum expression for \texorpdfstring{$\mathcal{K}^{\lambda}_{\mu\nu}(x;t)$}{K lambda mu nu(x;t)}}

Recall the commutation relations 
\begin{align}
\label{b-e}
\Tb(x) \Te(y) 
&=
\frac{y-tx}{y-x} \Te(y) \Tb(x)
+
\frac{(1-t)x}{x-y} \Te(x) \Tb(y),
\\
\label{b-b}
\Tb(x) \Tb(y)
&=
\Tb(y) \Tb(x),
\end{align}
from Section \ref{fer-bos-model}.

\begin{lem}{\rm 
By virtue of the commutation relations \eqref{b-e} and \eqref{b-b}, one can derive the following equation:
\begin{multline}
\label{bbb-e}
\bra{{\re\cev \mu}}
\otimes
\bra{\ }
\Tb(x_1)
\dots
\Tb(x_{k-1})
\Te(x_k)
=
\prod_{j=1}^{k-1}
\left(
\frac{x_k-tx_j}{x_k-x_j}
\right)
\bra{{\re\cev \mu}}
\otimes
\bra{\ }
\Te(x_k)
\Tb(x_1)
\dots
\Tb(x_{k-1})
\\
+
\sum_{i=1}^{k-1}
\left(
\frac{(1-t)x_i}{x_i-x_k}
\prod_{j\not=i}^{k-1}
\left(
\frac{x_i-tx_j}{x_i-x_j}
\right)
\bra{{\re\cev \mu}}
\otimes
\bra{\ }
\Te(x_i)
\Tb(x_1)
\dots
\widehat{\Tb}(x_i)
\dots
\Tb(x_k)
\right),
\end{multline}
valid for any $k \geq 1$.}
\end{lem}

\begin{proof}{\rm
The proof is almost identical to that of Lemma \ref{lem:bethe}, except that the commutation relation \eqref{b-e} differs from \eqref{g-b}, leading to a slightly different final expression.
}
\end{proof}

Next, we notice that \eqref{bbb-e} can be written more succinctly as a single sum:
\begin{multline}
\label{1-iterat2}
\bra{{\re\cev \mu}}
\otimes
\bra{\ }
\Tb(x_1)
\dots
\Tb(x_{k-1})
\Te(x_k)
=
\sum_{i=1}^{k}
\frac{\prod_{j=1}^{k-1} (x_i-tx_j)}{\prod_{j\not=i}^{k} (x_i-x_j)}
\bra{{\re\cev \mu}}
\otimes
\bra{\ }
\Te(x_i)
\Tb(x_1)
\dots
\widehat{\Tb}(x_i)
\dots
\Tb(x_k).
\end{multline}
Repeatedly iterating \eqref{1-iterat2}, it is then straightforward to show that
\begin{multline}
\label{n-iterat2}
\mathcal{K}^{\lambda}_{\mu\nu}(x;t)
=
\sum_{1 \leq i_1 \leq k_1}
\cdots
\sum_{\substack{1 \leq i_{n} \leq k_{n} \\ i_n \not=i_1,\dots,i_{n-1}}}
\\
\frac{\prod_{j_1}^{k_1-1} (x_{i_1}-tx_{j_1})}
{\prod_{j_1\not=i_1}^{k_1} (x_{i_1}-x_{j_1})}
\cdots
\frac{\prod_{j_n\not=i_1,\dots,i_{n-1}}^{k_n-1} (x_{i_n}-tx_{j_n})}
{\prod_{j_n\not=i_1,\dots,i_n}^{k_n} (x_{i_n}-x_{j_n})}
\bra{{\re\cev \mu}}
\otimes
\bra{\ }
\Te(x_{i_1})
\dots
\Te(x_{i_n})
\Tb(x_{\hat\imath_1})
\dots
\Tb(x_{\hat\imath_{|\nu|}})
\ket{{\re\cev \lambda}}
\otimes
\ket{{\cev 0}},
\end{multline}
where the summation is over distinct integers $\{i_1,\dots,i_n\}$ such that $i_j \leq k_j$ for all $1 \leq j \leq n$, and where $\{\hat\imath_1,\dots,\hat\imath_{|\nu|}\}$ denotes the complement of $\{i_1,\dots,i_n\}$ in $\{1,\dots,n+|\nu|\}$.

\subsection{Trivial action}

The following lemma is the direct analogue of Lemma \ref{lem:triv}, and can be proved in the very same way.

\begin{lem}{\rm
Let $\ket{{\re\cev \lambda}}$ be an arbitrary length-$\ell$ partition state in ${\re \mathcal{F}}$, and $\ket{{\cev 0}}$ the length-$N$ zero partition $(0,\dots,0)$ in $\mathcal{B}$ ($N$ repetitions of 0). Then
\begin{align}
\label{trivial-act2}
\Tb(y_1)
\dots
\Tb(y_N) 
\ket{{\re\cev \lambda}}
\otimes
\ket{{\cev 0}}
=
t^{\ell N}
\prod_{i=1}^{N}
y_i
\ket{{\re\cev \lambda}}
\otimes
\ket{\ }.
\end{align}

}
\end{lem}

\subsection{Multiple integral expression for \texorpdfstring{$\mathcal{K}^{\lambda}_{\mu\nu}(x;t)$}{K lambda mu nu(x;t)}}

Returning to the sum \eqref{n-iterat2}, we apply the relation \eqref{trivial-act2} in order to eliminate all $\Tb$ operators, giving
\begin{multline}
\label{something2}
\prod_{i=1}^{n+|\nu|}
\b{x}_i
\times
\mathcal{K}^{\lambda}_{\mu\nu}(x;t)
=
t^{m |\nu|}
\sum_{1 \leq i_1 \leq k_1}
\cdots
\sum_{\substack{1 \leq i_{n} \leq k_{n} \\ i_n \not=i_1,\dots,i_{n-1}}}
(\b{x}_{i_1}
\dots
\b{x}_{i_n})
\times
\\
\frac{\prod_{j_1}^{k_1-1} (x_{i_1}-t x_{j_1})}
{\prod_{j_1\not=i_1}^{k_1} (x_{i_1}-x_{j_1})}
\cdots
\frac{\prod_{j_n\not=i_1,\dots,i_{n-1}}^{k_n-1} (x_{i_n}-t x_{j_n})}
{\prod_{j_n\not=i_1,\dots,i_n}^{k_n} (x_{i_n}-x_{j_n})}
\bra{{\re\cev \mu}}
\otimes
\bra{\ }
\Te(x_{i_1})
\dots
\Te(x_{i_n})
\ket{{\re\cev \lambda}}
\otimes
\ket{\ }.
\end{multline}
By virtue of the fact that 
\begin{align*}
\bra{{\re\cev \mu}}
\otimes
\bra{\ }
\Te(x_{i_1})
\dots
\Te(x_{i_n})
\ket{{\re\cev \lambda}}
\otimes
\ket{\ }
=
\prod_{j=1}^{n} x_{i_j}^{m+M}
S_{\mu/\lambda}(\b{x}_{i_1},\dots,\b{x}_{i_n};t)
=
\prod_{j=1}^{n} x_{i_j}^{m+M}
S_{\b\lambda/\b\mu}(\b{x}_{i_1},\dots,\b{x}_{i_n};t),
\end{align*}
we can then convert the sum \eqref{something2} into a multiple integral whose integrand contains a skew $t$-Schur polynomial:
\begin{multline*}
\prod_{i=1}^{n+|\nu|}
\b{x}_i
\times
\mathcal{K}^{\lambda}_{\mu\nu}(x;t)
=
t^{m |\nu|}
\times
\\
\oint_{w_1}
\cdots
\oint_{w_n}
\prod_{1 \leq i<j \leq n}
\frac{
(w_j - w_i)
}
{
(w_j-t w_i)
}
\prod_{i=1}^{n}
\frac{
\prod_{j=1}^{k_i-1}
(w_i - t x_j)
}
{
\prod_{j=1}^{k_i}
(w_i - x_j)
}
\prod_{i=1}^{n}
w_i^{m+M-1}
S_{\b\lambda/\b\mu}(\b{w}_1,\dots,\b{w}_n;t).
\end{multline*}

\subsection{Acting with divided-difference operators}

Recalling the action \eqref{dd-act} of the divided-difference operators, we find that
\begin{multline*}
\left(
\prod_{j \not\in k(\nu)}
\Delta_j
\right)
\left(
\prod_{i=1}^{n+|\nu|}
\b{x}_i
\right)
\mathcal{K}^{\lambda}_{\mu\nu}(x;t)
=
t^{(m+1)|\nu|}
\times
\\
\oint_{w_1}
\cdots
\oint_{w_n}
\prod_{1 \leq i<j \leq n}
\frac{
(w_j - w_i)
}
{
(w_j-t w_i)
}
\prod_{i=1}^{n}
\frac{
\prod_{1\leq j \leq k_{i-1}}
(w_i - t x_j)
}
{
\prod_{1 \leq j \leq k_i}
(w_i - x_j)
}
\prod_{i=1}^{n}
w_i^{m+M-1}
S_{\b\lambda/\b\mu}(\b{w}_1,\dots,\b{w}_n;t),
\end{multline*}
where we have defined $k_0 = 0$.

\subsection{Homogeneous limit}

Taking the limit $x_i \rightarrow 0$ for all $1 \leq i \leq n+|\nu|$, we obtain
\begin{align*}
t^{-(m+1)|\nu|}
&\times
\lim_{x \rightarrow 0}
\left(
\prod_{j \not\in k(\nu)}
\Delta_j
\right)
\left(
\prod_{i=1}^{n+|\nu|}
\b{x}_i
\right)
\mathcal{K}^{\lambda}_{\mu\nu}(x;t)
\\
&=
\oint_{w_1}
\cdots
\oint_{w_n}
\prod_{1 \leq i<j \leq n}
\frac{
(w_j - w_i)
}
{
(w_j-t w_i)
}
\prod_{i=1}^{n}
w_i^{m+M-1+k_{i-1}-k_i}
S_{\b\lambda/\b\mu}(\b{w}_1,\dots,\b{w}_n;t)
\\
&=
{\rm Coeff}
\left[
\prod_{1 \leq i<j \leq n}
\frac{(w_j-w_i)}{(w_j-tw_i)}
S_{\b\lambda/\b\mu}(\b{w}_1,\dots,\b{w}_n;t),
\prod_{i=1}^{n}
w_i^{k_i-k_{i-1}-m-M}
\right]
\\
&=
{\rm Coeff}
\left[
\prod_{1 \leq i<j \leq n}
\frac{(z_i-z_j)}{(z_i-tz_j)}
S_{\b\lambda/\b\mu}(z_1,\dots,z_n;t),
\prod_{i=1}^{n}
z_i^{m+M+k_{i-1}-k_i}
\right].
\end{align*}
Finally, since $m+M+k_{i-1}-k_i = m + M - 1 - \nu_{n-i+1} = \b\nu_i$, the final expression is equal to $b_{\b\nu}(t)\b{K}^{\b\lambda}_{\b\mu\b\nu}(t)$.

\section{Example of Theorem \ref{thm-kost-puzzle}}
\label{app:c}

Fix three partitions $\b\lambda=(1,1,1),\b\mu=(0,0,0),\b\nu=(2,1)$. $\b\lambda,\b\mu$ are the complements of $\lambda=(0,0,0),\mu=(1,1,1)$ by 1, while $\b\nu$ is the complement of $\nu=(2,1)$ by 3 (see the statement of Theorem \ref{thm-kost} for the details on taking complements). Our puzzles therefore have the frame
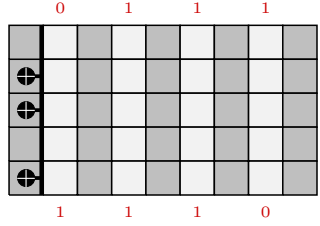
\begin{align*}
\begin{tikzpicture}[scale=0.45]
\plusb{0}{2}
\bdark{0}{3}
\plusb{0}{4}
\plusb{0}{5}
\bdark{0}{6}
\blight{1}{6}
\bdark{2}{6}
\blight{3}{6}
\bdark{4}{6}
\blight{5}{6}
\bdark{6}{6}
\blight{7}{6}
\bdark{8}{6}
\blight{1}{5}
\bdark{2}{5}
\blight{3}{5}
\bdark{4}{5}
\blight{5}{5}
\bdark{6}{5}
\blight{7}{5}
\bdark{8}{5}
\blight{1}{4}
\bdark{2}{4}
\blight{3}{4}
\bdark{4}{4}
\blight{5}{4}
\bdark{6}{4}
\blight{7}{4}
\bdark{8}{4}
\blight{1}{3}
\bdark{2}{3}
\blight{3}{3}
\bdark{4}{3}
\blight{5}{3}
\bdark{6}{3}
\blight{7}{3}
\bdark{8}{3}
\blight{1}{2}
\bdark{2}{2}
\blight{3}{2}
\bdark{4}{2}
\blight{5}{2}
\bdark{6}{2}
\blight{7}{2}
\bdark{8}{2}
\node[text centered] at (1.5,7.5) {\re \tiny $0$};
\node[text centered] at (3.5,7.5) {\re \tiny $1$};
\node[text centered] at (5.5,7.5) {\re \tiny $1$};
\node[text centered] at (7.5,7.5) {\re \tiny $1$};
\node[text centered] at (1.5,1.5) {\re \tiny $1$};
\node[text centered] at (3.5,1.5) {\re \tiny $1$};
\node[text centered] at (5.5,1.5) {\re \tiny $1$};
\node[text centered] at (7.5,1.5) {\re \tiny $0$};
\draw[ultra thick] (0.95,2) -- (0.95,7);
\end{tikzpicture}
\end{align*}
A total of 12 configurations are possible, leading to the following sum:
\begin{align*}
\begin{tikzpicture}[scale=0.45,baseline=2cm]
\plusb{0}{2}
\bdark{0}{3}
\plusb{0}{4}
\plusb{0}{5}
\bdark{0}{6}
\plusr{1}{6}
\flatR{2}{6}
\minr{3}{6}
\bdark{4}{6}
\blight{5}{6}
\bdark{6}{6}
\blight{7}{6}
\bdark{8}{6}
\flatB{1}{5}
\flatb{2}{5}
\flatB{3}{5}
\flatb{4}{5}
\flatB{5}{5}
\flatb{6}{5}
\flatB{7}{5}
\minb{8}{5}
\flatB{1}{4}
\flatb{2}{4}
\flatB{3}{4}
\flatb{4}{4}
\flatB{5}{4}
\flatb{6}{4}
\flatB{7}{4}
\minb{8}{4}
\blight{1}{3}
\bdark{2}{3}
\plusr{3}{3}
\flatR{4}{3}
\flatr{5}{3}
\flatR{6}{3}
\minr{7}{3}
\bdark{8}{3}
\flatB{1}{2}
\flatb{2}{2}
\flatB{3}{2}
\flatb{4}{2}
\flatB{5}{2}
\flatb{6}{2}
\flatB{7}{2}
\minb{8}{2}
\node[text centered] at (1.5,7.5) {\re \tiny $0$};
\node[text centered] at (3.5,7.5) {\re \tiny $1$};
\node[text centered] at (5.5,7.5) {\re \tiny $1$};
\node[text centered] at (7.5,7.5) {\re \tiny $1$};
\node[text centered] at (1.5,1.5) {\re \tiny $1$};
\node[text centered] at (3.5,1.5) {\re \tiny $1$};
\node[text centered] at (5.5,1.5) {\re \tiny $1$};
\node[text centered] at (7.5,1.5) {\re \tiny $0$};
\draw[ultra thick] (0.95,2) -- (0.95,7);
\node[text centered] at (4.5,0.5) {$-t^{10} (1-t)^2 $};
\end{tikzpicture}
-
\begin{tikzpicture}[scale=0.45,baseline=2cm]
\plusb{0}{2}
\bdark{0}{3}
\plusb{0}{4}
\plusb{0}{5}
\bdark{0}{6}
\plusr{1}{6}
\flatR{2}{6}
\minr{3}{6}
\bdark{4}{6}
\blight{5}{6}
\bdark{6}{6}
\blight{7}{6}
\bdark{8}{6}
\flatB{1}{5}
\flatb{2}{5}
\flatB{3}{5}
\flatb{4}{5}
\flatB{5}{5}
\flatb{6}{5}
\flatB{7}{5}
\minb{8}{5}
\flatB{1}{4}
\flatb{2}{4}
\flatB{3}{4}
\flatb{4}{4}
\flatB{5}{4}
\flatb{6}{4}
\flatB{7}{4}
\minb{8}{4}
\blight{1}{3}
\plusb{2}{3}
\flatB{3}{3}
\flatb{4}{3}
\flatB{5}{3}
\flatb{6}{3}
\flatB{7}{3}
\minb{8}{3}
\flatB{1}{2}
\minb{2}{2}
\plusr{3}{2}
\flatR{4}{2}
\flatr{5}{2}
\flatR{6}{2}
\minr{7}{2}
\bdark{8}{2}
\node[text centered] at (1.5,7.5) {\re \tiny $0$};
\node[text centered] at (3.5,7.5) {\re \tiny $1$};
\node[text centered] at (5.5,7.5) {\re \tiny $1$};
\node[text centered] at (7.5,7.5) {\re \tiny $1$};
\node[text centered] at (1.5,1.5) {\re \tiny $1$};
\node[text centered] at (3.5,1.5) {\re \tiny $1$};
\node[text centered] at (5.5,1.5) {\re \tiny $1$};
\node[text centered] at (7.5,1.5) {\re \tiny $0$};
\draw[ultra thick] (0.95,2) -- (0.95,7);
\node[text centered] at (4.5,0.5) {$-t^{10} (1-t)^3 $};
\end{tikzpicture}
-
\begin{tikzpicture}[scale=0.45,baseline=2cm]
\plusb{0}{2}
\bdark{0}{3}
\plusb{0}{4}
\plusb{0}{5}
\bdark{0}{6}
\blight{1}{6}
\bdark{2}{6}
\blight{3}{6}
\bdark{4}{6}
\blight{5}{6}
\plusb{6}{6}
\flatB{7}{6}
\minb{8}{6}
\flatB{1}{5}
\flatb{2}{5}
\flatB{3}{5}
\flatb{4}{5}
\flatB{5}{5}
\flatb{6}{5}
\flatB{7}{5}
\minb{8}{5}
\flatB{1}{4}
\flatb{2}{4}
\flatB{3}{4}
\flatb{4}{4}
\flatB{5}{4}
\flatb{6}{4}
\flatB{7}{4}
\minb{8}{4}
\plusr{1}{3}
\flatR{2}{3}
\flatr{3}{3}
\flatR{4}{3}
\flatr{5}{3}
\flatR{6}{3}
\minr{7}{3}
\bdark{8}{3}
\flatB{1}{2}
\flatb{2}{2}
\flatB{3}{2}
\flatb{4}{2}
\flatB{5}{2}
\minb{6}{2}
\blight{7}{2}
\bdark{8}{2}
\node[text centered] at (1.5,7.5) {\re \tiny $0$};
\node[text centered] at (3.5,7.5) {\re \tiny $1$};
\node[text centered] at (5.5,7.5) {\re \tiny $1$};
\node[text centered] at (7.5,7.5) {\re \tiny $1$};
\node[text centered] at (1.5,1.5) {\re \tiny $1$};
\node[text centered] at (3.5,1.5) {\re \tiny $1$};
\node[text centered] at (5.5,1.5) {\re \tiny $1$};
\node[text centered] at (7.5,1.5) {\re \tiny $0$};
\draw[ultra thick] (0.95,2) -- (0.95,7);
\node[text centered] at (4.5,0.5) {$t^{12} (1-t)^2 $};
\end{tikzpicture}
\\
-
\begin{tikzpicture}[scale=0.45,baseline=2cm]
\plusb{0}{2}
\bdark{0}{3}
\plusb{0}{4}
\plusb{0}{5}
\bdark{0}{6}
\plusr{1}{6}
\flatR{2}{6}
\minr{3}{6}
\bdark{4}{6}
\blight{5}{6}
\plusb{6}{6}
\flatB{7}{6}
\minb{8}{6}
\flatB{1}{5}
\flatb{2}{5}
\flatB{3}{5}
\flatb{4}{5}
\flatB{5}{5}
\minb{6}{5}
\blight{7}{5}
\bdark{8}{5}
\flatB{1}{4}
\flatb{2}{4}
\flatB{3}{4}
\flatb{4}{4}
\flatB{5}{4}
\flatb{6}{4}
\flatB{7}{4}
\minb{8}{4}
\blight{1}{3}
\bdark{2}{3}
\plusr{3}{3}
\flatR{4}{3}
\flatr{5}{3}
\flatR{6}{3}
\minr{7}{3}
\bdark{8}{3}
\flatB{1}{2}
\flatb{2}{2}
\flatB{3}{2}
\flatb{4}{2}
\flatB{5}{2}
\flatb{6}{2}
\flatB{7}{2}
\minb{8}{2}
\node[text centered] at (1.5,7.5) {\re \tiny $0$};
\node[text centered] at (3.5,7.5) {\re \tiny $1$};
\node[text centered] at (5.5,7.5) {\re \tiny $1$};
\node[text centered] at (7.5,7.5) {\re \tiny $1$};
\node[text centered] at (1.5,1.5) {\re \tiny $1$};
\node[text centered] at (3.5,1.5) {\re \tiny $1$};
\node[text centered] at (5.5,1.5) {\re \tiny $1$};
\node[text centered] at (7.5,1.5) {\re \tiny $0$};
\draw[ultra thick] (0.95,2) -- (0.95,7);
\node[text centered] at (4.5,0.5) {$-t^{10} (1-t)^3 $};
\end{tikzpicture}
+
\begin{tikzpicture}[scale=0.45,baseline=2cm]
\plusb{0}{2}
\bdark{0}{3}
\plusb{0}{4}
\plusb{0}{5}
\bdark{0}{6}
\plusr{1}{6}
\flatR{2}{6}
\minr{3}{6}
\bdark{4}{6}
\blight{5}{6}
\plusb{6}{6}
\flatB{7}{6}
\minb{8}{6}
\flatB{1}{5}
\flatb{2}{5}
\flatB{3}{5}
\flatb{4}{5}
\flatB{5}{5}
\minb{6}{5}
\blight{7}{5}
\bdark{8}{5}
\flatB{1}{4}
\flatb{2}{4}
\flatB{3}{4}
\flatb{4}{4}
\flatB{5}{4}
\flatb{6}{4}
\flatB{7}{4}
\minb{8}{4}
\blight{1}{3}
\plusb{2}{3}
\flatB{3}{3}
\flatb{4}{3}
\flatB{5}{3}
\flatb{6}{3}
\flatB{7}{3}
\minb{8}{3}
\flatB{1}{2}
\minb{2}{2}
\plusr{3}{2}
\flatR{4}{2}
\flatr{5}{2}
\flatR{6}{2}
\minr{7}{2}
\bdark{8}{2}
\node[text centered] at (1.5,7.5) {\re \tiny $0$};
\node[text centered] at (3.5,7.5) {\re \tiny $1$};
\node[text centered] at (5.5,7.5) {\re \tiny $1$};
\node[text centered] at (7.5,7.5) {\re \tiny $1$};
\node[text centered] at (1.5,1.5) {\re \tiny $1$};
\node[text centered] at (3.5,1.5) {\re \tiny $1$};
\node[text centered] at (5.5,1.5) {\re \tiny $1$};
\node[text centered] at (7.5,1.5) {\re \tiny $0$};
\draw[ultra thick] (0.95,2) -- (0.95,7);
\node[text centered] at (4.5,0.5) {$-t^{10} (1-t)^4 $};
\end{tikzpicture}
+
\begin{tikzpicture}[scale=0.45,baseline=2cm]
\plusb{0}{2}
\bdark{0}{3}
\plusb{0}{4}
\plusb{0}{5}
\bdark{0}{6}
\blight{1}{6}
\bdark{2}{6}
\blight{3}{6}
\plusb{4}{6}
\flatB{5}{6}
\flatb{6}{6}
\flatB{7}{6}
\minb{8}{6}
\flatB{1}{5}
\flatb{2}{5}
\flatB{3}{5}
\minb{4}{5}
\blight{5}{5}
\plusb{6}{5}
\flatB{7}{5}
\minb{8}{5}
\flatB{1}{4}
\flatb{2}{4}
\flatB{3}{4}
\flatb{4}{4}
\flatB{5}{4}
\flatb{6}{4}
\flatB{7}{4}
\minb{8}{4}
\plusr{1}{3}
\flatR{2}{3}
\flatr{3}{3}
\flatR{4}{3}
\flatr{5}{3}
\flatR{6}{3}
\minr{7}{3}
\bdark{8}{3}
\flatB{1}{2}
\flatb{2}{2}
\flatB{3}{2}
\flatb{4}{2}
\flatB{5}{2}
\minb{6}{2}
\blight{7}{2}
\bdark{8}{2}
\node[text centered] at (1.5,7.5) {\re \tiny $0$};
\node[text centered] at (3.5,7.5) {\re \tiny $1$};
\node[text centered] at (5.5,7.5) {\re \tiny $1$};
\node[text centered] at (7.5,7.5) {\re \tiny $1$};
\node[text centered] at (1.5,1.5) {\re \tiny $1$};
\node[text centered] at (3.5,1.5) {\re \tiny $1$};
\node[text centered] at (5.5,1.5) {\re \tiny $1$};
\node[text centered] at (7.5,1.5) {\re \tiny $0$};
\draw[ultra thick] (0.95,2) -- (0.95,7);
\node[text centered] at (4.5,0.5) {$t^{12} (1-t)^3 $};
\end{tikzpicture}
\\
-
\begin{tikzpicture}[scale=0.45,baseline=2cm]
\plusb{0}{2}
\bdark{0}{3}
\plusb{0}{4}
\plusb{0}{5}
\bdark{0}{6}
\plusr{1}{6}
\flatR{2}{6}
\minr{3}{6}
\plusb{4}{6}
\flatB{5}{6}
\flatb{6}{6}
\flatB{7}{6}
\minb{8}{6}
\flatB{1}{5}
\flatb{2}{5}
\flatB{3}{5}
\minb{4}{5}
\blight{5}{5}
\bdark{6}{5}
\blight{7}{5}
\bdark{8}{5}
\flatB{1}{4}
\flatb{2}{4}
\flatB{3}{4}
\flatb{4}{4}
\flatB{5}{4}
\flatb{6}{4}
\flatB{7}{4}
\minb{8}{4}
\blight{1}{3}
\bdark{2}{3}
\plusr{3}{3}
\flatR{4}{3}
\flatr{5}{3}
\flatR{6}{3}
\minr{7}{3}
\bdark{8}{3}
\flatB{1}{2}
\flatb{2}{2}
\flatB{3}{2}
\flatb{4}{2}
\flatB{5}{2}
\flatb{6}{2}
\flatB{7}{2}
\minb{8}{2}
\node[text centered] at (1.5,7.5) {\re \tiny $0$};
\node[text centered] at (3.5,7.5) {\re \tiny $1$};
\node[text centered] at (5.5,7.5) {\re \tiny $1$};
\node[text centered] at (7.5,7.5) {\re \tiny $1$};
\node[text centered] at (1.5,1.5) {\re \tiny $1$};
\node[text centered] at (3.5,1.5) {\re \tiny $1$};
\node[text centered] at (5.5,1.5) {\re \tiny $1$};
\node[text centered] at (7.5,1.5) {\re \tiny $0$};
\draw[ultra thick] (0.95,2) -- (0.95,7);
\node[text centered] at (4.5,0.5) {$-t^{10} (1-t)^3 $};
\end{tikzpicture}
+
\begin{tikzpicture}[scale=0.45,baseline=2cm]
\plusb{0}{2}
\bdark{0}{3}
\plusb{0}{4}
\plusb{0}{5}
\bdark{0}{6}
\plusr{1}{6}
\flatR{2}{6}
\minr{3}{6}
\plusb{4}{6}
\flatB{5}{6}
\flatb{6}{6}
\flatB{7}{6}
\minb{8}{6}
\flatB{1}{5}
\flatb{2}{5}
\flatB{3}{5}
\minb{4}{5}
\blight{5}{5}
\bdark{6}{5}
\blight{7}{5}
\bdark{8}{5}
\flatB{1}{4}
\flatb{2}{4}
\flatB{3}{4}
\flatb{4}{4}
\flatB{5}{4}
\flatb{6}{4}
\flatB{7}{4}
\minb{8}{4}
\blight{1}{3}
\plusb{2}{3}
\flatB{3}{3}
\flatb{4}{3}
\flatB{5}{3}
\flatb{6}{3}
\flatB{7}{3}
\minb{8}{3}
\flatB{1}{2}
\minb{2}{2}
\plusr{3}{2}
\flatR{4}{2}
\flatr{5}{2}
\flatR{6}{2}
\minr{7}{2}
\bdark{8}{2}
\node[text centered] at (1.5,7.5) {\re \tiny $0$};
\node[text centered] at (3.5,7.5) {\re \tiny $1$};
\node[text centered] at (5.5,7.5) {\re \tiny $1$};
\node[text centered] at (7.5,7.5) {\re \tiny $1$};
\node[text centered] at (1.5,1.5) {\re \tiny $1$};
\node[text centered] at (3.5,1.5) {\re \tiny $1$};
\node[text centered] at (5.5,1.5) {\re \tiny $1$};
\node[text centered] at (7.5,1.5) {\re \tiny $0$};
\draw[ultra thick] (0.95,2) -- (0.95,7);
\node[text centered] at (4.5,0.5) {$-t^{10} (1-t)^4 $};
\end{tikzpicture}
+
\begin{tikzpicture}[scale=0.45,baseline=2cm]
\plusb{0}{2}
\bdark{0}{3}
\plusb{0}{4}
\plusb{0}{5}
\bdark{0}{6}
\blight{1}{6}
\plusb{2}{6}
\flatB{3}{6}
\flatb{4}{6}
\flatB{5}{6}
\flatb{6}{6}
\flatB{7}{6}
\minb{8}{6}
\flatB{1}{5}
\minb{2}{5}
\blight{3}{5}
\bdark{4}{5}
\blight{5}{5}
\plusb{6}{5}
\flatB{7}{5}
\minb{8}{5}
\flatB{1}{4}
\flatb{2}{4}
\flatB{3}{4}
\flatb{4}{4}
\flatB{5}{4}
\flatb{6}{4}
\flatB{7}{4}
\minb{8}{4}
\plusr{1}{3}
\flatR{2}{3}
\flatr{3}{3}
\flatR{4}{3}
\flatr{5}{3}
\flatR{6}{3}
\minr{7}{3}
\bdark{8}{3}
\flatB{1}{2}
\flatb{2}{2}
\flatB{3}{2}
\flatb{4}{2}
\flatB{5}{2}
\minb{6}{2}
\blight{7}{2}
\bdark{8}{2}
\node[text centered] at (1.5,7.5) {\re \tiny $0$};
\node[text centered] at (3.5,7.5) {\re \tiny $1$};
\node[text centered] at (5.5,7.5) {\re \tiny $1$};
\node[text centered] at (7.5,7.5) {\re \tiny $1$};
\node[text centered] at (1.5,1.5) {\re \tiny $1$};
\node[text centered] at (3.5,1.5) {\re \tiny $1$};
\node[text centered] at (5.5,1.5) {\re \tiny $1$};
\node[text centered] at (7.5,1.5) {\re \tiny $0$};
\draw[ultra thick] (0.95,2) -- (0.95,7);
\node[text centered] at (4.5,0.5) {$t^{12} (1-t)^3 $};
\end{tikzpicture}
\\
+
\begin{tikzpicture}[scale=0.45,baseline=2cm]
\plusb{0}{2}
\bdark{0}{3}
\plusb{0}{4}
\plusb{0}{5}
\bdark{0}{6}
\plusr{1}{6}
\flatR{2}{6}
\minr{3}{6}
\plusb{4}{6}
\flatB{5}{6}
\flatb{6}{6}
\flatB{7}{6}
\minb{8}{6}
\flatB{1}{5}
\flatb{2}{5}
\flatB{3}{5}
\minb{4}{5}
\blight{5}{5}
\plusb{6}{5}
\flatB{7}{5}
\minb{8}{5}
\flatB{1}{4}
\flatb{2}{4}
\flatB{3}{4}
\flatb{4}{4}
\flatB{5}{4}
\minb{6}{4}
\blight{7}{4}
\bdark{8}{4}
\blight{1}{3}
\bdark{2}{3}
\plusr{3}{3}
\flatR{4}{3}
\flatr{5}{3}
\flatR{6}{3}
\minr{7}{3}
\bdark{8}{3}
\flatB{1}{2}
\flatb{2}{2}
\flatB{3}{2}
\flatb{4}{2}
\flatB{5}{2}
\flatb{6}{2}
\flatB{7}{2}
\minb{8}{2}
\node[text centered] at (1.5,7.5) {\re \tiny $0$};
\node[text centered] at (3.5,7.5) {\re \tiny $1$};
\node[text centered] at (5.5,7.5) {\re \tiny $1$};
\node[text centered] at (7.5,7.5) {\re \tiny $1$};
\node[text centered] at (1.5,1.5) {\re \tiny $1$};
\node[text centered] at (3.5,1.5) {\re \tiny $1$};
\node[text centered] at (5.5,1.5) {\re \tiny $1$};
\node[text centered] at (7.5,1.5) {\re \tiny $0$};
\draw[ultra thick] (0.95,2) -- (0.95,7);
\node[text centered] at (4.5,0.5) {$-t^{10} (1-t)^4 $};
\end{tikzpicture}
-
\begin{tikzpicture}[scale=0.45,baseline=2cm]
\plusb{0}{2}
\bdark{0}{3}
\plusb{0}{4}
\plusb{0}{5}
\bdark{0}{6}
\plusr{1}{6}
\flatR{2}{6}
\minr{3}{6}
\plusb{4}{6}
\flatB{5}{6}
\flatb{6}{6}
\flatB{7}{6}
\minb{8}{6}
\flatB{1}{5}
\flatb{2}{5}
\flatB{3}{5}
\minb{4}{5}
\blight{5}{5}
\plusb{6}{5}
\flatB{7}{5}
\minb{8}{5}
\flatB{1}{4}
\flatb{2}{4}
\flatB{3}{4}
\flatb{4}{4}
\flatB{5}{4}
\minb{6}{4}
\blight{7}{4}
\bdark{8}{4}
\blight{1}{3}
\plusb{2}{3}
\flatB{3}{3}
\flatb{4}{3}
\flatB{5}{3}
\flatb{6}{3}
\flatB{7}{3}
\minb{8}{3}
\flatB{1}{2}
\minb{2}{2}
\plusr{3}{2}
\flatR{4}{2}
\flatr{5}{2}
\flatR{6}{2}
\minr{7}{2}
\bdark{8}{2}
\node[text centered] at (1.5,7.5) {\re \tiny $0$};
\node[text centered] at (3.5,7.5) {\re \tiny $1$};
\node[text centered] at (5.5,7.5) {\re \tiny $1$};
\node[text centered] at (7.5,7.5) {\re \tiny $1$};
\node[text centered] at (1.5,1.5) {\re \tiny $1$};
\node[text centered] at (3.5,1.5) {\re \tiny $1$};
\node[text centered] at (5.5,1.5) {\re \tiny $1$};
\node[text centered] at (7.5,1.5) {\re \tiny $0$};
\draw[ultra thick] (0.95,2) -- (0.95,7);
\node[text centered] at (4.5,0.5) {$-t^{10} (1-t)^5 $};
\end{tikzpicture}
-
\begin{tikzpicture}[scale=0.45,baseline=2cm]
\plusb{0}{2}
\bdark{0}{3}
\plusb{0}{4}
\plusb{0}{5}
\bdark{0}{6}
\blight{1}{6}
\plusb{2}{6}
\flatB{3}{6}
\flatb{4}{6}
\flatB{5}{6}
\flatb{6}{6}
\flatB{7}{6}
\minb{8}{6}
\flatB{1}{5}
\minb{2}{5}
\blight{3}{5}
\plusb{4}{5}
\flatB{5}{5}
\flatb{6}{5}
\flatB{7}{5}
\minb{8}{5}
\flatB{1}{4}
\flatb{2}{4}
\flatB{3}{4}
\minb{4}{4}
\blight{5}{4}
\plusb{6}{4}
\flatB{7}{4}
\minb{8}{4}
\plusr{1}{3}
\flatR{2}{3}
\flatr{3}{3}
\flatR{4}{3}
\flatr{5}{3}
\flatR{6}{3}
\minr{7}{3}
\bdark{8}{3}
\flatB{1}{2}
\flatb{2}{2}
\flatB{3}{2}
\flatb{4}{2}
\flatB{5}{2}
\minb{6}{2}
\blight{7}{2}
\bdark{8}{2}
\node[text centered] at (1.5,7.5) {\re \tiny $0$};
\node[text centered] at (3.5,7.5) {\re \tiny $1$};
\node[text centered] at (5.5,7.5) {\re \tiny $1$};
\node[text centered] at (7.5,7.5) {\re \tiny $1$};
\node[text centered] at (1.5,1.5) {\re \tiny $1$};
\node[text centered] at (3.5,1.5) {\re \tiny $1$};
\node[text centered] at (5.5,1.5) {\re \tiny $1$};
\node[text centered] at (7.5,1.5) {\re \tiny $0$};
\draw[ultra thick] (0.95,2) -- (0.95,7);
\node[text centered] at (4.5,0.5) {$t^{12} (1-t)^4 $};
\end{tikzpicture}
\end{align*}
Summing the resulting Boltzmann weights, we find that $\sum_{P \in \mathbb{P}_{\nu}(\mu,\lambda)} (-1)^{L(P)} W(P) = -t^{13} (1-t)^2 (1+t)$. The inverse Kostka polynomial is obtained by normalizing this sum:
\begin{align*}
\b{K}^{\b\lambda}_{\b\mu\b\nu}(t)
\equiv
\b{K}^{\b\lambda}_{\b\nu}(t)
&=
\frac{t^{-(m+1)|\nu|}}{b_{\b\nu}(t)}
\times
\sum_{P \in \mathbb{P}_{\nu}(\mu,\lambda)} (-1)^{L(P)} W(P) 
\\
&=
\frac{t^{-12}}{(1-t)^2}
\times
-t^{13} (1-t)^2 (1+t)
=
-t(1+t). 
\end{align*}

\bibliographystyle{abbrv}
\bibliography{references}

\end{document}